\documentclass[12pt]{article}
 \usepackage{a4wide}
\usepackage{authblk}

\usepackage[breaklinks]{hyperref}
\usepackage{booktabs}

\usepackage{todonotes,soul}

\usepackage{mathptmx}       % selects Times Roman as basic font
\usepackage{graphicx}        % standard LaTeX graphics tool
\usepackage{multicol}        % used for the two-column index
\usepackage[bottom]{footmisc}% places footnotes at page bottom

\usepackage[utf8]{inputenc}
\usepackage{amsmath,amssymb,amsthm}%,bbm,phystex} 
\usepackage{leftidx,extarrows}
\usepackage{tikz,tikz-3dplot}
\usetikzlibrary{matrix%
}%,arrows,shapes,decorations.pathmorphing}

\usepackage[mathscr]{eucal}
\usepackage{mathrsfs}

\newcommand{\smartqed}{}

\DeclareMathAlphabet{\mathcal}{OMS}{cmsy}{m}{n}

\makeatletter
\g@addto@macro\bfseries{\boldmath}
\makeatother

\DeclareMathAlphabet\BEuScript{U}{eus}{b}{n}

\newcommand{\CC}{{\mathbb C}}
\newcommand{\KK}{{\mathbb K}}
\newcommand{\NN}{{\mathbb N}}
\newcommand{\RR}{{\mathbb R}}
\newcommand{\TT}{{\mathbb T}}
\newcommand{\ZZ}{{\mathbb Z}}

\usepackage{bm}
\DeclareMathAlphabet{\mathbfsf}{\encodingdefault}{\sfdefault}{bx}{n}

\DeclareBoldMathCommand\Db{D}
\DeclareBoldMathCommand\Fb{F}
\DeclareBoldMathCommand\Ib{I}
\DeclareBoldMathCommand\Lb{L}
\DeclareBoldMathCommand\Mb{M}
\DeclareBoldMathCommand\Nb{N}
\DeclareBoldMathCommand\Pb{P}
\DeclareBoldMathCommand\Ob{O}
\DeclareBoldMathCommand\rb{R}
\DeclareBoldMathCommand\ab{a}
\DeclareBoldMathCommand\bb{b}
\DeclareBoldMathCommand\cb{c}
\DeclareBoldMathCommand\eb{e}
\DeclareBoldMathCommand\ib{i}
\DeclareBoldMathCommand\jb{j}
\DeclareBoldMathCommand\kb{k}
\DeclareBoldMathCommand\pb{p}
\DeclareBoldMathCommand\rb{r}
\DeclareBoldMathCommand\ub{u}
\DeclareBoldMathCommand\vb{v}
\DeclareBoldMathCommand\xb{x}

\newcommand{\dvol}{{\rm dvol}}

\newcommand{\id}{\text{id}}

\newcommand{\Reals}{\mathbb{R}}

\newcommand{\DD}{\mathscr{D}}

\newcommand{\HH}{\mathscr{H}}

\newcommand{\Ac}{{\mathcal{A}}}

\newcommand{\Nc}{{\mathcal{N}}}
\newcommand{\Rc}{{\mathcal{R}}}
\newcommand{\OO}{{\mathcal{O}}}

\newcommand{\CoinX}[1]{C_0^\infty({#1})}

\newcommand{\II}{\leavevmode\hbox{\rm{\small1\kern-3.8pt\normalsize1}}}

\newcommand{\supp}{\textrm{supp}\,}

\newcommand{\WF}{\textrm{WF}\,}

\DeclareMathOperator{\Aut}{Aut}
\DeclareMathOperator{\End}{End}

\DeclareMathOperator{\Fld}{Fld}
\newcommand{\Fldlin}{{\Fld}_{\text{lin}}}

\DeclareMathOperator{\Fol}{Fol}

\newcommand{\Ct}{\mathbfsf{ C}}

\newcommand{\Loc}{\mathbfsf{ Loc}}

\newcommand{\SpLoc}{\mathbfsf{ SpinLoc}}
\newcommand{\Mand}{\Loc}

\newcommand{\Set}{\mathbfsf{ Set}}

\newcommand{\Alg}{\mathbfsf{ Alg}}
\newcommand{\CAlg}{\mathbfsf{ C^*\hbox{-}Alg}}

\newcommand{\Vect}{\mathbfsf{ Vec}}

\newcommand{\Phys}{\mathbfsf{ Phys}}
\newcommand{\Stsp}{\mathbfsf{ Stsp}}

\newcommand{\Af}{{\mathscr A}}

\newcommand{\Bf}{{\mathscr B}}

\newcommand{\Cf}{{\mathscr C}}
\newcommand{\Df}{{\mathscr D}}

\newcommand{\Ff}{{\mathscr F}}

\newcommand{\If}{{\mathscr I}}

\newcommand{\Mf}{{\mathscr M}}

\newcommand{\Sf}{{\mathscr S}}

\newcommand{\Uf}{{\mathscr U}}
\newcommand{\Vf}{{\mathscr V}}
\newcommand{\Dc}{{\mathcal D}}
\newcommand{\Hc}{{\mathcal H}}
\newcommand{\Mc}{{\mathcal M}}

\newcommand{\Lc}{{\mathcal{L}}}

\newcommand{\nto}{\stackrel{\cdot}{\to}}

\newcommand{\ogth}{{\mathfrak o}}
\newcommand{\tgth}{{\mathfrak t}}

\newcommand{\wgth}{{\mathfrak w}}

\newcommand{\Tgth}{{\mathfrak T}}

\newcommand{\kin}{{\text{kin}}}
\newcommand{\dyn}{{\text{dyn}}}

\newcommand{\fb}{{\boldsymbol{f}}}

\newcommand{\hb}{{\boldsymbol{h}}}

\newcommand{\Ts}{\mathbfsf{ T}}

\newcommand{\rce}{{\rm rce}}

\newcommand{\sS}{\leftidx{_*}{}{S}}
\newcommand{\sT}{\leftidx{_*}{}{T}}

\newcommand{\SL}{\text{SL}}

\newtheorem{theorem}{Theorem}[section]
\newtheorem{proposition}[theorem]{Proposition}
\newtheorem{definition}[theorem]{Definition}
\newtheorem{lemma}[theorem]{Lemma}
\newtheorem{corollary}[theorem]{Corollary}
\newtheorem{remark}[theorem]{Remark}

\newtheorem{assumption}{Assumption}[section]

\makeindex             % used for the subject index
\begin{document}

\title{Algebraic quantum field theory in curved spacetimes}
%\titlerunning{AQFT in curved spacetimes} 

\author[1]{Christopher J. Fewster\thanks{\tt chris.fewster@york.ac.uk}}
\author[2]{Rainer Verch\thanks{\tt rainer.verch@uni-leipzig.de}}
%\author{Christopher J. Fewster\thanks{\tt chris.fewster@york.ac.uk} and Rainer Verch\thanks{rainer.verch@uni-leipzig.de}}
%\institute{Christopher J. Fewster  \at 
\affil[1]{Department of Mathematics,
University of York, Heslington, York YO10 5DD, United Kingdom.}
\affil[2]{Institut f\"ur Theoretische Physik, 
Universit\"at Leipzig, Br\"uderstr.\ 16, 04103 Leipzig, Germany.} 
\date{\today}
\maketitle

%\abstract*{Each chapter should be preceded by an abstract (10--15 lines long) that summarizes the content. The abstract will appear \textit{online} at \url{www.SpringerLink.com} and be available with unrestricted access. This allows unregistered users to read the abstract as a teaser for the complete chapter. As a general rule the abstracts will not appear in the printed version of your book unless it is the style of your particular book or that of the series to which your book belongs.
%Please use the 'starred' version of the new Springer \texttt{abstract} command for typesetting the text of the online abstracts (cf. source file of this chapter template \texttt{abstract}) and include them with the source files of your manuscript. Use the plain \texttt{abstract} command if the abstract is also to appear in the printed version of the book.}

\begin{abstract}
This article sets out the framework of algebraic quantum field theory in
curved spacetimes, based on the idea of local covariance. In this framework,
a quantum field theory is modelled by a functor from a category of spacetimes
to a category of ($C^*$)-algebras obeying supplementary conditions. 
Among other things: (a)
the key idea of relative Cauchy evolution is described in detail, and related to
the stress-energy tensor; (b) a systematic `rigidity argument' is used to generalise results from flat to curved spacetimes; (c)
a detailed discussion of the issue of selection of physical states is given, linking
notions of stability at microscopic, mesoscopic and macroscopic scales; (d) the notion of subtheories and global gauge transformations are formalised; (e) it is
shown that the general framework excludes the possibility of there being
a single preferred state in each spacetime, if the choice of states is local
and covariant. 
Many of the ideas are illustrated by the example of the free Klein--Gordon theory,
which is given a new `universal definition'. 
\end{abstract}

\section{Introduction}

There are many approaches to quantum field theory, almost all of which rely heavily, in one way or another, on concepts of symmetry. This refers, in particular, to the behaviour of a quantum field theory with respect to the symmetries of the spacetime (or space, for Euclidean formulations) on which it exists.
For example, Poincar\'e covariance is one of the defining
properties for a relativistic quantum field theory on Minkowski space, in conjunction with the concept of locality for 
observables \cite{StreaterWightman,Haag}. Any general account of quantum field 
theory on curved spacetimes faces the problem that reliance on symmetries is of little assistance
except in special cases. For these reasons, most work in the area has, until recently, focussed
on specific (usually non-self-interacting) models in specific spacetimes. Our purpose in this
article is different: we will explore what can be said, in a model-independent way, about
quantum fields in general curved spacetimes. This is motivated from several directions:
(a) as a matter of general principle, one wishes to understand how far the ideas and methods
of quantum field theory can be extended; (b) it recognizes that our knowledge of our own
universe is limited in both scope and detail, and that its actual geometry (even setting aside
questions of quantum geometric structure) is by no means that of a symmetric spacetime;
(c) it allows for some macroscopic features (e.g., a collapsing star, or an experimental apparatus) to be treated as `given' and not obtained from the microscopic  theory; (d) it provides a framework in which controlled approximations of complex situations by simple ones can be discussed.

Time has proved that the mathematical framework of operator algebras permits a very clear and efficient way to precisely 
formulate the conceptual underpinnings of quantum field theory --- locality and covariance --- and to analyse the consequences
\cite{Haag,Araki,BLOT,BaumWollen:1992},
hence the name ``algebraic quantum field theory''. Our presentation attempts to follow that line of thought for quantum field theory in
curved spacetime. In place of symmetry,  the concept of covariance or, more precisely, of local general covariance, is put at the centre of our approach, 
reflecting the considerable progress that has been made since this concept was given a concise formulation in the early years
of this millennium \cite{BrFrVe03,Verch01}. Early forms of the idea appear in work such as \cite{Kay79,Dimock1980,Wald_qft}.

We should make clear that the main purpose of this contribution is to set out the conceptual and mathematical structure 
of algebraic quantum field theory in curved spacetime, not its application to concrete situations, such as the Hawking effect 
\cite{KayWald-PhysRep,FredenhagenHaag:1990,MorettiPinamonti-BH,HaNaSt:1984,GLRV,GuidoLongo:2003,Sewell} 
or cosmology \cite{DapFrePin-Cosmo,DegVer2010,VerchRegensburg} or the Casimir effect
\cite{Few&Pfen06, Marecki:2006, Fewster2007}.  It is intended to be read alongside other articles (in particular, those cited)  which provide the context and application for the structures discussed here.
There are also a number of topics that have not been discussed owing to constraints of time, space and energy. 
Some of these will be listed below, after we  indicate what is covered here.

We start by considering the quantized linear scalar Klein-Gordon field on globally hyperbolic spacetimes 
as a motivating example from which some basic concepts of locally covariant
quantum field theory can be read off. Taking these as guidelines, the general concept of locally covariant quantum field theory will be formulated in terms of a
functor between a category of spacetimes and a category of $*$-algebras. Further assumptions will be added and their consequences studied, so that the general structure of the theory, in a model-independent
algebraic framework, begins to take shape. Central parts
of that structure are played by Einstein causality, and the time-slice axiom, which, by interplay with local covariance, induces the 
notion of relative Cauchy evolution and provides the theory with a dynamical structure. Then the states and their Hilbert-space representations are
discussed, as well as the concept of a state space for locally covariant quantum field theories. In that context, the microlocal spectrum condition
makes its appearance as the most promising, and at the same time most general, selection criterion for a space of physical states.

A further step in developing the theory derives from the fact that globally hyperbolic spacetimes can be deformed to more symmetric spacetimes. Together with the time-slice axiom and local covariance, this allows one to transfer properties that hold on Minkowski spacetime to general spacetimes. 
This fact has been observed and exploited in the literature \cite{FullingNarcowichWald,Verch:1993,Verch_nucspldua:1993,Verch01,Sanders_ReehSchlieder,Fewster:gauge}, however here we systemize it as a ``rigidity argument'' for locally covariant quantum field theory, and this is a new and original ingredient of this contribution. We will show that
Einstein causality, the Schlieder property and extended locality can all be extended to locally covariant QFTs  in this way, and also that the Reeh-Schlieder and the split properties are consequences of closely related arguments.

We will then discuss the relation of several other selection criteria for state spaces of physical states in locally covariant quantum
field theories and their relation to the microlocal spectrum condition, mainly quantum energy inequalities, and the existence of 
ground states (or more generally, passive states) in ultrastatic spacetimes. 
We go on to consider locally covariant quantum fields and present the spin and statistics relation in that framework. Furthermore, we
discuss embeddings of locally covariant quantum field theories into each other, which leads to a locally covariant concept of internal symmetries.
That provides a setting in which one can consider the question what it means that quantum field theories on different spacetimes can be
regarded as representing `the same physics''. Another new feature of this contribution is the observation that
models such as the Klein--Gordon theory can be given a ``universal definition'' at the functorial level, without direct reference to the theory on particular spacetimes.  
Finally, we present an argument showing that there is no locally covariant state under very general assumptions. 

As mentioned above, there are various topics that we have not been able to include, and while the
list of references is extensive, it is certainly not complete. 
Notable absences include discussion of gauge fields and charge superselection theory in curved spacetime \cite{GLRV,Br&Ru05,Ruzzi:2005,BrunettiRuzzi_topsect} 
as well as the perturbative
construction of locally covariant interacting quantum fields in curved spacetime \cite{BrFr2000,Ho&Wa01,Ho&Wa02,Hollands:2008,BruDueFre2009,FreRej_BVqft:2012}. 
The development of the latter has led to
a formalization of operator product expansions that may be seen as a particular approach towards algebraic quantum field theory in curved 
spacetimes, mainly investigated by Hollands and Wald  
\cite{Hollands-opecst:2007,HollandsWald-Axio:2010} (besides other results, this has led to a version of a PCT-theorem 
in curved spacetimes \cite{Hollands-PCT:2004}). 
We shall not discuss  Haag duality \cite{Ruzzi_punc:2005} or situations of
``geometric modular action'' \cite{GLRV,BuMundSummers:2000,GuidoLongo:2003}, nor the relation between Euclidean and Lorentzian quantum field theory for quantum fields
in curved spacetimes \cite{JaffeRitter:2007,JaffeJaekelMartinez:2014}, nor any form of constructive quantum field theory (beyond free fields) in curved spacetimes,
on this, cf.\ \cite{BarataJaekelMund:2013} and references cited there. 
A further omission concerns spacetimes that are not globally hyperbolic or have boundaries 
\cite{Kay_Flocality:1992,Yurtsever:1994,KayRadWald:1997,Rehren-algHolo:2000,LongoRehren-bdCFT:2004,LongoRehren-boundQFT:2011}.

Aside from a familiarity with quantum field theory on Minkowski spacetime and
the standard terminology from general relativity e.g.\ at the level of \cite{Wald_gr},
some basic knowledge of category theory is assumed on part of the reader (as regards concepts like category, functor, morphism, naturality, for which see e.g.,~\cite{MacLane}). 
Knowledge of typical mathematical concepts of functional analysis in Hilbert spaces are taken for granted, but we will summarize most of the relevant 
background on operator algebras as far as it is needed. 

The abstract structures and arguments will be illustrated by the example of the free linear scalar field. This might give the impression that the theory only makes statements about linear quantum fields. We emphasize that this is not so and that 
locally covariant quantum field theories with self-interaction have been constructed perturbatively, and that there are also such 
interacting quantum field theories obeying the time-slice axiom \cite{Ho&Wa01,Ho&Wa02,Hollands:2008,ChiFre:2009}; it is these assumptions 
on which our theoretical arguments mainly
rest. The long-standing problem of establishing the existence of interacting quantum field theories beyond perturbation theory 
in physical spacetime dimension remains; yet we hope that the principle of local covariance will provide a new guideline in the attempts
of their construction.

\section{A motivating example}
\label{sec:KG}

We begin with the simplest model of QFT in curved spacetime, the linear Klein--Gordon field with field equation
$P_\Mb\phi:=(\Box_\Mb+m^2+\xi R_\Mb)\phi=0$ on spacetime $\Mb$,\footnote{We adopt signature convention $+-\cdots-$ for the metric.}
where the mass $m\ge 0$ and coupling $\xi\in\RR$ are fixed but
arbitrary.
Here, anticipating later developments, we have put subscripts on the d'Alembertian and scalar
curvature, to indicate the spacetime under consideration. 
In each globally hyperbolic spacetime $\Mb$, the field algebra $\Af(\Mb)$ of this theory may
be presented in terms of generators $\Phi_\Mb(f)$ labelled
by complex-valued test functions $f\in\CoinX{\Mb}$ and subject to relations
\begin{description}%[KG4]
\item[KG1] linearity of $f\mapsto\Phi_\Mb(f)$ 
\item[KG2] hermiticity, $\Phi_\Mb(f)^*=\Phi_\Mb(\overline{f})$
\item[KG3] the field equation, $\Phi_\Mb(P_\Mb f)=0$
\item[KG4] the canonical commutation relations $[\Phi_\Mb(f), \Phi_\Mb(h)]=i E_\Mb(f,h)\II_{\Af(\Mb)}$
\end{description}
which hold for all $f,h\in\CoinX{\Mb}$.
Equivalently, one can define $\Af(\Mb)$ as the Borchers-Uhlmann algebra, i.e.\ the quotient of the tensor algebra over the test-function
space $\CoinX{\Mb}$ by the relations described in KG1 to KG4 \cite{Borchers:1962,Uhlmann:1962}. Thus defined, $\Af(\Mb)$ does not admit the structure of a
$C^*$-algebra. This is not always required, but in some situations, it is useful to have $\Af(\Mb)$ as a $C^*$-algebra. The canonical
way of reaching a $C^*$-algebraic description of the quantized linear Klein-Gordon field on $\Mb$ proceeds as follows: Define 
 $\mathcal{K}(\Mb) = C_0^\infty(\Mb,\Reals)/(P_\Mb C_0^\infty(\Mb,\Reals))$, and write $f_\sim = f + P_\Mb C_0^\infty(\Mb,\Reals)$ for $f \in C_0^\infty(\Mb,\Reals)$. Then define
$\Af(\Mb)$ to be the \emph{Weyl algebra} of the linear Klein-Gordon field on $\Mb$, which is defined as the (unique~\cite{BratteliRobinson}) $C^*$-algebra
generated by elements $W_\Mb (f_\sim)$, $f_\sim \in \mathcal{K}(\Mb)$, and a unit element $\II$, subject to the relations 
$W_\Mb(f_\sim)W_\Mb(h_\sim) = {\rm e}^{iE_\Mb(f,h)/2}W_\Mb(f_\sim + h_\sim)$, $W_\Mb(- f_\sim) = W_\Mb(f_\sim)^*$ and $W_\Mb(0) = \II$.

The algebraic description of the theory
on $\Mb$ is useful for many applications, when supplemented
by a suitable class of states such as the Hadamard class described
in Sect.~\ref{sect:states} and \cite{Wald_qft,KayWald-PhysRep,Radzikowski_ulocal1996}.   
A rather richer structure is revealed, however, 
when one relates the algebras obtained on different, but suitably related, spacetimes.

Consider two spacetimes $\Mb$ and $\Nb$ and 
a smooth map $\psi:\Mb\to\Nb$. Our first aim is to understand
what constraints $\Mb$, $\Nb$ and $\psi$ should satisfy in order that there can be a meaningful relationship between $\Af(\Mb)$ and $\Af(\Nb)$. 
The sort of relationship we intend here is one in which
the generating smeared fields are related directly to one another
in the following way. Provided that
$\psi$ is smoothly invertible on its range, we may push forward test functions from $\CoinX{\Mb}$ to $\CoinX{\Nb}$
according to 
\begin{equation}
(\psi_* f)(p) = \begin{cases}  f(\psi^{-1}(p)) & p\in \psi(\Mb)\\
0 & \text{otherwise}\end{cases}
\end{equation}
It is natural to use the push-forward to map smeared fields in $\Mb$ to smeared fields in $\Nb$, writing
\begin{equation}\label{eq:covariantscalarfield}
\Af(\psi)\Phi_\Mb(f):=\Phi_\Nb(\psi_*f).
\end{equation}
(In the Weyl formulation, one uses $\Af(\psi)W_\Mb(f_\sim) = W_\Nb((\psi_*f)_\sim)$
where the subscript $\sim$ refers to $\Mb$ and $\Nb$, respectively,
and the following discussion would proceed completely 
analogously.)
However, the assignment~\eqref{eq:covariantscalarfield} is only well-defined if it is compatible with the algebraic relations holding in $\Af(\Mb)$ and $\Af(\Nb)$;
in particular, we must have
$\Phi_\Nb(\psi_* P_\Mb f)=0$ for all test functions $f\in\CoinX{\Mb}$.
Further conditions arise if we wish to extend $\Af(\psi)$ from 
the generators to the full algebra. Here, the simplest possibility is that $\Af(\psi)$ should be
a $*$-homomorphism that also preserves units. In that
case, the commutation relations, together with \eqref{eq:covariantscalarfield}, give
\begin{equation}
[\Phi_\Nb(\psi_*f), \Phi_\Nb(\psi_* h)] =
\Af(\psi)[\Phi_\Mb(f), \Phi_\Mb( h)] = i E_\Mb(f,h)\II_{\Af(\Nb)}
\end{equation}
and hence $E_\Nb(\psi_*f,\psi_*h)=E_\Mb(f,h)$ for
all $f,h\in\CoinX{\Mb}$. We see that $E_\Mb$ is the pull-back, $E_\Mb=(\psi\times\psi)^*E_\Nb$,
of $E_\Nb$ and its wave-front set~\cite{Hoermander1} therefore obeys 
\[
\WF(E_\Mb)\subset (\psi\times\psi)^*\WF(E_\Nb).
\]
Given the known structure of both sides, 
we deduce that $\psi^*$ must map null-covectors on $\Nb$ to null covectors on $\Mb$,
and preserve time-orientation. This already restricts $\psi$ to be a conformal isometry,
and in fact one may see that it must be an isometry unless $P_\Mb$ is conformally invariant.
With an eye to other theories such as the pseudoscalar or Maxwell fields,  we might reasonably
require $\psi$ to preserve not only the time-orientation, but also the spacetime orientation. 

We have seen that the local structure of $\psi$ is quite restricted if there is to be any hope
of implementing \eqref{eq:covariantscalarfield}. In fact, the condition $E_\Mb=(\psi\times\psi)^*E_\Nb$,
of $E_\Nb$ also has global consequences: the image $\psi(\Mb)$ must be a causally convex subset of $\Nb$, which requires that every causal curve in $\Nb$ whose endpoints lie in $\psi(\Mb)$ should
be contained entirely in $\psi(\Mb)$. Examples showing the failure of this relation in the absence of causal convexity are to be found in
\cite{Kay_Flocality:1992, BarGinouxPfaffle}; more generally, the conclusion follows from the fact that
singularities propagate along null geodesics.

Our discussion has led us, with very little alternative, to a specification of those maps of interest: 
$\psi:\Mb\to\Nb$ is a smooth, isometric embedding, preserving orientation and time-orientation, and
with causally convex image. For such $\psi$, we now have a unit-preserving $*$-homomorphism
$\Af(\psi):\Af(\Mb)\to\Af(\Nb)$ which turns out to be injective.\footnote{The algebra $\Af(\Mb)$
is simple (and not the zero algebra!), so $\Af(\psi)$ either has trivial kernel or full kernel; the latter case 
is excluded because $\Af(\psi)\II_{\Af(\Mb)}=\II_{\Af(\Nb)}\neq 0$.} We may observe
something more: if we also consider a map $\varphi:\Lb\to\Mb$ obeying these conditions, 
then the same is true of the composition $\psi\circ\varphi:\Lb\to\Nb$. It is clear from 
\eqref{eq:covariantscalarfield} that 
\begin{equation}
\Af(\psi\circ\varphi)\Phi_\Lb(f) = \Phi_\Nb((\psi\circ\varphi)_*f) =
\Phi_\Nb(\psi_*\varphi_*f)= \Af(\psi)(\Af(\varphi)\Phi_\Lb(f))
\end{equation}
for all $f\in\CoinX{\Lb}$, and we 
extend from the generators to obtain  
\begin{equation}\label{eq:covariance}
\Af(\psi\circ\varphi)= \Af(\psi)\circ \Af(\varphi).
\end{equation}
In addition, it is clear that the identity map $\id_\Mb$ of $\Mb$ corresponds to $\Af(\id_\Mb)=
\id_{\Af(\Mb)}$. These observations may all be summarised in the single statement that
\emph{the theory is described by a covariant functor} $\Af$ between two categories:
\begin{description}%[$\Loc$]
\item[$\Loc$] the category whose objects are all globally hyperbolic spacetimes
$\Mb=(\Mc,g,\ogth,\tgth)$ of fixed dimension $n$ with finitely many connected components, and whose morphisms are
smooth isometric embeddings, preserving orientation and time-orientation, 
having causally convex image. Here $\Mc$ is the underlying manifold, with metric $g$, while the symbol  $\ogth$ stands
for a choice of orientation, represented by one of the components 
of the set of nowhere-zero smooth $n$-forms on $\Mc$. Similarly, 
$\tgth$ denotes the time-orientation, represented by one of the components 
of the set of nowhere-zero smooth $g$-timelike $1$-forms on $\Mc$. 
\item[$\Alg$] the category of unital $*$-algebras  \emph{excluding the zero algebra}, with unit-preserving injective $*$-homomorphisms 
as morphisms. 
\end{description}
Here, we have anticipated future developments by allowing for disconnected spacetimes
when defining $\Loc$. If $\Mb$ is disconnected, and $\psi:\Mb\to\Nb$ according to 
the definition above, then causal convexity of $\psi(\Mb)$ forces its various components
to be causally disjoint -- no causal curve can join one component to another. Of course, one should be alive to the possibility that the
example has features that might not be shared by all theories.
In particular, $\Loc$ admits spacetime embeddings in which
the components of the image have closures that are in causal
contact, or allow for self-touchings at their boundaries. While
the Klein--Gordon theory has well-defined morphisms corresponding
to such embeddings, it is conceivable that there are reasonable theories 
that do not, and that a more conservative starting point should
be found in due course. In the definition of $\Alg$, we have excluded the zero unital algebra (consisting of a single element
which is both the zero and unit) to avoid some pathologies
and to ensure that $\Alg$ has an initial object, namely the
algebra of complex numbers.\footnote{An initial object 
in a category $\Ct$ is an object $I$ with the property that 
there is, to each object $C$ of $\Ct$, exactly one morphism
from $I$ to $C$.} Accordingly, the unit is distinct from the zero element
in every object of $\Alg$.

Although we have reached this structure by means of an example, it has a clear physical interpretation and could be motivated in its own terms. Namely, the morphisms $\psi$ specify embeddings in which 
all causal relations between points in the image $\psi(\Mb)$ (with respect to $\Nb$) are 
already causal relations between the corresponding points in $\Mb$. Physics, by which 
we mean here degrees of freedom and laws of motion (without yet specifying boundary
or initial conditions), in the image region would be expected to correspond to that in the domain spacetime: this is a version of the principle of locality. In particular, we expect the physics on
the smaller spacetime to be faithfully represented within that of the larger,
and that there should be no distinction between physics
in the embedded region $\psi(\Mb)$ and in the spacetime $\Mb$. 
The functorial definition provides a consistency
mechanism that protects the ignorance of an experimenter
within $\psi(\Mb)$ of the nature (or even existence) of
the spacetime beyond the region under her control.  
Further discussion of these ideas can be found in~\cite{Few_Chicheley:2015}. 

Taking all of the above into account, we will adopt the general assumption that any theory that is both covariant and respects the principle of locality should be described by a covariant functor from $\Loc$ (or another
suitable category of spacetimes) to a category of physical systems, in
which morphisms represent embeddings of one system as a subsystem of another
and are required to be monic. This is a working hypothesis for the development of a model-independent theory, but note that
\begin{enumerate}
\item the whole enterprise is questionable for spacetimes that are smaller in scale than
the physical systems they  support (e.g., measured by a Compton wavelength)
\item as mentioned, $\Loc$ may admit too wide a variety of morphisms for some theories
\item for gauge theories in particular the requirement of injectivity is sometimes in conflict
with other desirable features of the theory, particularly in order to capture topological aspects. See, e.g., the discussion following
Theorem~\ref{thm:causality}. 
\end{enumerate}

\section{General assumptions and first consequences}
\label{sect:assumptions}

We begin a more formal development of the structure, which rests on a 
number of general assumptions. The first has already been motivated:

\begin{assumption}[Local covariance]
\label{ax:loc_cov} A locally covariant theory is a functor 
\[
\Af:\Loc\to\Alg.
\]
\end{assumption}
Depending on the application, one might wish to specify the target category
more stringently, e.g., requiring that $\Af$ takes values in the subcategory $\CAlg$ of $\Alg$,
consisting of unital $C^*$-algebras. One may also formulate locally covariant descriptions
for theories other than QFT by allowing a more general category $\Phys$. 
Here, however,  we will remain in the algebraic description for the most part.
 
Given this starting point, we may define a net of local algebras in each spacetime $\Mb=
(\Mc,g,\ogth,\tgth)$.
Let $\OO(\Mb)$ be the set of all open causally convex subsets of $\Mb$, with at most
finitely many connected components (which are necessarily
causally disjoint). 
Then, for each nonempty $O\in\OO(\Mb)$, we may define a new object $\Mb|_O=
(O,g|_O,\ogth|_O,\tgth|_O)$ of $\Loc$, which is simply the set $O$ equipped with the
causal structures induced from $\Mb$ and regarded as a spacetime in its own right. 
In addition, the subset embedding of $O$ in $\Mc$ is evidently a smooth embedding
which is an isometric (time)-orientation preserving map owing to the way we have
defined $\Mb$. Thus it defines a morphism $\iota_{\Mb;O}:\Mb|_O\to\Mb$ in $\Loc$. 
The functor $\Af$ therefore assigns both an algebra $\Af(\Mb|_O)$ and a morphism
$\Af(\iota_{\Mb;O})$ of $\Af(\Mb|_O)$ into $\Af(\Mb)$. The image of this morphism,
\begin{equation}
\Af^\kin(\Mb;O) = \Af(\iota_{\Mb;O})(\Af(\Mb_O)),
\end{equation}
is called the \emph{kinematic algebra} associated with region $O$, and gives a 
description of the physics of the theory within $O$.\footnote{Alternatively, and perhaps more in the spirit of a categorical description, one might say that the morphism $\Af(\iota_{\Mb;O})$, regarded as defining a subobject of $\Af(\Mb)$, should be the focus here~\cite{FewVer:dynloc_theory}.} The kinematic algebras have some immediate properties.
First, suppose that $O_1\subset O_2$. Then the factorisation
 $\iota_{\Mb;O_1}=
\iota_{\Mb;O_2}\circ\iota_{\Mb|_{O_2};O_1}$ of the inclusion morphism 
implies that $\Af(\iota_{\Mb;O_1})=\Af(
\iota_{\Mb;O_2})\circ\Af(\iota_{\Mb|_{O_2};O_1})$ and hence
that
\begin{equation}\label{eq:kin_isotony}
\Af^\kin(\Mb;O_1)\subset \Af^\kin(\Mb;O_2).
\end{equation}
In other words, the kinematic net is isotonous.  Consequently, if  $O_1,O_2\in \OO(\Mb)$ have
a nonempty intersection $O_1\cap O_2$ (which is
seen to be causally convex and therefore an element of $\OO(\Mb)$)
then 
\begin{equation}\label{eq:kin_intersection1}
\Af^\kin(\Mb;O_1\cap O_2) \subset \Af^\kin(\Mb;O_1)\cap
\Af^\kin(\Mb;O_2).
\end{equation}
(One does not expect equality here.) On the other hand, if
$O_1,O_2\in\OO(\Mb)$ are nonempty and their union is causally convex
then $\Af^\kin(\Mb;O_1\cup O_2)$ contains both $\Af^\kin(\Mb;O_i)$
and therefore the algebra that they generate, so
\begin{equation}\label{eq:kin_union1}
\Af^\kin(\Mb;O_1)\vee \Af^\kin(\Mb;O_2)\subset \Af^\kin(\Mb;O_1\cup O_2).
\end{equation} 
In the $C^*$-algebraic setting, where $\Af:\Loc\to\CAlg$, we may sharpen 
this result so that the left-hand side
is  the $C^*$-subalgebra of $\Af^\kin(\Mb;O_1\cup O_2)$ generated 
by the $\Af^\kin(\Mb;O_i)$.\footnote{In a general categorical setting, 
one would employ the \emph{categorical union} of the $\Af^\kin(\Mb;O_i)$.}
Both~\eqref{eq:kin_intersection1} and~\eqref{eq:kin_union1}
extend to finitely many $O_i$ in obvious ways; if
equality holds in \eqref{eq:kin_union1}, the theory $\Af$
will be described as \emph{finitely additive}. 

Next, consider a morphism $\psi:\Mb\to\Nb$. 
Then the spacetimes $\Mb|_{O}$ and $\Nb|_{\psi(O)}$ are
isomorphic via the map $\hat{\psi}_O$ obtained as the restriction of
$\psi$ to $O$, obeying 
 $\iota_{\Nb;\psi(O)}\circ\hat{\psi}_O=\psi\circ\iota_{\Mb;O}$. 
Applying the functor $\Af$, and noting that $\Af(\hat{\psi}|_O)$
is an isomorphism, we find
\begin{equation}\label{eq:kin_covariance}
\Af^\kin(\Nb;\psi(O)) = \Af(\psi)(\Af^\kin(\Mb;O)).
\end{equation}
An important special case arises where $\psi:\Mb\to\Mb$, 
i.e., $\psi\in\End(\Mb)$, in which $\Af(\psi)$ defines an
endomorphism of the kinematic net, or a net isomorphism
in the case where $\psi$ is an isomorphism, $\psi\in\Aut(\Mb)$.
In particular, we see that there is a homomorphism of monoids
from $\End(\Mb)$ to $\End(\Af(\Mb))$, that
restricts to a group homomorphism from $\Aut(\Mb)$ to
$\Aut(\Af(\Mb))$.

The properties \eqref{eq:kin_isotony}, \eqref{eq:kin_intersection1}, \eqref{eq:kin_union1} and \eqref{eq:kin_covariance} are direct generalisations
of the properties of the nets of local algebras encountered in
Minkowski space AQFT; see \cite{Haag} and \cite{BaumWollen:1992}. 
It is remarkable that they all follow without further input from the 
single assumption that the theory is described functorially. 
As an application of \eqref{eq:kin_union1} and \eqref{eq:kin_covariance}, suppose that
$\Mb$ has finitely many connected components $\Mc_i$ and $\psi:\Mb\to\Nb$. Then
we have
\begin{equation}\label{eq:gen_by_cpts}
\bigvee_{i} \Af^\kin(\Nb;\psi(\Mc_i)) \subset \Af^\kin(\Nb;\psi(\Mc)).
\end{equation}
If $\Af$ is finitely additive, then equality holds in \eqref{eq:gen_by_cpts};
this need not be true in general. 
%\todo[inline]{We will return to this later.}

Before proceeding to the other standard assumptions, two further definitions are required. For
a region $O\subset \Mb$, we write $O':=\Mb\setminus\overline{J_\Mb(O)}$
for its open causal complement; in addition, a morphism
$\psi:\Mb\to\Nb$ will be described as \emph{Cauchy}
if $\psi(\Mb)$ contains a Cauchy surface for $\Nb$ 
(equivalently, if every inextendible timelike curve in $\Nb$
intersects $\psi(\Mb)$). The remaining assumptions are:
\begin{assumption}[Einstein Causality]
 \label{ax:Einstein_causality}
If $O_1, O_2\in \OO(\Mb)$ are causally disjoint
in the sense that $O_1\subset O_2':=\Mb\setminus\overline{J_\Mb(O_2)}$, then
\begin{equation}\label{eq:Einstein_causality}
[\Af^\kin(\Mb;O_1),\Af^\kin(\Mb;O_2)] = \{0\}.
\end{equation}
\end{assumption}
\begin{assumption}[Timeslice] \label{ax:timeslice} 
If $\psi:\Mb\to\Nb$ is Cauchy then $\Af(\psi)$ is an isomorphism.
\end{assumption}
Unless otherwise specified, the term `locally covariant QFT' will refer
to a functor obeying Assumptions  \ref{ax:loc_cov}--\ref{ax:timeslice}.
In Sect.~\ref{sect:rigidity}
it will be seen that Assumption~\ref{ax:Einstein_causality} is partly
redundant: it is enough that Einstein causality should hold
for \emph{one} pair of spacelike separated regions in \emph{one}
spacetime for it to hold for suitable spacelike separated regions in general spacetimes. Together with other assumptions, Einstein
causality leads to an additional \emph{monoidal structure}
on the theory -- see~\cite{BrFrImRe:2014} and Sect.~\ref{sect:rigidity}
for discussion. If the theory not only  
describes observables, but also smeared fermionic fields, for example,
then a suitable graded commutator should be employed. 

The timeslice assumption is one of the lynch-pins of the structure
and encodes the idea that the theory \emph{has} a dynamical law, although \emph{what} it is is left unspecified. It has an immediate
consequence: if $O\in\OO(\Mb)$ is nonempty, then 
$O$ contains a Cauchy surface of the Cauchy development 
$D_\Mb(O)$ -- the set of all points $p$ in $\Mb$ with the property that
all inextendible piecewise-smooth causal curves through $p$ intersect $O$,
which is open, causally convex and therefore a member of 
$\OO(\Mb)$.\footnote{Some
authors, notably Penrose~\cite{Penrose1972} and Geroch~\cite{Geroch_domdep:1970}, define the Cauchy development  with timelike curves of various types. We follow~\cite{ONeill,BeemEhrlichEasley,HawkingEllis,Wald_gr}. Many authors
only define the Cauchy development for achronal sets.
The fact that $D_\Mb(O)$ is open is most easily seen using limit curves cf.\ \cite[Prop.~3.31]{BeemEhrlichEasley} or~\cite[Lem.~6.2.1]{HawkingEllis}.}
We already know that $\iota_{\Mb;O}$ factors as
$\iota_{\Mb;O}=\iota_{\Mb;D_\Mb(O)}\circ \iota_{\Mb|_{D_\Mb(O)};O}$. Applying the functor, the timeslice 
property entails that $\Af(\iota_{\Mb|_{D_\Mb(O)};O})$ is an isomorphism and so
\begin{equation}\label{eq:kin_timeslice}
\Af^\kin(\Mb;O) = \Af^\kin(\Mb;D_\Mb(O)).
\end{equation}
Hence we may immediately strengthen \eqref{eq:kin_intersection1} to
\begin{equation}\label{eq:kin_intersection2}
\Af^\kin(\Mb;D_\Mb(O_1)\cap D_\Mb(O_2)) \subset \Af^\kin(\Mb;O_1)\cap \Af^\kin(\Mb;O_2).
\end{equation}
The timeslice property and its ramifications will be dominant
themes in our discussion.

To conclude this section, we note that various models
obeying the general assumptions listed have been constructed.
The prototypical example is the free Klein--Gordon model 
described in Sec.~\ref{sec:KG}. There, it was shown that
the theory is given in terms of a functor $\Af:\Loc\to\Alg$,
with algebras $\Af(\Mb)$ generated by `smeared fields'
$\Phi_\Mb(f)$ ($f\in\CoinX{\Mb})$ and subject to relations
KG1--4. It is easily seen that, for nonempty
 $O\in\OO(\Mb)$, the kinematic algebra $\Af^\kin(\Mb;O)$
is the subalgebra of $\Af(\Mb)$ generated by those
$\Phi_\Mb(f)$ with $f\in C_0^\infty(O)$. Then \eqref{eq:gen_by_cpts} holds with equality
and Einstein causality holds because
$\supp E_\Mb f\subset J_\Mb(\supp f)$. The
timeslice property can be shown by  standard arguments: 
if $\psi:\Mb\to\Nb$
is a Cauchy morphism, let 
$\chi\in C^\infty(\Nb)$ be chosen so that
$\chi\equiv 0$ to the future of $\Sigma^+$ and $\chi\equiv 1$ to the past of $\Sigma^-$,
where $\Sigma^\pm$ are  Cauchy surfaces in $\psi(\Mb)$.
If $f\in\CoinX{\Nb}$ then $f'=P_\Nb \chi E_\Nb f$ may be shown
to have compact support in $\psi(\Mb)$ and to obey $f-f'\in P_\Nb\CoinX{\psi(\Mb)}$.\footnote{
One observes that $\supp f'\subset J_\Nb^-(\Sigma^+)\cap J_\Nb^+(\Sigma^-)\cap J_\Nb(\supp f)$,
which is compact (see e.g., \cite[Lem.~A.5.4]{BarGinouxPfaffle}) and contained in $\psi(\Mb)$. Moreover, by the support properties of $\chi$ and the definition of $E_\Nb^-$, we have $E_\Nb^- f' = \chi E_\Nb f$; similarly,  $E_\Nb^+ P_\Nb (1-\chi)E_\Nb f=(1-\chi)E_\Nb f$. Adding these two expressions and using
$P_\Nb E_\Nb f=0$, we obtain $E_\Nb f'=E_\Nb f$ and hence 
$f-f'\in P_\Nb\CoinX{\psi(\Mb)}$ by e.g.,~\cite[Thm 3.4.7]{BarGinouxPfaffle}. }
Then 
\[
\Phi_{\Nb}(f)= \Phi_{\Nb}(f')=\Phi_\Nb(\psi_* \psi^*f') = \Af(\psi)\Phi_\Mb(\psi^*f')
\] 
by KG3, the support properties of $f'$ and the definition of $\Af(\psi)$. As
every generator of $\Af(\Nb)$ lies in the image of the injective map $\Af(\psi)$, 
it is an isomorphism.

Similarly, models such as the Proca and (with modifications) Dirac fields 
also fit into the framework~\cite{Dappiaggi:2011,Sanders_dirac:2010}, as do the perturbatively
constructed models of pAQFT -- see, e.g., \cite{Ho&Wa01,Ho&Wa02,Hollands:2008}  for details.

As a further model of interest, let us return to the 
Klein--Gordon theory $\Af$. Each algebra $\Af(\Mb)$
contains a unital $*$-subalgebra $\Af^{\text{ev}}(\Mb)$ of
elements that generated by the unit together with 
bilinear elements $\Phi_\Mb(f)\Phi_\Mb(h)$ ($f,h\in\Df(\Mb)$). 
It is easily seen that, for $\psi:\Mb\to\Nb$, the morphism
$\Af(\psi)$ restricts to a morphism $\Af^{\text{ev}}(\psi):
\Af^{\text{ev}}(\Mb)\to\Af^{\text{ev}}(\Nb)$, which 
defines a new locally covariant theory $\Af^{\text{ev}}:\Loc\to\Alg$. 
Like $\Af$, this theory obeys Assumptions \ref{ax:loc_cov}--\ref{ax:timeslice}.
However, relation~\eqref{eq:kin_union1}
cannot be strengthened to an equality for this theory: 
consider spacelike separated $O_i\in\OO(\Mb)$ and let
$f_i\in\CoinX{O_i}$ ($i=1,2$). Then  
the kinematic algebra $\Af^{\text{ev},\kin}(\Mb;O_1\cup O_2)$
contains an element $\Phi_\Mb(f_1)\Phi_\Mb(f_2)$,
which is not contained in $\Af^{\text{ev},\kin}(\Mb;O_1)\vee
\Af^{\text{ev},\kin}(\Mb;O_2)$; in other words, 
$\Af^{\text{ev}}$ is not finitely additive.

One should also bear in mind that the general
assumptions so far also allow for models that display
unphysical properties. For example, define a  theory
by
\begin{equation}\label{eq:oneortwo_obj}
\Bf(\Mb)=\begin{cases} \Af(\Mb) & \text{if $\Mb$ has noncompact
Cauchy surfaces} \\
\Af(\Mb)\otimes\Af(\Mb) & \text{if $\Mb$ has compact
Cauchy surfaces} 
\end{cases}
\end{equation}
and, for $\psi:\Mb\to\Nb$,
\begin{equation}\label{eq:oneortwo_mor}
\Bf(\psi)A =\begin{cases} \Af(\psi)A  & \text{$\Nb$ has noncompact Cauchy surfaces} \\
\left(\Af(\psi)\otimes\Af(\psi) \right) A & \text{$\Mb$ has compact Cauchy surfaces}\\
\left(\Af(\psi)\right)A\otimes\II &  \text{otherwise}
\end{cases}
\end{equation}
(by a general result in Lorentzian geometry, the only case that
can arise under `otherwise' is that of $\Nb$ having compact
Cauchy surfaces and $\Mb$ having noncompact Cauchy surfaces; see \cite[Prop.~A.1]{FewVer:dynloc_theory},
which is based on results in~\cite{BILY}).

The reader may verify that theory obeys Assumptions \ref{ax:loc_cov}--\ref{ax:timeslice}, while being a theory of a single scalar field in some spacetimes, and of two independent
scalar fields in others. More examples in a similar vein can be
found in \cite{FewVer:dynloc_theory}; indeed, one could employ
the same construction using any locally covariant theory as
a starting point. Therefore, the Assumptions \ref{ax:loc_cov}--\ref{ax:timeslice} are
not in themselves sufficient to guarantee that a theory
represents the same physics in all spacetimes,
an issue that will be studied further in Sect.~\ref{sect:SPASs}.  

\section{Relative Cauchy evolution}\label{sect:rce}

The locally covariant approach not only conveniently summarises many 
general facts about QFT in curved spacetimes, but has also led to new
developments in the subject. One such is the idea of \emph{relative Cauchy
evolution}, introduced in \cite{BrFrVe03} and further developed in
\cite{FewVer:dynloc_theory}, which allows for the comparison
of the dynamics of a theory on different spacetimes, even when
one cannot be embedded in the other.  

Let $\Mb=(\Mc,g,\ogth,\tgth)$ be
a globally hyperbolic spacetime. If $h$ is a smooth compactly supported
rank-$2$ covariant tensor field that is `not too big' then we can define
a deformed spacetime $\Mb[h]=(\Mc,g+h, \ogth,\tgth[h])$
which is still globally hyperbolic, where $\tgth[h]$ is the unique 
choice of time orientation agreeing with $\tgth$ outside the support of $h$.\footnote{
The orientation need not be changed when the metric changes; recall
that $\ogth$ is a component of the nonzero smooth $n$-forms, and
not e.g., the volume form.} The set of all such metric perturbations will be denoted $H(\Mb)$.
The idea is now to select regions to the
past and future of the metric perturbation that are common to $\Mb$ and $\Mb[h]$
and contain Cauchy surfaces thereof. This is achieved by choosing
Cauchy morphisms $\imath^\pm:\Mb^\pm\to \Mb$ 
with images $\imath^\pm(\Mb^\pm)$ contained in $\Mc\setminus J^\mp_\Mb(\supp h)$,
i.e., so that $\supp h$ has trivial intersection
with the causal future of $\imath^+(\Mb^+)$ and the causal past of $\imath^-(\Mb^-)$.
Given these choices, there are Cauchy morphisms $j^\pm:\Mb^\pm\to\Mb[h]$ 
with {\em the same underlying maps} as $\imath^\pm$. 

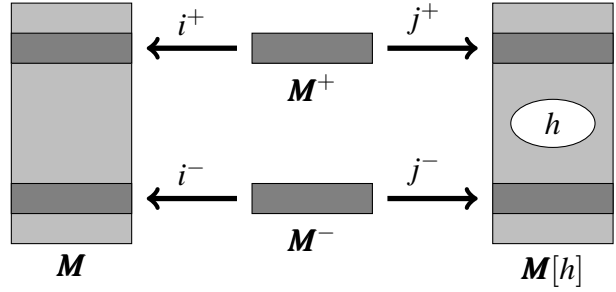
\begin{figure}[t]
\begin{center}
\begin{tikzpicture}[scale=0.8]
\definecolor{Gold}{rgb}{.93,.82,.24}
\definecolor{Orange}{rgb}{1,0.5,0}
\draw[fill=lightgray] (-4,0) -- ++(2,0) -- ++(0,4) -- ++(-2,0) -- cycle;
\draw[fill=lightgray] (4,0) -- ++(2,0) -- ++(0,4) -- ++(-2,0) -- cycle;
\draw[fill=gray] (4,3) -- ++(2,0) -- ++(0,0.5) -- ++(-2,0) -- cycle;
\draw[fill=gray] (0,3) -- ++(2,0) -- ++(0,0.5) -- ++(-2,0) -- cycle;
\draw[fill=gray] (-4,3) -- ++(2,0) -- ++(0,0.5) -- ++(-2,0) -- cycle;
\draw[fill=gray] (4,0.5) -- ++(2,0) -- ++(0,0.5) -- ++(-2,0) -- cycle;
\draw[fill=gray] (0,0.5) -- ++(2,0) -- ++(0,0.5) -- ++(-2,0) -- cycle;
\draw[fill=gray] (-4,0.5) -- ++(2,0) -- ++(0,0.5) -- ++(-2,0) -- cycle;
\draw[color=black,line width=2pt,->] (2.25,3.25) -- (3.75,3.25) node[pos=0.4,above]{$j^+$};
\draw[color=black,line width=2pt,->] (2.25,0.75) -- (3.75,0.75) node[pos=0.4,above]{$j^-$};
\draw[color=black,line width=2pt,->] (-0.25,3.25) -- (-1.75,3.25) node[pos=0.5,above]{$i^+$};
\draw[color=black,line width=2pt,->] (-0.25,0.75) -- (-1.75,0.75) node[pos=0.5,above]{$i^-$};
\draw[fill=white] (5,2) ellipse (0.7 and 0.4);
\node at (5,2) {$h$};
\node[anchor=north] at (5,0) {$\Mb[h]$};
\node[anchor=north] at (-3,0) {$\Mb$};
\node[anchor=north] at (1,3) {$\Mb^+$};
\node[anchor=north] at (1,0.5) {$\Mb^-$};
\end{tikzpicture}
\end{center}
\caption{Spacetimes involved in the construction of the relative Cauchy evolution}\label{fig:rce}
\end{figure}

The arrangement of spacetimes is displayed in pictorially in~Figure~\ref{fig:rce} 
and can be portrayed diagrammatically as
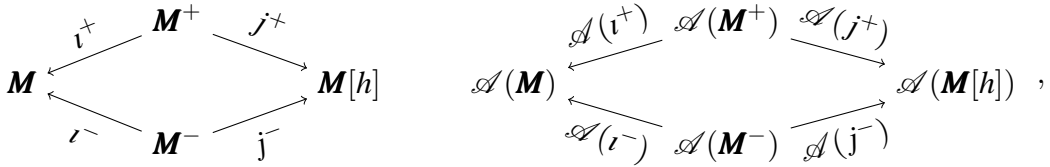
\[
\begin{tikzpicture}[baseline=0 em, description/.style={fill=white,inner sep=2pt}]
\matrix (m) [ampersand replacement=\&,matrix of math nodes, row sep=0.5em,
column sep=3em, text height=1.5ex, text depth=0.25ex]
{   \& \Mb^+ \&  \\
 \Mb \&  \& \Mb[h] \\
 \& \Mb^- \& \\ };
\path[->]
(m-1-2) edge node[above,sloped] {$ \imath^+$} (m-2-1)
        edge node[above,sloped] {$ j^+ $} (m-2-3)
(m-3-2) edge node[below,sloped] {$ \imath^-$} (m-2-1)
        edge node[below,sloped] {$ j^- $} (m-2-3);
\end{tikzpicture} \hspace{1.5em}
\begin{tikzpicture}[baseline=0 em, description/.style={fill=white,inner sep=2pt}]
\matrix (m) [ampersand replacement=\&,matrix of math nodes, row sep=0.5em,
column sep=3em, text height=1.5ex, text depth=0.25ex]
{   \& \Af(\Mb^+) \&  \\
 \Af(\Mb) \&  \& \Af(\Mb[h]) \\
 \& \Af(\Mb^-) \& \\ };
\path[->]
(m-1-2) edge node[above,sloped] {$ \Af(\imath^+) $} (m-2-1)
        edge node[above,sloped] {$\Af( j^+) $} (m-2-3)
(m-3-2) edge node[below,sloped] {$ \Af(\imath^-) $} (m-2-1)
        edge node[below,sloped] {$ \Af(j^-) $} (m-2-3);
\end{tikzpicture},
\]
where the diagram on the right is obtained by applying 
the functor corresponding to a locally covariant theory $\Af$. 
The important point is that the timeslice property ensures that all the morphisms in
this second diagram are isomorphisms, and so can be inverted. This
permits us to traverse the right-hand diagram clockwise, starting and
ending at $\Af(\Mb)$, to obtain the relative Cauchy evolution 
\[
\rce_\Mb[h] = \Af(\imath^-)\circ  \Af(j^-)^{-1} \circ
\Af( j^+)\circ  \Af(\imath^+)^{-1},
\]
which encodes the response of the theory to the metric variation $h$
as an automorphism of $\Af(\Mb)$. The relative Cauchy evolution is independent of the choices of
Cauchy morphisms made -- there is also a canonical choice in which $\Mb^\pm$ have
underlying manifolds $\Mc^\pm=\Mc\setminus J^\mp_\Mb(\supp h)$
(see~\cite[\S 3.4]{FewVer:dynloc_theory}). 

It is not hard to compute the relative Cauchy evolution
for the real scalar field. Fix any $\Mb\in\Loc$ and any compactly supported metric perturbation $h\in H(\Mb)$. As the relative Cauchy evolution is independent of the specific Cauchy morphisms used,  convenient choices can be made. Choose Cauchy surfaces $\Sigma^\pm$ of $\Mb$ so that $\Sigma^\pm\subset I^\pm_\Mb(\Sigma^\mp)$ and with $\supp h \subset I^-_\Mb(\Sigma^+)\cap I^+_\Mb(\Sigma^-)$, i.e., $h$ lies to the future of $\Sigma^-$ and the past of $\Sigma^+$. 
Then let $\Mc^\pm=I^\pm_\Mb(\Sigma^\pm)$ and define $\Mb^\pm=\Mb|_{\Mc^\pm}$, 
letting $\imath^\pm$ and $j^\pm$ be the inclusion morphisms of $\Mc^\pm$ in $\Mb$ and $\Mb[h]$ respectively, which are necessarily Cauchy.  
It is enough to evaluate the action of the relative Cauchy evolution on
the generators $\Phi_\Mb(f)$ of $\Af(\Mb)$ -- moreover, by the timeslice property (on $\Mb$) it is sufficient
to restrict to test functions $f$ supported in $\Mc^+$, for which  
\[
\Af( j^+)\circ  \Af(\imath^+)^{-1}\Phi_\Mb(f) = \Phi_{\Mb[h]}(f).
\]
Using the timeslice property on $\Mb[h]$ (cf.\ the discussion in Section~\ref{sect:assumptions}) 
we may write
$\Phi_{\Mb[h]}(f)= \Phi_{\Mb[h]}(P_{\Mb[h]}\chi E_{\Mb[h]} f)$,  where 
$\chi\in C^\infty(\Mc)$ has been chosen to vanish identically in $J_{\Mb[h]}^+(\Sigma^-)$,
and to take the value $1$ identically to the past of some other Cauchy surface in $\Mc^-$. 
In particular, $\chi$ vanishes on the support of $h$ and also on $\Mc^+$. With these choices, $P_{\Mb[h]}\chi E_{\Mb[h]} f$ is supported in
$\Mc^-$, whereupon 
\[
\rce_\Mb[h]\Phi_\Mb(f)= \Af(\imath^-)\circ  \Af(j^-)^{-1}\Phi_{\Mb[h]}(P_{\Mb[h]}\chi E_{\Mb[h]} f)
=\Phi_{\Mb}(P_{\Mb[h]}\chi E_{\Mb[h]} f).
\]
This expression can be simplified so as to remove the dependence on $\chi$. The support properties of $\chi$ entail 
that $\chi E_{\Mb[h]} f =\chi E^-_{\Mb[h]} f$, and so
\begin{align*}
P_{\Mb[h]}\chi E_{\Mb[h]} f &= f - P_{\Mb[h]}(1-\chi) E^-_{\Mb[h]} f \\
&=
f - (P_{\Mb[h]}-P_{\Mb})(1-\chi) E^-_{\Mb[h]} f  
- P_{\Mb}(1-\chi )E^-_{\Mb[h]} f .
\end{align*}
Moreover,  $P_{\Mb[h]}$ and $P_\Mb$ differ only where $\chi=0$, and outside
$J_{\Mb[h]}^+(\supp f)$, so
\[
(P_{\Mb[h]}-P_{\Mb})(1-\chi) E^-_{\Mb[h]} f  = (P_{\Mb[h]}-P_{\Mb}) E^-_{\Mb[h]} f 
=(P_{\Mb[h]}-P_{\Mb}) E_{\Mb[h]} f 
\]
which gives $ P_{\Mb[h]}\chi E_{\Mb[h]} f = f - (P_{\Mb[h]}-P_{\Mb}) E_{\Mb[h]} f -
P_{\Mb} (1-\chi)E^-_{\Mb[h]} f$. Crucially, $(1-\chi)E^-_{\Mb[h]} f$ is compactly supported, 
and so the field equation axiom KG3 gives
\begin{equation}\label{eq:rce_scalar}
\rce_{\Mb}[h] \Phi_\Mb(f) =\Phi_\Mb(f) - \Phi_\Mb((P_{\Mb[h]}-P_{\Mb}) E_{\Mb[h]} f)
\end{equation}
at least for those $f$ supported in $\Mc^+$. As previously mentioned, this suffices to fix
the action of $\rce_\Mb[h]$ on the whole of $\Af(\Mb)$. 

Equation~\eqref{eq:rce_scalar} clearly shows that the relative Cauchy evolution is
trivial if $f$ is supported within the causal complement of $\supp h$ (at least for
our class of $f$). In fact, it is true for any locally covariant theory that  
$\rce_\Mb[h]$ acts trivially on $\Af^\kin(\Mb;O)$ for any region $O\in\OO(\Mb)$
with $O\subset (\supp h)^\perp:=\Mb\setminus \supp h$,
which shows that the relative Cauchy evolution is local with respect
to the metric perturbation \cite[Prop. 3.5]{FewVer:dynloc_theory}.
Moreover, it is also covariant (again, for any locally covariant theory): given $\Mb\to\Nb$ and $h\in H(\Mb)$,
it can be shown \cite[Prop. 3.7]{FewVer:dynloc_theory} that $\psi_* h\in H(\Nb)$ and 
\[
\rce_\Nb[\psi_*h]\circ\Af(\psi) = \Af(\psi)\circ \rce_\Mb[h].
\]
A particular case of interest is where $\Nb=\Mb$ and $\psi\in\Aut(\Mb)$
is a spacetime symmetry.  Alternatively, if $K\subset\Mb$ is compact
and contains the support of $h\in H(\Mb)$, then 
\begin{equation}\label{eq:rce_diff_cov}
\rce_\Mb[\psi_*(h+g)-g]=\rce_\Mb[h]
\end{equation} 
for any diffeomorphism $\psi$
that acts trivially outside $K$ \cite[Prop. 4.1]{BrFrVe03}.

The interpretation of the relative Cauchy evolution is best understood by means of its functional derivatives, which turn out to be related to a stress-energy tensor. Let $s\mapsto h(s)$ be a smooth $1$-parameter family of metric perturbations with $h(0)=0$, so that  $\Mb[h(s)]$ is a globally hyperbolic
spacetime for all sufficiently small $|s|$. Assuming the relevant derivatives exist (in a suitable topology,
which might, for instance, be a weak topology induced by a state space --- see Section~\ref{sect:states})
we may define a derivation $\delta$ on $\Af(\Mb)$ by 
\[
\delta(A) = -2i \left.\frac{d}{ds}\rce_\Mb[h(s)] A\right|_{s=0}
\]
that depends linearly on $f=\dot{h}(0)$. It is
convenient to denote this by $\delta(\cdot) = [\Ts_\Mb(f), \cdot ]$, 
without any implication that $\Ts_\Mb(f)$ is an element of $\Af(\Mb)$.
Moreover, the right-hand side can be written in functional derivative notation,
leading to the suggestive equation
\[
[\Ts_\Mb(f), A] = \frac{2}{i} \int_\Mb f_{\mu\nu} \frac{\delta\rce_{\Mb}}{\delta g_{\mu\nu}} (A).
\]

Although $\Ts_\Mb(f)$ (or rather, the derivation it represents) has only been defined
for symmetric test tensors $f$, we can extend it to arbitrary smearings by demanding
that it vanish on antisymmetric $f$. Then \eqref{eq:rce_diff_cov}  has an interesting consequence 
\cite{BrFrVe03}. 
Let $X^a$ be a smooth compactly supported vector field, and define
$\psi_s = \exp (sX)$ be the $1$-parameter family of diffeomorphisms
it generates, which act trivially outside a fixed compact set (for $|s|<s_*$, say). 
This induces a $1$-parameter family of metric perturbations $h(s) = \psi(s)_*g-g$, 
with
\[
\dot{h}(0)_{ab} = (\pounds_X g)_{ab} = \nabla_a X_b + \nabla_b X_a. 
\]
By \eqref{eq:rce_diff_cov}, $\rce_\Mb[h(s)] = \rce_\Mb[0]=id_{\Af(\Mb)}$ for all $s$, so 
\[
[\Ts_\Mb(\pounds_X g), A] =-2i \left.\frac{d}{ds}\rce_\Mb[h(s)] A\right|_{s=0} = 0,
\]
which asserts that $\Ts_\Mb$ is conserved, when regarded as a symmetric derivation-valued tensor field.\footnote{It is natural to write $\Ts_\Mb(\pounds_X g)=-2(\nabla\cdot \Ts_\Mb)(X)$,
regarding the divergence in a weak sense.}  If the derivation is inner, i.e., given by the commutator with  
elements $\Ts_\Mb(f)\in\Af(\Mb)$, then we deduce that $\nabla\cdot\Ts_\Mb$ belongs to the centre of $\Af(\Mb)$. 

To investigate the interpretation further, we return to the example of the scalar field. 
Starting from \eqref{eq:rce_scalar}, it is clear that 
\begin{equation}\label{eq:rec_diff1}
\left.\frac{d}{ds}\rce_{\Mb}[sh]\Phi_\Mb(f)\right|_{s=0} 
= \Phi_\Mb( L_\Mb[h] E_\Mb f)
\end{equation}
in any topology for which the derivative exists, 
where  
\[
L_\Mb[h]\phi =-
\left.\frac{d}{ds}  P_{\Mb[sh]}\right|_{s=0}\phi
=\nabla_a \left(h^{ab}\nabla_b\phi\right) - \frac{1}{2}\left(\nabla^a h^b_{\phantom{b}b}\right)\nabla_a\phi.
\]
Equation~\eqref{eq:rec_diff1} is derived  for $f$ supported in $\Mc^+$,
but is actually valid for all $f\in\Df(\Mb)$.\footnote{Decompose
$f=f_0 + P_\Mb f_1$, where $f_0$ is supported in the image of $\Mc^+$, and $f_1\in\Df(\Mb)$. 
As $\Phi_\Mb(f)=\Phi_\Mb(f_0)$, we may apply~\eqref{eq:rec_diff1} to $f_0$ and
then use the fact that $E_\Mb f_0=E_\Mb f$.} The link with the stress-energy tensor is obtained as follows. Working in `unsmeared notation', 
\[
[\Phi_\Mb(x)\Phi_\Mb(x'),\Phi_\Mb(f)] = i(E_\Mb f)(x)\Phi_\Mb(x') +  i(E_\Mb f)(x')\Phi_\Mb(x)
\]
and since renormalisation of the Wick square involves subtracting
a $C$-number from the point-split square and then taking the points back together, 
\[
[\Phi^2_{\Mb,\text{ren}}(x),\Phi_\Mb(f)] = 2i (E_\Mb f)(x)\Phi_\Mb(x)
\]
or, smearing against $k\in\CoinX{\Mb}$,
\[
[\Phi^2_{\Mb,\text{ren}}(k),\Phi_\Mb(f)] = 2i\Phi_\Mb(k E_\Mb f).
\]
Similarly, 
\[
[((\nabla_a\Phi)(\nabla_b\Phi))_{\Mb,\text{ren}}(x),\Phi_\Mb(f)] =
2(i\nabla_{(a}\Phi_\Mb(x) )(\nabla_{b)}E_\Mb f|_x)
\]
or, smearing against a symmetric tensor $h$,
\[
[((\nabla\Phi)(\nabla\Phi))_{\Mb,\text{ren}}(h),\Phi_\Mb(f)] =-
2i\Phi_\Mb(\nabla_a h^{ab}\nabla_{b}E_\Mb f).
\]
Applying these formulae to the stress-energy tensor,  
\[
T_{ab} = (\nabla_a \phi )(\nabla_b \phi) - \frac{1}{2}g_{ab} g^{cd}
(\nabla_c \phi )(\nabla_d \phi)+\frac{1}{2}m^2 g_{ab}\phi^2,
\] 
quantized by point-splitting, a short calculation using Leibniz' rule gives
\[
[T_{\Mb,\text{ren}}(h),\Phi_\Mb(f)]  = 
-2i\Phi_\Mb(L_\Mb[h] E_\Mb f)  +i\Phi_\Mb(h^c_{\phantom{c}c}(\Box+m^2)E_\Mb f) ,
\]
%\begin{align*}
%[T_{\Mb,\text{ren}}(h),\Phi_\Mb(f)]  
%&= -2i \Phi_\Mb(\nabla_a (h^{ab}\nabla_{b}E_\Mb f))
%-\frac{1}{2}(2i) \Phi_\Mb(\nabla_a (g^{ab}h^c_{\phantom{c}c}\nabla_{b}E_\Mb f))
%\\
%&\qquad
%+\frac{1}{2}(2i)\Phi_\Mb(m^2h^c_{\phantom{c}c} E_\Mb f) \\
%&=-2i\Phi_\Mb(L_\Mb[h] E_\Mb f)  +i\Phi_\Mb(h^c_{\phantom{c}c}(\Box+m^2)E_\Mb f) \\
%&= -2i\Phi_\Mb(L_\Mb[h] E_\Mb f) .
%\end{align*}
of which the last term vanishes. 
Comparison with \eqref{eq:rec_diff1} yields the important formula
\[
\left.\frac{d}{ds}\rce_{\Mb}[sh]\Phi_\Mb(f)\right|_{s=0} 
=-\frac{1}{2i} [T_{\Mb,\text{ren}}(h),\Phi_\Mb(f)],
\]
where $T_\Mb$ is the renormalised stress-energy tensor
of the theory.\footnote{This result differs by a sign from that in \cite{BrFrVe03} (and repeated e.g., in~\cite{FewVer:dynloc_theory,FewVer:dynloc2}).
The source of the difference arises on p. 61 of \cite{BrFrVe03}, where the action of an {\em `advanced'} Green function is taken to have support in the {\em causal future} of the source. The sign error does not affect
the results of~\cite{FewVer:dynloc_theory,FewVer:dynloc2}.}

Similar computations for other models \cite{Sanders_dirac:2010,Benini_Masters,Ferguson:2013,FewSchenkel:2014,FewLang:2014a} 
support the view that, in general, 
the functional derivative of the relative Cauchy evolution may be interpreted as a stress-energy
tensor. One is therefore led to regard the relative Cauchy evolution as a proxy for the action.  
This is quite remarkable, because we have not assumed that locally covariant theories are specified in terms of classical actions. It is a striking illustration of the power of the general framework.

\section{States and state spaces}\label{sect:states}

\subsection{States and representations}

While much of the structure of quantum field theory lies in the algebraic relations, particularly the concepts of locality and causality and their relation to covariance,
an important role is played by the states (or, synonymously,
expectation value functionals). In particular, states permit the comparison of the mathematical framework with
experiment, and the interpretation of the formal framework in more concrete, physical terms, leading to an understanding of the meaning of certain observables, the charge structure, field content and the degrees of freedom of a quantum field theory. Moreover, certain aspects which distinguish quantized fields from classical fields, like entanglement, are best understood at the level of states.

The mathematical definition asserts only that a state on a $*$-algebra $\Ac$ is 
a positive, normalized and (suitably) continuous functional $\omega : \Ac \to \mathbb{C}$, interpreted as yielding the expectation values $\omega(A)$ of any observable $A$ (an element $A \in \Ac$ with
$A^* = A$). Here,
positivity means $\omega(A^*A) \ge 0$ for all $A \in \Ac$, while normalization requires $\omega({\bf 1}) = 1$
for the unit element of $\Ac$ (which, by default, is assumed to exist if $\Ac$ is the
algebra of observables of a physical system).
Continuity may be an involved issue. If $\Ac$ is a $C^*$-algebra, however, norm continuity
is already implied by positivity, and furthermore, the existence of a large set of states
is warranted from the outset. This is one of the reasons why, from a mainly mathematical perspective, 
it is very convenient to treat observable algebras as $C^*$-algebras.

However, it is known by several examples that this very general mathematical description of a ``state'' allows for many which can hardly be interpreted as 
physically realistic configurations of a quantum field because their behaviour on 
certain observables of interest (particularly those measuring local quantities of momentum and energy) is too singular
\cite{Wald_qft}.
Thus, suitable regularity properties must be imposed to select physically realistic states to which an interpretation
of the observables can be tied and on which an identification of the field content can be built. In quantum field theory
on Minkowski spacetime, Poincar\'e covariance is instrumental in specifying the vacuum state, starting from which one can proceed with a characterization (or at least selection) of physical states,
but, in curved spacetimes, that is not at hand, and one must devise other criteria. We will address that issue in a while, but first, it is necessary
to introduce some terminology.

We begin with the concept of a Hilbert space representation of a $*$-algebra $\Ac$, denoted by 
 $(\Hc,\pi,\Dc)$ and consisting of a Hilbert space $\Hc$ together with a dense subspace $\Dc$, and a representation
$\pi$ of $\Ac$ by closable operators defined on $\Dc$. Furthermore, it is required that 
$\pi(A) \Dc\subset \Dc$ $(A \in \Ac)$ and that one has
\[
\pi(AB) = \pi(A)\pi(B),\quad \pi(\alpha A + \beta B) = \alpha \pi(A) + \beta \pi(B)
\quad\text{and}\quad \pi(A)^* = \pi(A^*)
\]
holding on $\Dc$ for all $\alpha, \beta \in
\mathbb{C},\ A,B \in \Ac$. 
A further requirement is the continuity of $A \mapsto \pi(A)$ at least weakly with respect to $\Dc$ in the topology of $\Ac$. 
If $\Ac$ is a $C^*$-algebra, the $\pi(A)$ are necessarily bounded operators for $A \in \Ac$ and there is no restriction to assume $\Dc = \Hc$; hence we will simply write
$(\Hc,\pi)$ for a Hilbert space representation of a $C^*$-algebra.

The \emph{folium} of a representation $(\Hc,\pi,\Dc)$ of $\Ac$, denoted by $\Fol(\pi)$, consists of all states
on $\Ac$ which can be written as finite convex sums of states of the form $\omega_\xi(A) = (\xi,\pi(A) \xi)$ where $\xi$ is
any unit vector in $\Dc$. If $\Ac$ is a $C^*$-algebra, 
we can define the $W^*$-\emph{folium} $\Fol^W(\pi)$
of $(\Hc,\pi)$ as the weak closure of $\Fol(\pi)$. The set of states $\Fol^W(\pi)$ is also called the set of  \emph{normal states} on $\Ac$ with respect
to the representation $(\Hc,\pi)$.  

Any state $\omega$ on a $*$-algebra $\Ac$, with suitable continuity properties, determines a unique (up to unitary equivalence) Hilbert space representation of $\Ac$, the \emph{GNS representation}, denoted $(\Hc_\omega,\pi_\omega,\Dc_\omega,\Omega_\omega)$. 
Here, $(\Hc_\omega,\pi_\omega,\Dc_\omega)$ is a Hilbert space representation of $\Ac$ and 
$\Omega_\omega$ is a unit vector in $\Dc_\omega$ such that 
$\omega(A) = (\Omega_\omega,\pi_\omega(A) \Omega_\omega)$ for all
$A \in \Ac$ and $\Dc_\omega=\pi_\omega(A)\Omega_\omega$, implying that $\Omega_\omega$ is a cyclic
vector for the representation. If $\Ac$ is a $C^*$-algebra,
one takes $\Dc_\omega = \Hc_\omega$ as before, denoting the GNS representation
more simply by $(\Hc_\omega,\pi_\omega,\Omega_\omega)$. Furthermore one can assign a folium $\Fol^W(\omega):=
\Fol^W(\pi_\omega)$
to any state $\omega$, i.e.\ the folium of its GNS representation. Correspondingly, one 
calls any state in $\Fol^W(\pi_\omega)$ a \emph{normal state} with respect to $\omega$, or simply a state normal to
$\omega$.

Again for the case of a state $\omega$ on a $C^*$-algebra $\Ac$, we define the 
\emph{induced von Neumann algebra}, $\mathcal{N}_\omega$, in the GNS-representation $(\Hc_\omega,\pi_\omega,\Omega_\omega)$ by
$$ \mathcal{N}_\omega = \pi_\omega(\Ac)'',$$
where the double prime denotes the bi-commutant and coincides with the weak closure of 
$\pi_\omega(\Ac)$ by von Neumann's theorem\footnote{By default, all $*$-algebras here
are unital, i.e.\ they have a unit element for the algebra product.}. This definition is mostly useful when applied to the local induced von Neumann algebras considered below.

Many quantum field theories are specified in terms of their fields and then, states usually arise
from (and are defined by) their $n$-point functions. In some examples, e.g.\ if the algebras 
$\mathscr{A}(\Mb)$ of observables assigned to a spacetime $\Mb$ are $C^*$-algebras, it may occur that the quantum field operators $\phi(f)$
are not contained in $\mathscr{A}(\Mb)$, but arise as objects affiliated to
the induced von Neumann algebra $\mathcal{N}_\omega(\Mb) = \pi_\omega(\mathscr{A}(\Mb))''$ in the GNS representations
of physical states $\omega$.  Affiliation means that bounded functions of the operators $U$ and $|\phi(f)|$
occurring in the polar decomposition $\phi(f)=U|\phi(f)|$ belong to the induced von Neumann algebra. As we wish to avoid a discussion of
precise domain properties for quantum field operators, we will introduce $n$-point functions, or Wightman functions, in a manner that is quite close to the original idea.
So let us suppose that $\Mb$ is an object of $\Loc$ and that $\mathscr{A}(\Mb)$ is a $*$-algebra of observables
assigned to spacetime $\Mb$ (it does not matter here if $\mathscr{A}(\Mb)$ is part of a theory functor $\mathscr{A}$,
or is a $C^*$-algebra). Then, we say that 
a state $\omega$ on $\mathscr{A}(\Mb)$ \emph{possesses affiliated $n$-point functions} 
if there are (1) a $C^\infty$ complex vector bundle $\mathcal{V}_{\Mb}$
having $\Mb$ as base manifold 
together with a fibre-wise complex conjugation $\Gamma$ on $\mathcal{V}_{\Mb}$
(2) a weakly dense (in the weak topology induced by the GNS-representation of $\omega$)
$*$-subalgebra $\Af_0(\Mb)$
and (3) a
sequence $w^\omega_n$ $(n \in \mathbb{N})$ of distributions of positive type on
$C_0^\infty(\mathcal{V}_\Mb^n)$ (the compactly supported $C^\infty$ sections
in $\mathcal{V}_{\Mb}^n$),
such that for any $A \in \Af_0(\Mb)$ there is a sequence
$F^{(n)}_A \in C^\infty_0(\mathcal{V}_{\Mb}^n)$ obeying 
$$ \sum_{| \boldsymbol{n} | \le N} w^{\omega}_{|\boldsymbol{n}|}(F^{(n_1)}_{A_1} \otimes \cdots \otimes F^{(n_k)}_{A_k})
 \underset{N \to \infty}{\longrightarrow} \omega(A_1 \cdots A_k)
 $$
for any finite collection of elements $A_1,\ldots, A_k \in \mathscr{A}_0(\Mb)$.
Here, $\boldsymbol{n} = (n_1,\ldots,n_k)$ is a multi-index and
$|\boldsymbol{n}| = n_1 + \ldots + n_k$; the definition $F^{(0)} \in \mathbb{C}$ and
$w^{\omega}_0(F^{(0)}) = F^{(0)}$ is adopted, and the positive type condition means that
$$ \sum_{m,n = 0}^N w_{n + m}^\omega(F^{(n)} \otimes \Gamma F^{(m)}) \ge 0 $$
holds for any finite selection of $F^{(0)}, \ldots , F^{(N)}$, understanding that $\Gamma$ acts
on each of the $m$ fibre factors, with $\Gamma F^{(0)} = \overline{F^{(0)}}$ (complex conjugation).

These definitions facilitate the introduction of conditions on the states $\omega$ via conditions on the wavefront sets \cite{Hoermander1}
$\WF(w^\omega_n)$ of their affiliated $n$-point functions $w^\omega_n$.  A generalized concept of wavefront set for states $\omega$ may be given without using affiliated $n$-point functions \cite{Verch-WF-AQFT}, but will not be pursued here.

\subsection{States in locally covariant theories}

Let $\Af$ be a locally covariant theory, i.e.\ a functor $\Af:\Loc\to\Alg$. By an $\Af$-\emph{state}, we mean a 
family $(\omega_{\Mb})_{\Mb \in \Loc}$ indexed by the objects in $\Loc$, where each $\omega_{\Mb}$ is a state on $\Af(\Mb)$. 
This is a very general mathematical definition and does not involve, as it stands, 
any regularity criteria selecting physically realistic states. Furthermore, the
definition does not relate the states $\omega_{\Mb_1}$ and $\omega_{\Mb_2}$ on spacetimes $\Mb_1$ and $\Mb_2$, even if 
(parts of) $\Mb_1$ can be embedded into (parts of) $\Mb_2$ by morphisms in $\Loc$.

Given a globally hyperbolic spacetime $\Mb$ admitting a non-trivial group ${\rm Aut}(\Mb)$
of spacetime isometries preserving orientation and time-orientation, it has been remarked before 
that $\Af$ induces a group representation of ${\rm Aut}(\Mb)$ by elements in
${\rm Aut}(\Af(\Mb))$. A state $\omega_{\Mb}$ on $\Af(\Mb)$ is called ${\rm Aut}(\Mb)$-\emph{invariant}  if 
\begin{align}
 \omega_{\Mb} \circ \Af(\psi) = \omega_{\Mb}
\end{align}
for all $\psi \in {\rm Aut}(\Mb)$. To give a concrete example, suppose that $\Mb$ is Minkowski spacetime,
then ${\rm Aut}(\Mb)$ is the proper orthochronous Poincar\'e group, and invariance of a state $\omega_{\Mb}$ is one of the important properties singling out a vacuum state. More generally, if 
a spacetime $\Mb$ is stationary, then there is a subgroup in ${\rm Aut}(\Af(\Mb))$ of time translations,
and invariance of a state $\omega_{\Mb}$ under this subgroup --- i.e.\ time-translation invariance ---
is one of the required properties of a ground state or KMS state. As states of this type are commonly
regarded as physical states of a quantum field, the invariance property of states is one of the features
of states to look for. This prompts the question of whether the concept of invariant state can be generalized to the locally covariant setting. We refer to such a generalization as a \emph{natural state}, and the definition is this: A state $(\omega_{\Mb})_{\Mb \in \Loc}$ for a locally covariant theory 
$\Af$ is called \emph{natural} if
$\omega_{\Nb} \circ \Af(\psi) = \omega_{\Mb}$ whenever $\Mb \overset{\psi}{\longrightarrow} \Nb$
is a morphism in $\Loc$. 

While this seems natural in the sense of reflecting the natural duality between algebras and states, it
turns out to be asking too much: Under additional very mild regularity properties, which are expected to be general features of large sets of states in quantum field theories, and have been proved to hold in many examples and
to be consequences of, e.g., the general Wightman framework on Minkowski
spacetime, there are no natural states for locally covariant theories (Theorem~\ref{thm:nogo}). 

The moral is that one should not expect physical states to be invariant under arbitrary
spacetime embeddings in locally covariant quantum field theories. However, there may be sets of states assigned to
spacetimes which behave invariantly under spacetime embeddings, and in fact, it is desirable to formulate selection
criteria for sets of physical states in a way that such an invariance property is fulfilled.

\subsection{State spaces}

Suppose that $\Ac$ is a $*$-algebra. Then we define a \emph{state space} for $\Ac$ to be a set $\boldsymbol{S}$ of states
on $\Ac$ having the property that $\boldsymbol{S}$ is invariant under operations in $\Ac$ and under forming finite convex
sums, i.e.,  given any $\omega$ in $\boldsymbol{S}$, then the states
$$ \omega'(A) = \sum_{i = 1}^N \lambda_i \frac{\omega(B_i^* A B_i)}{\omega(B_i^*B_i)} $$
are also contained in $\boldsymbol{S}$, for any choice of finitely many $B_i \in \Ac$ (with $\omega(B^*_iB_i) > 0$)
and $\lambda_i > 0$ with $\sum_i \lambda_i = 1$
$(i = 1,\ldots,N\,, \ N \in \mathbb{N})$. 
%%\hl{On general grounds, any state space of a physical system should have these
%% properties. One can always produce classical statistical mixtures of given states which amounts to taking convex sums
%% of states. Moreover, one can use the observables of a systems as operations (particularly, projection-valued observables,
%% which can be approximated in a suitable sense by observables in $\Ac$ in the GNS-representations of states), especially as
%% filtering operations. Therefore, a state space of a physical system should always be closed under operations and under forming
%% convex sums of states. }

Note that the folium of any Hilbert space representation of $\Ac$ is closed in the above sense. 
Thus one can introduce a category $\Stsp$ of
state spaces which is ``dual'' to the category of algebras: The objects in $\Stsp$ are state spaces $\boldsymbol{S}$
of $*$-algebras $\Ac$; more precisely, they are pairs $(\boldsymbol{S},\Ac)$, where
$\Ac$ indicates the $*$-algebra for which $\boldsymbol{S}$ is a state space. Then a morphism $\alpha^*$ in $\Stsp$
is the dual of a suitable  morphism $\alpha$ in $\Alg$. In more detail, if $(\boldsymbol{S}_1,\Ac_1)$ and $(\boldsymbol{S}_2,\Ac_2)$
are objects in $\Stsp$, then any morphism $\alpha : \Ac_1 \to \Ac_2$ in $\Alg$
such that $\alpha^*\boldsymbol{S}_2\subset \boldsymbol{S}_1$  defines
a morphism $(\boldsymbol{S}_2,\Ac_2) \overset{\alpha^*}{\longrightarrow} (\boldsymbol{S}_1,\Ac_1)$ in $\Stsp$. Here, the dual action of $\alpha$ 
on states is given by $\alpha^*\omega_2 = \omega_2 \circ \alpha$ for $\omega_2\in \boldsymbol{S}_2$.  Thus (with some abuse of notation that is unlikely to give rise to ambiguities)
$\alpha^*(\boldsymbol{S}_2,\Ac_2) = (\alpha^*(\boldsymbol{S}_2),\Ac_1)$. 
The morphism composition rule is defined as the composition rule
of positive dual, convex maps between dual spaces of $*$-algebras. 

Now let $\Af$ be a locally covariant theory. Then we define a \emph{state space for} $\Af$ 
(henceforth, $\Af$-state space) to be a 
contravariant functor $\mathscr{S}$ between $\Loc$ and $\Stsp$, such that $\mathscr{S}(\Mb)$ is a state space for
$\Af(\Mb)$ for any object $\Mb$ of $\Loc$, and such that, if $\psi$ is a morphism in $\Loc$, then $\mathscr{S}(\psi)$ is induced by $\Af(\psi)^*$, the dual map of
$\Af(\psi)$.\footnote{Note that $\Sf$, or its opposite covariant functor $\Sf^{\text{op}}:\Loc\to\Stsp^{\text{op}}$,
contains all the information in $\Af$, and could be used by itself to
specify the theory in full -- this is done e.g., in~\cite{Fewster:gauge,Few_split:2015}.}
The state space is said to obey the timeslice condition if $\Sf(\Mb)=\Sf(\psi)(\Sf(\Nb))$ for every Cauchy morphism $\psi:\Mb\to\Nb$.

\subsection{Conditions on states}

The definition of ``state space'' just given is very general, and to
ensure that the states contained in the state space of a locally covariant theory
is formed by states which can be given a reasonable and consistent physical interpretation, 
it needs to be supplemented with further conditions.
The conditions are expressed as conditions
which the states of a state space fulfil individually, or in relation to each other. Therefore we shall list
some of them; they correspond to regularity properties one expects physical states to have for locally covariant
theories and are motivated either by examples featuring such properties (in curved or flat spacetime) or by structural
arguments in favour of such properties in general quantum field theory (in flat spacetime). We refer in particular
to \cite{Haag} for discussion. 

The $*$-algebra $\mathscr{A}(\Mb)$
may, but need not, derive from a functor $\mathscr{A}$ from $\Loc$ to $\Alg$. However, we assume that $\mathscr{A}(\Mb)$ is generated by a system
of $*$-subalgebras ($C^*$-subalgebras if $\mathscr{A}(\Mb)$ is a $C^*$-algebra) $\mathscr{A}(\Mb;O)$, $O \in \OO(\Mb)$ fulfilling
isotony. 
\\[4pt]
{\bf Reeh-Schlieder property}. Let $\omega$ be a state on $\Af(\Mf)$ and let $O \in \OO(\Mb)$. Then one says that
$\omega$ has the \emph{Reeh-Schlieder property} with respect to $O$ if, in 
GNS-representation $(\Hc_\omega,\pi_\omega,\Dc_\omega,\Omega_\omega)$ of any $\omega \in \boldsymbol{S}(\Mb)$, the
set of vectors
$$ \pi_\omega(\Af(\Mb;O))\Omega_\omega \quad \text{is dense in} \quad \Hc_\omega\,.$$
{\bf Split property}. 
This property is conveniently formulated under the 
assumption that the $\Af(\Mb;O)$ are $C^*$-algebras, although
a definition can also be given in more general cases.
Let $O_1$ and $O_2$ be two spacetime regions in $\OO(\Mb)$ so that $\overline{O_1} \subset O_2$.
A state $\omega$ on $\Af(\Mb)$ is said to fulfil the \emph{split property} for the pair $O_1$ and $O_2$ of regions if
there is some type I factor von Neumann subalgebra
$\mathcal{N}$ of ${\sf B}(\Hc_\omega)$ so that
$$ \mathcal{N}_\omega(O_1) \subset \mathcal{N} \subset \mathcal{N}_\omega(O_2) \,.$$
Here, $\mathcal{N}_{\omega}(\Mb;O) = \pi_{\omega}(\Af(\Mb;O))''$ are the local von Neumann algebras induced in the
GNS representations. At this point we recall that $\mathcal{N}$ is a factor if $\mathcal{N} \cap \mathcal{N}' = \mathbb{C} {\bf 1}$, i.e.\
if only multiples of the identity operator are contained in both $\mathcal{N}$ and its commutant. A von Neumann algebra $\mathcal{N}$ is of type I if
there is a von Neumann algebra isomorphism $\gamma : \mathcal{N} \to {\sf B}(\tilde{\mathcal{H}})$ where $\tilde{\mathcal{H}}$ is some (possibly inseparable) Hilbert space.
\\[4pt]
{\bf Intermediate factoriality}. A state $\omega$ on $\Af(\Mb)$ will be defined as having the property of \emph{intermediate factoriality} if for any 
$O \in \OO(\Mb)$ there are some $\tilde{O} \in \OO(\Mb)$ and some factor von Neumann subalgebra $\mathcal{N}$ of ${\sf B}(\Hc)$ such that
$$ \mathcal{N}_\omega(O) \subset \mathcal{N} \subset \mathcal{N}_\omega(\tilde{O}) \,.$$
We note that, while this is technically reminiscent of the split property, the condition here is different, and it has a different purpose --- as a consequence
of intermediate factoriality, the GNS representations $\pi_{\omega |_{\Af(\Mb;O)}}$ and $\pi_{\omega} |_{\Af(\Mb;O)}$ are quasiequivalent, i.e.\ the have
the same folia. (For a fuller discussion  and proofs, see \cite{BrFrVe03}.)
\\[4pt]
{\bf Primarity}. A state $\omega$  fulfils the condition of \emph{primarity} with respect to some region $O \in \OO(\Mb)$ if $\mathcal{N}_\omega(O)$ is
a factor. An immediate consequence is this: If $\omega$ fulfils primarity for a subset of regions $O \in \OO(\Mb)$ such that any relatively compact subset of
$\Mb$ is contained in some such $O$, then $\omega$ satisfies intermediate factoriality. 
\\[4pt]
{\bf Duality}. The condition of \emph{duality} of $\omega$ with respect to some region $O \in \OO(\Mb)$ requires that
$$ \mathcal{N}_\omega(O') = \mathcal{N}_\omega(O)' $$
where $O' = \Mb \backslash \overline{J_\Mb(O)}$  
 is the open causal complement of $O$ in $\Mb$ and
$\mathcal{N}_\omega(O')$ is the von Neumann algebra generated by all $\mathcal{N}_\omega(O_\times)$, for relatively compact $O_\times \subset O'$.
\\[4pt]
{\bf Local quasiequivalence}. This condition is best formulated under the assumption that the $\Af(\Mb;O)$ are
$C^*$-algebras. A set of states $\boldsymbol{S}_0(\Mb)$ is said to fulfil \emph{local quasiequivalence} if
for any pair of states $\omega_1,\omega_2 \in \boldsymbol{S}_0(\Mb)$ the equality
\begin{align} \label{eqn:locquasieq} 
\Fol(\pi_{\omega_1} |_{\Af(\Mb;O)}) = \Fol(\pi_{\omega_2} |_{\Af(\Mb;O)}) 
\end{align}
holds for all $O \in \OO(\Mb)$, i.e.\ if the GNS-representations of the states have the same folia when restricted to
local algebras $\Af(\Mb;O)$.

Note that  local quasiequivalence given in this form is equivalent to the condition 
$$ \Fol(\pi_{\omega_1 | \Af(\Mb;O)}) = \Fol(\pi_{\omega_2 |\Af(\Mb;O)})\,, $$
stating that the folia of the GNS-representations of states restricted to the local algebras coincide, once
the states fulfil intermediate factoriality or the Reeh-Schlieder property. (Again, we refer to \cite{BrFrVe03} for further
discussion.)
\\[4pt]
{\bf Triviality of local von Neumann algebras over points}. Once more, this condition assumes that the $\Af(\Mb;O)$
are $C^*$-algebras. We say that a state $\omega$ is \emph{point-trivial} if 
$$ \bigcap_{O \owns p} \mathcal{N}_\omega(\Mb;O) = \mathbb{C} {\bf 1} $$
for any $p \in \Mb$. This says that the induced local von Neumann algebras induced by $\omega$ contain
only multiples of the identity operator if their localization regions are shrunk to any point in spacetime.
\\[4pt] 
{\bf Scaling limits}.
For simplicity of notation, we will introduce this concept for  
$n$-point functions $w^\omega_n$  of scalar type, i.e.\ $\mathcal{V}_\Mb\cong\mathbb{C}$.
The generalization to higher-dimensional vector bundles is not difficult (see, e.g., \cite{SahlmannVerch:2000RMP}).
With this assumption,
let $\omega$ be a state on $\Af(\Mb)$ with affiliated $n$-point functions $w_n^\omega$, and let $x$ be a point in $\Mb$.
 With the help of the exponential map ${\rm exp}_x$ at $x$,
$T_x\Mb$ can be identified with Minkowski spacetime (of suitable dimension). 
If $f \in C_0^\infty(T_x\Mb)$, we define 
$$ f^{[\lambda]}({\rm exp}_x(x')) = f( \lambda^{-1}x')\,, \quad \lambda > 0 \,,$$
for $x'$ in a neighbourhood of the origin in $T_x\Mb$ which is contained in the domain of the exponential map. This way,
the functions $f^{[\lambda]}$ are defined and $C_0^\infty$ on an open neighbourhood of $x$ in $\Mb$.
Then one says that the state $\omega$ has a \emph{regular scaling limit} at $x$ if (i) the state is point-trivial (at $x$) and (ii)
if there is a monotonous function $\nu(\lambda) > 0$ of the scaling parameter $\lambda > 0$ such that the limits
$$ w_n^0(f_1 \otimes \cdots \otimes f_n) = \lim_{\lambda \to 0}\, \nu(\lambda)^n w_n^\omega(f^{[\lambda]}_1 \otimes \cdots \otimes
 f^{[\lambda]}_n) $$
exist for all $f_j \in C_0^\infty(T_x\Mb)$ and if the $n$-point distributions thus obtained satisfy the Streater-Wightman
axioms (in a nontrivial manner).
%, i.e.\ for a non-trivial Wightman-type quantum field theory).
\\[4pt]
{\bf Wavefront set spectrum condition, or microlocal spectrum condition ($\boldsymbol{\mu}$SC)}. If a state $\omega$ has affiliated 
$n$-point functions $w_n^\omega$, a \emph{microlocal spectrum condition},
abbreviated $\mu$SC, is a condition on the wavefront sets $\WF(w_n^\omega)$. We shall not pause here to give the definition
of the wavefront sets of distributions defined on $C_0^\infty$ sections in vector bundles as this is well-explained elsewhere
\cite{SahlmannVerch:2000RMP}, nor shall we record the precise form of the $\mu$SC which has been given in 
\cite{BruFreKoe,BrFr2000}. The $\mu$SC can be seen as a microlocal remnant of the spectrum condition imposed on $n$-point functions in
the Streater-Wightman approach to quantum field theory on Minkowski spacetime \cite{StreaterWightman}. One of the main features
of the $\mu$SC is that it is manifestly covariant (provided the vector bundle $\mathcal{V}_{\Mb}$ 
connects appropriately to the functorial structure of $\Loc$) and this is at the heart of the considerable advances which quantum field
theory in curved spacetimes has seen since the introduction of the $\mu$SC. A certain asymmetry under exchange of the 
order of entries in $\WF(w_n^\omega)$ is characteristic of the $\mu$SC. At the level of the 2-point function $w_2^\omega$ of
a state, the $\mu$SC requires  
$$\WF(w_2^\omega) \subset \{ (x,\xi;x',\xi') \in T^*M \times T^*M : (x,\xi) \sim (x',-\xi')\,, \xi \vartriangleright 0 \} $$
where $(x,\xi) \sim (x',-\xi')$ means that the manifold base-points $x$ and $x'$ are connected by a lightlike geodesic and
that $\xi$ and $-\xi'$ are co-parallel to that geodesic, and $-\xi'$ is the parallel transport of $\xi$ along the connecting
geodesic. The relation $\xi \vartriangleright 0$ means that $\xi$ is future-pointing with respect to the time-orientation on
$\Mb$.\footnote{Our convention on Fourier transforms of compactly supported distributions is (in Minkowski space) 
$\hat{u}(k)=u(e_k)$, where $e_k(x)=e^{ik_\mu x^\mu}$; 
this is extended to manifolds using coordinate charts.}
\\[4pt]
The conditions listed above are fulfilled in several models of quantum fields on curved spacetimes which fulfil
linear hyperbolic field equations, and are quantized imposing canonical commutation relations (CCRs) in the case
of integer spin fields, or canonical anti-commutation relations (CARs) in the case of half-integer spin fields,
when choosing as set of states $\boldsymbol{S}_0(\Mb)$ the set of quasifree Hadamard states. Hadamard states are
specified by a particular form of the two-point function
 \cite{Wald:1977-backre,KayWald-PhysRep}.  

The most complete investigation in this respect has been carried out for the minimally coupled Klein-Gordon
field. Even though the conditions are partially inspired by the behaviour of linear quantum field models,
many of them are viewed as being valid also for interacting quantum fields and required for a consistent
interpretation of the theory. We collect results and references below. 
 
As just mentioned, the Hadamard condition for linear quantum fields on curved spacetimes, which requires that
the two-point function $w^\omega_2$ of a state $\omega$ takes the Hadamard form where the singular part of
$w^\omega_2$ is determined by the spacetime metric and the field equation 
implies the microlocal spectrum condition. In fact, as was first shown in a seminal
paper by Radzikowski for the quantized minimally coupled Klein-Gordon field, for quasifree states the
Hadamard condition and the microlocal spectrum condition $\mu$SC are equivalent \cite{Radzikowski_ulocal1996}.
This was also shown to hold for other models of linear quantized fields on curved spacetimes \cite{SahlmannVerch:2000RMP}. 

When settling for some choice of sets of states $\boldsymbol{S}_0(\Mb)$ for any $\Mb$ in $\Loc$ for
a given locally covariant theory $\Af$, where $\boldsymbol{S}_0(\Mb)$ satisfies some, or even all, of
the above stated conditions, one can obtain an $\Af$-state space $\mathscr{S}$ in the following manner
\cite{BrFrVe03}:
First, one must check if the $\boldsymbol{S}_0(\Mb)$, $\Mb \in \Loc$, transform contravariantly under the dualized
morphisms of $\Af$, which means $\Af(\psi)^*\boldsymbol{S}_0(\Nb) \subset \boldsymbol{S}_0(\Mb)$ whenever
$\psi : \Mb \to \Nb$ is a morphism in $\Loc$, with equality holding in case that $\psi(\Mb) = \Nb$. 
Next, one ought to check if the sets of states $\boldsymbol{S}_0(\Mb)$ fulfil the condition of
local quasiequivalence for all $\Mb \in \Loc$, as well as intermediate factoriality
--- to facilitate the discussion, we will assume from now on that
the $\Af(\Mb)$ are $C^*$ algebras. If that is the case, one can augment the sets of states as
\begin{align} \label{eqn:def-of-statespace}
 \boldsymbol{S}(\Mb) = \{ \omega \ \text{is state on} \ \Af(\Mb) : \omega|_{\Af(\Mb;O)} \in
 Fol^W(\pi_{\omega_0} |_{\Af(\Mb;O)})\,\} 
\end{align}
which is to hold for all $O \in \OO(\Mb)$ and any $\omega_0 \in \boldsymbol{S}_0(\Mb)$. Note that,
in view of the assumed local quasiequivalence and intermediate factoriality of the $\boldsymbol{S}_0(\Mb)$, the
definition of $\boldsymbol{S}(\Mb)$ is independent of the choice of $\omega_0 \in \boldsymbol{S}_0(\Mb)$ in
\eqref{eqn:def-of-statespace}. 

Then, setting $\mathscr{S}(\Mb) = \boldsymbol{S}(\Mb)$ for objects $\Mb$ of $\Loc$ and
$\mathscr{S}(\psi) = \Af(\psi)^*$ for morphisms $\psi$ of $\Loc$ yields a state space for $\Af$ which inherits several
of the properties featuring in the $\boldsymbol{S}_0(\Mb)$. More precisely, one finds the following
statement, which, as mentioned, assumes that the $\boldsymbol{S}_0(\Mb)$ transform contravariantly under the
dualized morphisms of $\Af$.

\begin{theorem} \label{thm:state-space}
Suppose that for any object $\Mb$ of $\Loc$, the set of states $\boldsymbol{S}_0(\Mb)$ satisfies local quasiequivalence and 
intermediate factoriality. Then $\mathscr{S}$ is an $\Af$-state space, thus any $\mathscr{S}(\Mb)$ is closed under
operations of $\Af(\Mb)$ and under forming convex sums, and consists, locally, of a single folium (so locally, i.e.\
in restriction to $\Af(\Mb;O)$ for any $O \in \OO(\Mb)$,
all states in $\mathscr{S}(\Mb)$ are normal to any/all states in $\boldsymbol{S}_0(\Mb)$). Moreover, if the states
in $\boldsymbol{S}_0(\Mb)$ are also point-trivial, the same holds for all the 
states in $\mathscr{S}(\Mb)$, for any object $\Mb$ of $\Loc$.
\end{theorem}

The proof of this statement (with slight variations) be found in \cite{BrFrVe03}, where it is referred to
as \emph{principle of local definiteness}, as put forward initially by Haag, Narnhofer and Stein
\cite{HaNaSt:1984}. We mention that it is also important
that locally, the state space coincides with a single folium of states and therefore, is minimal, as this rules out 
the occurrence of local superselection
rules, akin to charges which may sit somewhere locally, but cannot be moved around by any device. Thus,
in this sense, the state space, at least formally, captures the idea that a replacement for a vacuum state should be
a set of states which are in a formal sense vacuum-like, meaning that they have a low particle density, temperature, and
stress-energy density. Of course, these properties are, on generic spacetimes, only approximately realized, and will in general
only have a `relative' meaning, e.g.\ compared to local curvature quantities.
In the situation described here, where the $\Af$-state space of a locally covariant quantum field theory $\Af$ consists locally of a single folium,
one has a situation very similar to quantum field theory on Minkowski spacetime where locally
(in restriction to algebras $\Af(\Mb;O)$ with $O \in \OO(\Mb)$), the state space consists of states
which are normal to the vacuum state. This is the starting point for the theory of superselection charges, at least of
localizable and transportable charges which are represented by (equivalence classes of) states which are normal to the
vacuum state on algebras $\Af(\Mb;O')$ if the spacetime region $O'$ is the causal complement of
a double cone (for Minkowski spacetime, $O'$ is not relatively compact), but are not normal to the vacuum state on
the full spacetime algebra $\Af(\Mb)$. See \cite{BaumWollen:1992,Haag} and literature cited there for an exposition
of superselection theory, and background material.
To some extent, the theory of localized, transportable superselection charges can be generalized to quantum field theory in
curved spacetime \cite{GLRV}, and also to locally covariant quantum field theory \cite{Ruzzi:2005,Br&Ru05,BrunettiRuzzi_topsect}.
However, the case of non-localized superselection charges is more complicated in curved spacetimes since the topology of
the Cauchy-surface of the spacetime under consideration may play an important role regarding the existence or non-existence of 
certain types of charges. Problems of that nature occur already when trying to obtain a locally covariant setting for free
quantum electrodynamics at the field algebra level. However, we shall not pursue this circle of problems any further at this
point, and instead refer to the literature \cite{BeniniDappiaggiSchenkel:2013,SandDappHack:2012,FewSchenkel:2014}.

There is an assertion about the type of the local von Neumann algebras $\mathcal{N}_\omega(O)$ which is implied if the state
$\omega$ has a regular scaling limit at a point in the spacelike boundary of $O$ (which means that $O$ must have a non-trivial
causal complement), together with causality of the local algebras. The statement has been given in the curved spacetime context
in \cite{Wollenb,BaumWollen:1992,Verch:SL-CST,FewsterVerch-NecHad}; it builds on a seminal paper by Fredenhagen \cite{Fredenhagen:1985}.
We rephrase it here as follows.

\begin{theorem} \label{thm:typeIII}
 Suppose that the state $\omega$ on $\Af(\Mb)$ possesses a regular scaling limit at some point $x$ lying in the spacelike boundary of
$O \in \OO(\Mb)$, that $O$ has a non-trivial (open) spacelike complement $O'$, and that $\mathcal{N}_\omega(O)$ and $\mathcal{N}_\omega(O')$
are pairwise commuting von Neumann algebras. Then  $\mathcal{N}_\omega(O)$ is of type ${\rm III}_1$.
\end{theorem}
We remark that this result holds also if $O$ is not relatively compact, provided the other conditions are met. The type ${\rm III}_1$ property of
$\mathcal{N}_\omega(O)$ means, roughly speaking, that $\mathcal{N}_\omega(O)$ contains no (non-zero) finite-dimensional projections. For the precise
mathematical statement, see \cite{Blackadar}. Suffice it to mention that the type ${\rm III}_1$ property of local algebras of von Neumann
algebras is a typical feature of local (von Neumann) algebras of observables in relativistic quantum field theory which does not appear in
quantum mechanics, or quantum statistical mechanics. There are some interesting consequences --- in particular, like the Reeh-Schlieder theorem
to be discussed below, the type ${\rm III}_1$ property of the local von Neumann algebras has as one of its consequences the ubiquity of states which
are entangled across acausally separated spacetime regions. The reader is referred to the references \cite{CliftonHalvorson,SummersWerner,VerchWerner,BuchholzStoermer}
for further material related to
that theme. 

Next, we shall compile what is known about states of linear quantum field models from the point of view of locally covariant quantum
field theory, providing examples for the properties of states listed above, and for Thms.\ \ref{thm:state-space} and \ref{thm:typeIII}.

\begin{proposition} ${}$ Assume that $\Af$ is the $C^*$-algebra version of locally covariant quantum field theory of the quantized linear Klein-Gordon field with a
 general curvature coupling, corresponding to the field equation $(\Box_\Mb + \xi R_\Mb + m^2)\phi = 0$, where $R_\Mb$ is the scalar
 curvature of the underlying spacetime $\Mb$ of $\Loc$, and $m^2 \ge 0$ and $\xi \ge 0$ are fixed constants (the same for all
$\Mb$). (That is, $\Af(\Mb)$  is the Weyl algebra of the 
quantized Klein--Gordon field on each $\Mb\in\Loc$.)
Then the following hold:
\begin{itemize}
\item[{\rm (1)}] ${}$ \
Any quasifree Hadamard state on $\Af(\Mb)$ fulfils
point-triviality, intermediate factoriality, existence of affiliated $n$-point functions, the $\mu$SC,
and for a certain class of spacetime regions: Split-property, primarity, and Haag-duality \cite{Verch97}.  
\item[{\rm (2)}] ${}$ \ 
Any two quasifree Hadamard states $\omega_1$ and $\omega_2$ on $\Af(\Mb)$ are locally quasiequivalent, i.e.\ the
condition \eqref{eqn:locquasieq} holds for any $O \in \OO(\Mb)$ \cite{Verch:1994}. 
\item[{\rm (3)}] ${}$ \
A quasifree Hadamard state $\omega$ on $\Af(\Mb)$ fulfils the Reeh-Schlieder property with respect to any spacetime region
$O \in \OO(\Mb)$ if the two-point function $w_2^\omega$ fulfils the {\em analytic microlocal spectrum condition} \cite{StVeWo:2002}. Without assuming the analytic microlocal spectrum condition, there are also
spacetime regions $O \in \OO(\Mb)$ and quasifree Hadamard states $\omega$ on $\Af(\Mb)$ such that $\omega$ has the Reeh-Schlieder property
with respect to $O$ \cite{Verch:1993,Sanders_ReehSchlieder,Stroh:2000}.
\item[{\rm (4)}] ${}$ \
Setting $\boldsymbol{S}_0(\Mb)$ to coincide with the set of all quasifree Hadamard states in the case of the locally covariant
quantized Klein-Gordon field, the assumptions stated for Thm.\ \ref{thm:state-space} are fulfilled \cite{BrFrVe03}.
\end{itemize}
\end{proposition}
We remark that similar results have also been obtained
for the quantized Dirac, Proca, and (partially) electromagnetic fields, choosing in each case the
set of quasifree Hadamard states as the set $\boldsymbol{S}_0(\Mb)$ \cite{DAnHol:2006,SahlmannVerch:2000RMP}.

\section{Spacetime deformation and the rigidity argument}

Techniques based on deformations of globally hyperbolic spacetimes
go back to the work of Fulling, Narcowich and Wald~\cite{FullingNarcowichWald} in which the existence
of Hadamard states on ultrastatic spacetimes was used to deduce
their existence on general globally hyperbolic spacetimes. 
As first recognised in~\cite{Verch01}, the same idea can be used to great effect in 
locally covariant QFT.

\subsection{Spacetime deformation}

There are two basic components to the spacetime deformation construction:
the existence of a standard form for globally hyperbolic spacetimes, 
and the actual deformation procedure.
Consider any object $\Mb= (\Mc,g,\ogth,\tgth)$ of $\Loc$. 
An important property of globally hyperbolic spacetimes is that
$\Mb$ admits foliations into smooth spacelike Cauchy surfaces.
Moreover, every spacelike Cauchy surface $\Sigma$ of $\Mb\in\Loc$ also carries an 
orientation $\wgth$ fixed by the requirement that $\tgth\wedge\wgth$ is the
restriction of $\ogth$ to $\Sigma$,\footnote{Recall that $\tgth$, $\ogth$
and $\wgth$ are all regarded as connected components of certain sets of
nowhere zero forms; by $\tgth\wedge\wgth$ we denote the
set of all possible exterior products from within $\tgth$ and $\wgth$.}
and all such oriented Cauchy surfaces are oriented-diffeomorphic (i.e., diffeomorphic via
an orientation-preserving map). 
These facts may be used to prove the following structure theorem for $\Loc$ (see \cite[\S 2.1]{FewVer:dynloc_theory}).
\begin{proposition}\label{prop:BS}
Supposing that $\Mb\in\Loc$, let $\Sigma$ be a 
smooth spacelike Cauchy surface of $\Mb$ with induced orientation $\wgth$, and let $t_*\in\RR$. Then there
is a $\Loc$-object  $\Mb_{\text{st}}=(\RR\times\Sigma,g, \tgth\wedge\wgth,\tgth)$
and an isomorphism $\rho:\Mb_{\text{st}}\to\Mb$ in $\Loc$ such that
\begin{itemize}
\item the metric $g$ is of split form 
\begin{equation}\label{eq:split_metric}
g = \beta  dt\otimes dt -h_t
\end{equation}
where  $t$ is the coordinate corresponding to the first factor of 
the Cartesian product $\RR\times\Sigma$, the function $\beta\in C^\infty(\RR\times\Sigma)$ is strictly positive and $t\mapsto h_t$ is a smooth choice of (smooth) Riemannian metrics on $\Sigma$; 
\item each $\{t\}\times\Sigma$ is a smooth spacelike Cauchy surface, and $\rho(t_*,\cdot)$ is the inclusion of $\Sigma$ in $\Mb$;
\item the vector $\partial/\partial t$ is future-directed according to $\tgth$.
\end{itemize}  
We refer to $\Mb_{\text{st}}$ as a \emph{standard form} for $\Mb$.
\end{proposition}
This statement  is a slight elaboration of results due to 
Bernal and S\'anchez (see particularly, \cite[Thm 1.2]{Bernal:2005qf} and \cite[Thm 2.4]{Bernal:2004gm}), which were previously long-standing folk-theorems.  
The main deformation result can now be stated (see~\cite[Prop.~2.4]{FewVer:dynloc_theory}):
\begin{proposition}\label{prop:Cauchy_chain}
Spacetimes $\Mb$, $\Nb$ in $\Loc$ have oriented-diffeomorphic Cauchy surfaces
if and only if there is a chain of Cauchy morphisms in $\Loc$ forming a diagram
\begin{equation} \label{eq:Cauchy_chain}
\Mb\xleftarrow{\alpha} \Pb \xrightarrow{\beta} \Ib \xleftarrow{\gamma} \Fb \xrightarrow{\delta} \Nb.
\end{equation} 
\end{proposition}

The importance of the deformation result is that a locally covariant theory $\Af$ obeying
the time-slice condition maps every Cauchy morphism of $\Loc$ to an isomorphism of $\Alg$,
so the chain of Cauchy morphisms in \eqref{eq:Cauchy_chain} induces an
isomorphism 
\begin{equation}
\Af(\delta)\circ\Af(\gamma)^{-1}\circ\Af(\beta)\circ
\Af(\alpha)^{-1}:\Af(\Mb)\to\Af(\Nb) .
\end{equation}
This isomorphism is not canonical, owing to the many choices used to construct it. 
Nonetheless, we will see that it is possible to use results of this type to transfer information and
structures between the instantiations of the theory on $\Mb$ and $\Nb$.

In the following, a few more definitions will be needed (\cite[Def.~2.5]{FewVer:dynloc_theory}). 
A \emph{Cauchy ball} in a Cauchy surface $\Sigma$  
of $\Mb\in\Loc$ is a subset $B\subset\Sigma$ for which there is a chart $(U,\phi)$ of $\Sigma$ such that $\phi(B)$ a nonempty open ball in $\RR^{n-1}$ whose closure is contained in $\phi(U)$. A {\em diamond} in $\Mb$ is any open relatively compact subset of the form $D_\Mb(B)$,
where $B$ is a Cauchy ball. The diamond is said to have \emph{base} $B$ and be \emph{based on} any Cauchy surface in which $B$ is a Cauchy ball. 
A {\em multi-diamond} $D$ is a union of finitely many causally disjoint diamonds; there exists
a common Cauchy surface on which each component is based, and the intersection of $D$
with an open causally convex neighbourhood of any such Cauchy surface 
is called a \emph{truncated (multi)-diamond}.

\subsection{The rigidity argument and some applications}\label{sect:rigidity}

It is often the case that if a locally covariant QFT has a property in one spacetime, 
then the same is true in all spacetimes, a phenomenon that we 
call \emph{rigidity}. This testifies to the strength of
the hypotheses given in section~\ref{sect:assumptions}, particularly the 
timeslice property. A simple example is provided by Einstein causality, 
which was originally included as Assumption~\ref{ax:Einstein_causality} for a locally covariant quantum field theory. 
However, this assumption is largely redundant:
Provided a theory is Einstein causal in one 
spacetime (e.g., Minkowski), it must be so in all spacetimes. 

Let $\Af:\Loc\to\Alg$ obey local covariance and the timeslice condition, but not necessarily
Einstein causality, and for each $\Mb\in\Loc$ let $\OO^{(2)}(\Mb)$ denote the
set of all ordered pairs $\langle O_1,O_2\rangle\in \OO(\Mb)\times\OO(\Mb)$ such that
the $O_i$ are causally disjoint in the sense that $O_1\subset O_2'$. 
For any such pair $\langle O_1,O_2\rangle\in\OO^{(2)}(\Mb)$ let $P_\Mb(O_1,O_2)$ be the proposition that Einstein causality holds for $O_1$ and $O_2$, i.e., that $\Af^\kin(\Mb;O_1)$ and $\Af^\kin(\Mb;O_2)$ commute.  These propositions have some simple properties:  
\begin{description}%[R2]
\item[R1] for all $\langle O_1,O_2\rangle\in\OO^{(2)}(\Mb)$,  
\[
P_\Mb(O_1,O_2) \iff P_\Mb(D_\Mb(O_1),D_\Mb(O_2)).
\] 
\item[R2] given $\psi:\Mb\to\Nb$ then, for all $\langle O_1,O_2\rangle\in\OO^{(2)}(\Mb)$, \[P_\Mb(O_1,O_2) \iff P_\Nb(\psi(O_1),\psi(O_2)).
\]
\item[R3] for all $\langle O_1,O_2\rangle\in\OO^{(2)}(\Mb)$
and all $\widetilde{O}_i\in\OO(\Mb)$ with $\widetilde{O}_i\subset O_i$ ($i=1,2$)
\[
P_\Mb(O_1,O_2) \implies P_\Mb(\widetilde{O}_1,\widetilde{O}_2) .
\]
\end{description}
Here, R1 holds trivially because $\Af(\Mb;O)=\Af(\Mb;D_\Mb(O))$,
while to prove R2 we recall from \eqref{eq:kin_covariance} that $\Af^\kin(\Nb;\psi(O_i))=\Af(\psi)(\Af^\kin(\Mb;O_i))$, 
and use the equality
\[
[\Af(\Nb;\psi(O_1)),\Af(\Nb;\psi(O_2))] = \Af(\psi)([\Af(\Mb;O_1),\Af(\Mb;O_2)])
\]
together with injectivity of $\Af(\psi)$. R3 is also trivial as $\Af^\kin(\Mb;\widetilde{O}_i)\subset\Af^\kin(\Mb;O_i)$.

These facts allow us to prove the following. 

\begin{theorem} \label{thm:causality}
Let the theory $\Af:\Loc\to\Alg$ obey local covariance and timeslice (Assumptions~\ref{ax:loc_cov} and~\ref{ax:timeslice}) 
but not necessarily Einstein causality (Assumption~\ref{ax:Einstein_causality}).
Suppose that, in the theory $\Af$, Einstein causality holds for some 
pair of causally disjoint regions $O_1,O_2\in\OO(\Mb)$ in some spacetime $\Mb\in\Loc$. Then Einstein causality holds in theory $\Af$ for every pair of causally disjoint regions $\widetilde{O}_1,\widetilde{O}_2\in\OO(\widetilde{\Mb})$ in every spacetime $\widetilde{\Mb}\in\Loc$ for which either of the following hold:
\item (a) the Cauchy surfaces of $\widetilde{O}_i$ are oriented diffeomorphic to those of $O_i$ for $i=1,2$;
\item (b) each component of $\widetilde{O}_1\cup \widetilde{O}_2$ has Cauchy surface topology $\RR^3$ (e.g., if the $\widetilde{O}_i$ are truncated multi-diamonds.) 
\end{theorem}
\begin{remark}
The regions $O_i$ in the hypotheses might be  
a finite spacelike distance from one another. 
However, the regions $\widetilde{O}_i$ need not be separated in this way and could touch at their boundaries, or even link around each other if they have nontrivial topology. 
For example,  consider a theory which obeys Einstein causality for 
a pair of causally disjoint diamonds based on the $t=0$
hyperplane of Minkowski space $\Mb_0$. Within each diamond, 
choose a subregion $D_{\Mb_0}(T_i)$, where
$T_i$ is an open subset of the $t=0$ hyperplane with topology $\RR^{n-2}\times\TT$
(e.g., a thickened closed curve). Einstein causality holds 
for these regions (by R3) and thus holds for every pair of causally disjoint regions 
$D_{\Mb_0}(\widetilde{T}_i)$, where the $\widetilde{T}_i$ have topology $\RR^{n-2}\times\TT$, even if they
are linked through one another. 
\end{remark}
\begin{proof}\smartqed
(a) By Proposition~\ref{prop:Cauchy_chain} there is a chain of morphisms 
\[
\widetilde{\Mb} \xlongleftarrow{\widetilde{\iota}} \widetilde{\Mb}|_{\widetilde{O}_1\cup\widetilde{O}_2}
\xlongleftarrow{\widetilde{\psi}} \widetilde{\Lb} \xlongrightarrow{\widetilde{\varphi}} \Ib \xlongleftarrow{\varphi} \Lb 
\xlongrightarrow{\psi} \Mb|_{O_1\cup O_2}\xlongrightarrow{\iota} \Mb
\]
where $\psi,\widetilde{\psi},\varphi,\widetilde{\varphi}$ are Cauchy morphisms and $\iota=\iota_{\Mb;O_1\cup O_2}$, 
$\widetilde{\iota}= \iota_{\widetilde{\Mb};\widetilde{O}_1\cup\widetilde{O}_2}$. 
The spacetime $\Mb|_{O_1\cup O_2}$ has two connected components, which are just the subsets
$O_1$ and $O_2$ (recall that the underlying manifold of $\Mb|_{O_1\cup O_2}$ is just $O_1\cup O_2$ 
as a set); the same holds, {\em mutatis mutandis}, for $\widetilde{\Mb}|_{\widetilde{O}_1\cup\widetilde{O}_2}$. Each of the spacetimes
$\Ib,\Lb,\widetilde{\Lb}$ has two connected components, which we label $I_i$, $L_i$, $\widetilde{L}_i$ respectively ($i=1,2$)
so that 
\[
D_{\Mb|_{O_1\cup O_2}}(\psi(L_i))=O_i ,\qquad
D_{\widetilde{\Mb}|_{\widetilde{O}_1\cup\widetilde{O}_2}}(\widetilde{\psi}(\widetilde{L}_i))=\widetilde{O}_i,
\qquad 
D_\Ib(\varphi(L_i))=I_i=D_\Ib(\widetilde{\varphi}(\widetilde{L}_i)) 
\]
for $i=1,2$. Using properties R1 and R2 we may now argue
\begin{align*}
P_{\Mb}(O_1,O_2) &\xLongleftrightarrow[\iota]{R2} P_{\Mb|_{O_1\cup O_2}}(O_1,O_2)   \xLongleftrightarrow{R1}
P_{\Mb|_{O_1\cup O_2}}(\psi(L_1),\psi(L_2))\\
&\qquad\qquad
\xLongleftrightarrow[\psi]{R2} P_{\Lb}(L_1, L_2) 
\xLongleftrightarrow[\varphi]{R2} P_{\Ib}(\varphi(L_1), \varphi(L_2))  \xLongleftrightarrow{R1} P_{\Ib}(I_1, I_2)  
\end{align*}
where we have indicated the morphism involved in each use of R2. 
By a similar chain of reasoning, $P_{\widetilde{\Mb}}(\widetilde{O}_1,\widetilde{O}_2)\iff P_{\Ib}(I_1, I_2)$. As $P_{\Mb}(O_1,O_2)$
holds by hypothesis, we deduce that $P_{\widetilde{\Mb}}(\widetilde{O}_1,\widetilde{O}_2)$ also holds. 

For (b), we observe first that for $i=1,2$, $O_i$ certainly contains a truncated multi-diamond $D_i$ with 
the same number of components as $\widetilde{O}_i$. Then $P_\Mb(D_1,D_2)$ holds
by R3 and so $P_{\widetilde{\Mb}}(\widetilde{O_1},\widetilde{O_2})$ also holds by part (a).
\end{proof}
Note that this argument makes no specific reference to Einstein causality at all: it simply
uses the  \emph{rigidity hypotheses} R1--3, and therefore allows a number of other
results to be proved in a similar fashion. 
\begin{corollary} 
Assume that, in addition to the hypotheses of Theorem~\ref{thm:causality},  the theory $\Af$ is \emph{additive with respect to truncated multidiamonds}, i.e., each $\Af(\Mb)$ is generated
by its subalgebras $\Af^\kin(\Mb;D)$ as $D$ runs over the truncated multidiamonds of $\Mb$. 
Then $\Af$ obeys Einstein causality in full. 
\end{corollary}
\begin{proof} \smartqed
Let $\langle\widetilde{O}_1,\widetilde{O}_2\rangle\in\OO(\widetilde{\Mb})$ be chosen
arbitrarily.  It follows from the additional hypothesis that each 
$\Af^\kin(\widetilde{\Mb};\widetilde{O}_i)$ is generated by subalgebras of the form $\Af^\kin(\widetilde{\Mb};D)$ where $D$ runs over the truncated multidiamonds in
$\widetilde{O}_i$.\footnote{Here, we use the stability
of (multi)-diamonds under $\Loc$ morphisms~
\cite[Lem.~2.8]{Br&Ru05}.} Applying Theorem~\ref{thm:causality}(b), it follows that
Einstein causality holds for $\widetilde{O}_1$ and $\widetilde{O}_2$. 
\end{proof} 
\begin{remark}
These results have an interesting consequence for free electromagnetism in $n=4$ dimensions.
Consider two observables, one of which is the magnetic flux $\Phi_1$ through a $2$-surface $S_1$ bounded by closed curve $C_1$, while the other is the electric flux $\Phi_2$ through $2$-surface $S_2$ bounded by $C_2$; 
we assume that these curves lie in the $t=0$ hyperplane and have thickenings $T_i$
that are causally disjoint.\footnote{The same arguments could be applied to more general
smearings of the field-strength.} Each observable can be written as a (gauge-invariant) line integral of suitable $1$-form potentials around the relevant bounding curve and it would be natural to expect that
$\Phi_i\in\Af^\kin(\Mb_0;U_i)$, where $U_i=D_{\Mb_0}(T_i)$. But these two algebras commute, 
while the commutator $[\Phi_1,\Phi_2]$ is proportional to the linking number of $C_1$ and $C_2$ \cite{Roberts:1977},
so at least one of these natural expectations is incorrect. Indeed, if the theory
respects electromagnetic duality as a local symmetry then neither $\Phi_1$ nor $\Phi_2$
can belong to the local algebra of the relevant thickened curve.\footnote{Of course, each $\Phi_i$  \emph{is} contained in the local algebra for regions containing the $2$-surfaces $S_i$.} Provided that $\Af$ maps spacetime embeddings
to injective maps, the algebras $\Af(\Mb_0|_{U_i})$ of the nonsimply
connected spacetimes $\Mb_0|_{U_i}$ also fail to contain observables corresponding
to the $\Phi_i$. 

Note that this conclusion required no discussion of how free electromagnetism should be
formulated in spacetimes other than Minkowski, beyond the requirements of local covariance
and the timeslice property. This explains why `topological observables'
are absent from quantizations of Maxwell theory obeying these properties \cite{FewLang:2014a}. To restore them, one must relax the assumption of local covariance to permit noninjective maps \cite{DappLang:2012,SandDappHack:2012,FewLang:2014a,BecSchSza:2014}.
\end{remark}

As mentioned above, the rigidity argument can be used for other purposes. 
For instance, the \emph{Schlieder property} \cite{Schlieder:1969} relates to the algebraic
independence of local algebras of spacelike separated regions: specifically, 
it demands that the product of elements taken from two such algebras can vanish
only if at least one of the elements vanishes.  In this case we will say that 
the Schlieder property holds for the given regions. Another example is 
\emph{extended locality}, which requires that local algebras corresponding to spacelike separated regions intersect only in multiples of the unit operator. In its original
formulation \cite{Schoch1968,Landau1969} extended locality 
was established for the local von Neumann algebras
of certain spacelike separated diamonds, under standard hypotheses of AQFT
plus an additional condition on the absence of translationally invariant quasi-local observables;
it is a necessary condition for the $C^*$-independence of the corresponding 
subalgebras~\cite[Def.~2.4]{Summers1990}.
Here we formulate extended locality in the category $\Alg$.  
\begin{theorem} \label{thm:Schlieder}
Let $\Af:\Loc\to\Alg$ obey local covariance and the timeslice condition. Then the statement of Theorem~\ref{thm:causality} holds with `Einstein causality'
replaced by (a) `the Schlieder property', or (b) `extended locality'. 
\end{theorem}
\begin{proof}\smartqed 
Define $P_\Mb(O_1,O_2)$ to be the proposition that the Schlieder property
(in case (a)) or extended locality (in case (b)) holds for the
kinematic algebras associated with
$\langle O_1,O_2\rangle\in\OO^{(2)}(\Mb)$. To apply the argument in the proof of Theorem~\ref{thm:causality}
we need only check that the rigidity hypotheses R1--R3 hold. In
each case, R1 and R3 hold by the reasoning used for Einstein causality. 
To see that R2 holds in case (a), consider $\psi:\Mb\to\Nb$ and suppose
$A_i\in \Af^\kin(\Nb;\psi(O_i))$ obey $A_1A_2=0$. By   \eqref{eq:kin_covariance}, there exist $B_i\in\Af^\kin(\Mb;O_i)$ such that
$A_i=\Af(\psi)B_i$ for $B_i\in\Af^\kin(\Mb;O_i)$, which necessarily
obey $B_1 B_2=0$ because $\Af(\psi)$ 
is an injective algebra homomorphism. It follows that  
$P_\Mb(O_1,O_2)\implies P_\Nb(\psi(O_1),\psi(O_2))$. The converse is proved similarly.

In the case (b), injectivity of $\Af(\psi)$ and the covariance property \eqref{eq:kin_covariance} give  
\begin{equation}
\Af^\kin(\Nb;\psi(O_1))\cap  \Af(\Nb;\psi(O_2)) = 
\Af(\psi)\left(\Af^\kin(\Mb;O_1)\cap
\Af^\kin(\Mb;O_2)\right)
\end{equation}
and R2 is immediate.
\end{proof}
If Einstein causality and the Schlieder property both hold for $\langle O_1,O_2\rangle\in\OO^{(2)}(\Mb)$, then there is an $\Alg$-isomorphism
\begin{align}\label{eq:Roos}
\Af^\kin(\Mb;O_1)\odot \Af^\kin(\Mb;O_2) &\longrightarrow
\Af^\kin(\Mb;O_1)\vee \Af^\kin(\Mb;O_2)  
 \nonumber \\
\sum_{i} A_i\odot B_i &\longmapsto \sum_i A_i B_i
\end{align}
as shown by Roos~\cite{Roos1970}. Here, 
$\odot$ denotes the algebraic tensor product. 
If $\Af$ is finitely additive then the subalgebra on the right-hand side of
\eqref{eq:Roos} can be replaced by $\Af^\kin(\Mb;O_1\cup O_2)$,
which is isomorphic to $\Af(\Mb|_{O_1\cup O_2})$.
Now the spacetime $\Mb_{O_1\cup O_2}$ is $\Loc$-isomorphic to
(but distinct from) the disjoint union $\Mb|_{O_1}\sqcup \Mb|_{O_2}$,
so $\Af(\Mb|_{O_1\cup O_2})\cong\Af(\Mb|_{O_1}\sqcup \Mb|_{O_2})$.
The upshot is that there is an isomorphism 
\begin{equation}\label{eq:monoidal}
\Af(\Mb|_{O_1}) \odot \Af(\Mb|_{O_2}) \cong \Af(\Mb|_{O_1}\sqcup \Mb|_{O_2}).
\end{equation}
This idea may be 
extended to show that $\Af$ is a monoidal functor
between $\Loc$ (with $\sqcup$ as the monoidal product, and extended
to include an empty spacetime as the monoidal unit) and $\Alg$ (with
the algebraic tensor product); see \cite{BrFrImRe:2014}, which, however, proceeds from different assumptions. Note that the
monoidal property is not just a restatement of Einstein causality; as
shown above, it involves additional properties, notably
the Schlieder property (or, as in \cite{BrFrImRe:2014}, a form
of the split property).

In the $C^*$-algebraic setting, with $\Af:\Loc\to\CAlg$, the statements of
Theorems~\ref{thm:causality} and~\ref{thm:Schlieder} go through without change. 
However the isomorphism \eqref{eq:Roos} remains at the algebraic level, and
further conditions are needed to determine whether it can be extended to a 
$\CAlg$-isomorphism between a $C^*$-tensor product and the $C^*$-algebra
generated by the local algebras. In this context, it is most natural
to employ the minimal $C^*$-tensor product -- we refer to
\cite{BrFrImRe:2014} for more discussion.

\section{Analogues of the Reeh--Schlieder theorem and split property}

In this section, we discuss the (partial) Reeh--Schlieder and split properties
described in sect.~\ref{sect:states} in greater detail. In particular, 
we show how spacetime deformation arguments 
can be used to deduce the existence of states with (partial) Reeh--Schlieder and split properties 
on a spacetime of interest, if such states exist on a spacetime
to which it can be linked by Cauchy morphisms. The arguments are based
on those of \cite{Verch_nucspldua:1993} (for split) and~\cite{Verch:1993} (for
Reeh--Schlieder) which applied to the Klein--Gordon theory. A general treatment of Reeh--Schlieder results for locally covariant 
quantum field theories was given by Sanders~\cite{Sanders_ReehSchlieder}. 
Our treatment follows \cite{Few_split:2015}, in which the geometrical underpinnings of these
arguments were placed into a common streamlined form, yielding states
that have both the split and (partial) Reeh--Schlieder properties, in general locally
covariant theories. The Reeh-Schlieder theorem implies the ubiquity of long-range correlations in quantum field theory,
which, among other things, lead generically to entanglement across acausally separated regions or spacetime horizons
\cite{Haag,VerchWerner,Wald-CorrelBeyHorizon}. The split property implies, on the contrary, that it is also possible to
fully isolate a local system in quantum field theory such that it has no correlations with its environment and that the 
states by which this can be achieved lie locally in the folium of physical states \cite{Haag}.

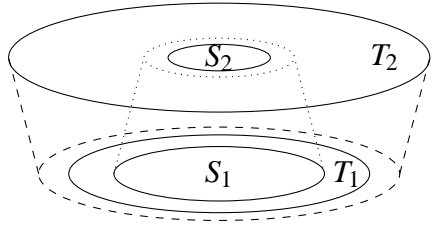
\begin{figure}[t]
\tdplotsetmaincoords{75}{90}
\pgfmathsetmacro{\rvec}{.8}
\pgfmathsetmacro{\thetavec}{15}
\pgfmathsetmacro{\phivec}{60}
\begin{center}
\begin{tikzpicture}[scale=4,tdplot_main_coords]
\coordinate (O) at (0,0,0);
\coordinate (Q) at (0,0,0.4);
\tdplotdrawarc{(O)}{0.35}{0}{360}{anchor=north}{}
\tdplotdrawarc{(O)}{0.5}{0}{360}{anchor=north}{}
\tdplotdrawarc[dashed]{(O)}{0.6}{0}{360}{anchor=north}{}
\tdplotdrawarc{(Q)}{0.17}{0}{360}{anchor=north}{}
\tdplotdrawarc[dotted]{(Q)}{0.25}{0}{360}{anchor=north}{}
\tdplotdrawarc{(Q)}{0.7}{0}{360}{anchor=north}{}
\draw[dotted] (0,0.35,0) -- (0,0.25,0.4);
\draw[dotted] (0,-0.35,0) -- (0,-0.25,0.4);
\draw[dashed] (0,0.6,0) -- (0,0.7,0.4);
\draw[dashed] (0,-0.6,0) -- (0,-0.7,0.4);
\node at (0,0.0,0.0) {$S_1$};
\node at (0,0.425,0) {$T_1$};
\node at (0,0.55,0.4) {$T_2$};
\node at (0,0.0,0.4) {$S_2$};
\end{tikzpicture}
\end{center}
\caption{\small Regular Cauchy pairs with $(S_1,T_1)\prec (S_2,T_2)$. Dotted (resp., dashed) lines
indicate relevant portions of $D_\Mb(S_1)$ (resp., $D_\Mb(T_2)$).}
\label{fig:preorder}
\end{figure}
We will make use of some particular subsets of Cauchy surfaces. 
\begin{definition} Let $\Mb\in\Loc$. A \emph{regular Cauchy pair} $(S,T)$ in $\Mb$
is an ordered pair of subsets of $\Mb$, that are nonempty open, relatively compact subsets of a common smooth spacelike Cauchy surface in which $\overline{T}$ has nonempty complement, and so that $\overline{S}\subset T$.
There is a preorder on regular Cauchy pairs so that $(S_1,T_1)\prec (S_2,T_2)$
if and only if $S_2\subset D_\Mb(S_1)$ and $T_1\subset D_\Mb(T_2)$ 
(see Fig.~\ref{fig:preorder}).\footnote{
The preorder is not a partial order, because $(S_1,T_1)\prec (S_2,T_2)\prec (S_1,T_1)$ implies $D_\Mb(S_1)=D_\Mb(S_2)$ and $D_\Mb(T_1)=D_\Mb(T_2)$, but not necessarily  $S_1=S_2$ and $T_1=T_2$.}
\end{definition}
These conditions ensure that $D_\Mb(S)$ and $D_\Mb(T)$ are 
open and casually convex, and hence elements of $\OO(\Mb)$.
Moreover, if $\psi:\Mb\to\Nb$ is a Cauchy morphism, then a pair of subsets $(S,T)$ of $\Mb$ is
a regular Cauchy pair if and only if $(\psi(S),\psi(T))$ is a regular Cauchy pair for $\Nb$.
The main property of the preorder that will be used is: 
\begin{lemma}{\cite{Few_split:2015}} \label{lem:step}
Suppose that $\Mb$ takes standard form with underlying manifold $\RR\times\Sigma$, 
and that $(S,T)$ is a regular Cauchy pair in $\Mb$, lying in the surface $\{t\}\times\Sigma$. 
Then there exists an $\epsilon>0$ such that every Cauchy surface $\{t'\}\times\Sigma$ with $|t'-t|<\epsilon$ contains a regular Cauchy pair preceding $(S,T)$ and also a regular Cauchy pair preceded by $(S,T)$.  
\end{lemma}
Clearly, $\epsilon$ may be chosen uniformly for any finite collection of regular Cauchy pairs in the Cauchy surface $\{t\}\times\Sigma$.

In this section we consider a general locally covariant theory $\Af:\Loc\to\CAlg$, 
because the properties we describe are most naturally given in the $C^*$-context. 
We assume throughout that $\Af$ has the timeslice property and obeys Einstein causality.  
\begin{definition}
Let $\Mb\in\Loc$ and suppose that $\omega$ is a state of $\Af(\Mb)$ with  
GNS representation $(\HH_\omega,\pi_\omega,\Omega_\omega)$. Then $\omega$ 
is said to have the \emph{Reeh--Schlieder property}
for a regular Cauchy pair $(S,T)$ if the GNS vector $\Omega_\omega$ is cyclic for
$\Rc_S$ and separating for $\Rc_T$,\footnote{That is, we require $\Rc_S\Omega_\omega$ to be dense in $\HH_\omega$ and $\Rc_T\owns A\mapsto A\Omega_\omega\in\HH_\omega$ to be injective.}
where $\Rc_U = \pi_\omega(\Af^\kin(\Mb;D_\Mb(U)))''$
denotes the local von Neumann algebra corresponding to $U=S,T$. 
For brevity, we will sometimes say that \emph{$\omega$ is Reeh--Schlieder for $(S,T)$}.

The state $\omega$ is said to have the \emph{split property}
for $(S,T)$ (or to be `split for' $(S,T)$) 
if there is a type-$\text{I}$ factor
$\Nc$ such that $\Rc_S\subset \Nc \subset  \Rc_T$.  
\end{definition}
\begin{remark}\label{rem:RS}
If a vector is separating for an algebra, it is separating for any subalgebra
thereof; if it is cyclic for an algebra, it is cyclic for any algebra of which it is a subalgebra. 
Thus,  if $\omega$ has the Reeh--Schlieder property for $(S,T)$ then 
it does for every $(\tilde{S},\tilde{T})$ with $(\tilde{S},\tilde{T})\prec (S,T)$. 
Similarly, if $\omega$ has the split property for $(S,T)$ then 
it does for every $(\tilde{S},\tilde{T})$ with $(S,T)\prec (\tilde{S},\tilde{T})$,
because $\Rc_{\tilde{S}}\subset \Rc_S\subset\Nc\subset \Rc_T\subset
\Rc_{\tilde{T}}$. 
\end{remark}

The proof of Theorem~\ref{thm:splitRS} below relies on  a careful geometric construction 
together with the following result, which follows easily from the uniqueness of
the GNS representation~\cite{Few_split:2015}:
\begin{lemma}\label{lem:RS} 
Let $(S,T)$ be a regular
Cauchy pair in $\Mb\in\Loc$ and suppose 
$\psi:\Mb\to\Nb$ is Cauchy.  (a) A state $\omega_\Nb$ on $\Af(\Nb)$ is Reeh--Schlieder for a regular Cauchy pair $(\psi(S),\psi(T))$ if and only if $\Af(\psi)^*\omega_\Nb$ is
Reeh--Schlieder for $(S,T)$. As $\Af(\psi)$ is an isomorphism, 
this implies that $\omega_\Mb$ is 
Reeh--Schlieder for $(S,T)$ if and only if $(\Af(\psi)^{-1})^*\omega_\Mb$
is Reeh--Schlieder for $(\psi(S),\psi(T))$.
(b) The previous statement also holds if `Reeh--Schlieder' is replaced by `split'. 
\end{lemma} 
The main result of this section is: 
\begin{theorem} \label{thm:splitRS}
Let $\Mb,\Nb\in\Loc$ have oriented-diffeomorphic Cauchy surfaces and suppose $\omega_\Nb$ is a state on $\Af(\Nb)$ 
that has the Reeh--Schlieder and split properties for all regular Cauchy pairs.  
Given any regular Cauchy pair $(S_\Mb,T_\Mb)$ in $\Mb$, there is a chain of Cauchy morphisms
between $\Mb$ and $\Nb$ inducing an isomorphism $\nu:\Af(\Mb)\to\Af(\Nb)$ such
that $\omega_\Mb =\nu^*\omega_\Nb$ is Reeh--Schlieder and split for $(S_\Mb,T_\Mb)$.  
\end{theorem}
\begin{proof}\smartqed (Sketch) Assume, without loss,
that $\Mb$ is in standard form $\Mb=(\RR\times\Sigma,g_\Mb,\ogth,\tgth_\Mb)$ so
that $S_\Mb$ and $T_\Mb$ are contained in the Cauchy surface  $\{t_\Mb\}\times\Sigma$
for some $t_\Mb\in\RR$, and that $\Nb$ is also in standard form with
$\Nb=(\RR\times\Sigma,g_\Nb,\ogth,\tgth_\Nb)$. 

By Lemma~\ref{lem:step} there exists $t_*>t_\Mb$ so that $\{t_*\}\times\Sigma$ contains 
regular Cauchy pairs $(\sS ,\sT)$ and $(S_*,T_*)$ with
\begin{equation}\label{eq:splitRSM}
(\sS ,\sT)\prec_\Mb (S_\Mb,T_\Mb)\prec_\Mb (S_*,T_*) ,
\end{equation}
where $\prec_\Mb$ indicates the preorder given by the causal structure of $\Mb$. 
As both $(\sS ,\sT)$ and $(S_*,T_*)$  are regular Cauchy pairs in a common
Cauchy surface of $\Nb$, we may also choose 
$t_\Nb>t_*$ so that $\{t_\Nb\}\times\Sigma$
contains regular Cauchy pairs $(S_\Nb,T_\Nb)$ and $(\leftidx{_\Nb}{}{S},\leftidx{_\Nb}{}{T})$
obeying
\begin{equation}\label{eq:splitRSN}
(\leftidx{_\Nb}{}{S},\leftidx{_\Nb}{}{T})\prec_\Nb (\sS ,\sT),\qquad
(S_*,T_*)\prec_\Nb (S_\Nb,T_\Nb).
\end{equation}
Using the method of Proposition~\ref{prop:Cauchy_chain}, an interpolating metric $g_\Ib$ may now 
be constructed (see~\cite{Few_split:2015}) so that
\begin{itemize}
\item $\Ib=(\RR\times\Sigma,g_\Ib,\ogth,\tgth_\Ib)$ is a $\Loc$-spacetime in standard form;
\item there is a chain of Cauchy morphisms of the form~\eqref{eq:Cauchy_chain}
between $\Mb$ and $\Nb$, via $\Ib$;
\item the orderings \eqref{eq:splitRSM} and \eqref{eq:splitRSN}
hold with $\prec_\Mb$, $\prec_\Nb$ replaced by $\prec_\Ib$.
\end{itemize}
The last item, together with transitivity of $\prec_\Ib$, entails
\begin{equation}\label{eq:splitRSI}
(\leftidx{_\Nb}{}{S},\leftidx{_\Nb}{}{T})\prec_\Ib (S_\Mb,T_\Mb)
\prec_\Ib  (S_\Nb,T_\Nb).
\end{equation}

Now $\omega_\Nb$ has the Reeh--Schlieder property for $(S_\Nb,T_\Nb)$ 
and is split for $(\leftidx{_\Nb}{}{S},\leftidx{_\Nb}{}{T})$ in $\Nb$,
and hence the same is true for $\Af(\delta)^*\omega_\Nb$ in $\Fb$
and  for $(\Af(\gamma)^{-1})^*\Af(\delta)^*\omega_\Nb$ in
$\Ib$, using Lemma~\ref{lem:RS} twice.  
By \eqref{eq:splitRSI} and Remarks~\ref{rem:RS} the latter state is both
Reeh--Schlieder and split for $(S_\Mb,T_\Mb)$, as a regular Cauchy pair in $\Ib$. 
Using Lemma~\ref{lem:RS} twice again, the same is true for 
$$ \Af(\beta)^*(\Af(\gamma)^{-1})^*\Af(\delta)^*\omega_\Nb$$
in $\Pb$
and finally for $\nu^*\omega_\Nb$ in $\Mb$, where $\nu= \Af(\delta)\circ\Af(\gamma)^{-1}\circ\Af(\beta)\circ
\Af(\alpha)^{-1}$.
\end{proof}
\begin{remark} See~\cite{Few_split:2015} for discussion, expanding on the following points:
\begin{enumerate}
\item The statement of Theorem~\ref{thm:splitRS} holds if modified so
as to refer the split or Reeh--Schlieder properties separately. 
\item We have combined the cyclic and separating aspects of the
Reeh--Schlieder results for convenience. However, if the GNS vector of
the state $\omega_\Nb$ is known to be cyclic for every local von Neumann
algebra corresponding to relatively compact $O\in\OO(\Nb)$ with nontrivial
causal complement, then it is Reeh--Schlieder for all regular Cauchy 
pairs in $\Nb$ and the conclusions of Theorem~\ref{thm:splitRS} apply.
\item The conclusions of Theorem~\ref{thm:splitRS} also
apply to more general regions:
if $O\in\OO(\Mb)$ is relatively compact with nontrivial causal complement
then one may find a regular Cauchy pair 
$(S_\Mb,T_\Mb)$ with $D_\Mb(S_\Mb)\subset O\subset D_\Mb(T_\Mb)$,
whereupon Theorem~\ref{thm:splitRS} yields a state that is both
cyclic and separating for $\pi_{\omega_\Mb}(\Af^\kin(\Mb;O))''$. 
Similarly, if $O_i\in\OO(\Mb)$ can be separated by a regular Cauchy pair
$(S_\Mb,T_\Mb)$, such that $O_1\subset D_\Mb(S_\Mb)$, $D_\Mb(T_\Mb)\subset O_2$, 
the local von Neumann algebras corresponding to the $O_i$ form a split inclusion
in the GNS representation induced by Theorem~\ref{thm:splitRS}. 
\item Suppose that $\omega_\Nb\in\Sf(\Nb)$, where $\Sf$ is a state space for 
$\Af$ obeying the timeslice condition. Then we also have $\omega_\Mb \in\Sf(\Mb)$,
because the isomorphism $\nu$ is formed from a chain of Cauchy morphisms. 
In the case of the Klein--Gordon field, for example, if $\omega_\Nb$ is
Hadamard, then so is $\omega_\Mb$. 
If the state space $\Sf$ also obeys local quasiequivalence,
then \emph{every} state of $\Sf(\Mb)$ has the split
property for \emph{every} regular Cauchy pair of $\Mb$~\cite{Few_split:2015}. If,  more strongly,
each $\Sf(\Mb)$ is a complete local quasiequivalence class, then 
there exists a full Reeh--Schlieder state in $\Sf(\Mb)$,
i.e., its GNS vector is cyclic and separating for every local von Neumann algebra
of a relatively compact region with nontrivial causal complement~\cite{Sanders_ReehSchlieder}.
\end{enumerate}
\end{remark}
Various applications of the partial Reeh--Schlieder result are discussed in~\cite{Sanders_ReehSchlieder}.
The ability to assert both partial Reeh--Schlieder and split properties simultaneously allows
one to show that local von Neumann algebras (in suitable representations) form
\emph{standard split inclusions}~\cite{DopLon:1984}, 
leading to various consequences, including the local implementation of gauge transformations
and to the classification of the local von Neumann algebras as the unique hyperfinite $\text{III}_1$ factor (up to isomorphism, and possibly tensored with an abelian centre)~\cite{Few_split:2015}. 
For the latter application, one must additionally assume the existence of a scaling limit
as described in Theorem~\ref{thm:typeIII}. 

Finally, Theorem~\ref{thm:splitRS} would be of little utility in the absence of  
spacetimes $\Nb$ for which $\Af(\Nb)$ admits
states that have the Reeh--Schlieder and split properties. Minkowski space provides
the canonical example, but one may give reasonable physical conditions that
would guarantee the existence of such states in connected ultrastatic spacetimes~\cite{Few_split:2015}. As
every connected spacetime may be linked to a connected ultrastatic spacetime by
a chain of Cauchy morphisms, one expects that Theorem~\ref{thm:splitRS}  applies
nontrivially at least in connected spacetimes $\Mb$ for most physically reasonable
locally covariant theories.

\section{Quantum energy inequalities, passivity, $\mu$SC and all that}

As already mentioned, the microlocal spectrum condition appears to be the most
promising criterion for specifying physical states (and state spaces) in quantum field theory in curved spacetime. While the $\mu$SC originated in the study of linear quantum fields, it has the potential to be relevant for interacting quantum fields and 
has proved instrumental in the perturbative construction of locally covariant
interacting quantum field theories \cite{BrFr2000,Ho&Wa01,Ho&Wa02,Hollands:2008}. One of the central points is that the
$\mu$SC permits the definition of renormalized Wick-ordered and time-ordered operators of the quantized linear
Klein-Gordon field (and its derivatives) as operator-valued distributions, with finite
fluctuations \cite{BrFr2000}. A converse of that statement has also recently been proved
for ultrastatic spacetimes: In order that the Wick-products of derivatives of the quantized linear Klein-Gordon field have finite fluctuations, it is necessary that
the Wick-ordering is defined with respect to a state obeying the $\mu$SC \cite{FewsterVerch-NecHad}. Thus, the $\mu$SC seems inevitable as
the basis for any perturbative construction of interacting quantum fields on generic globally hyperbolic spacetimes.

In this section, we review various relations between the $\mu$SC and other conditions on physical states which appear reasonable to demand in locally covariant quantum field theory and which in some way express dynamical stability.  We will also mention
some other states which have been proposed for consideration as special states, mainly for the quantized linear Klein-Gordon field.

\subsection{Quantum energy inequalities}

The locally covariant setting of quantum field theory is the theoretical basis for the semiclassical Einstein equation,
\begin{align} \label{eqn:SclEinst}
G_\Mb(h)= 8 \pi G \omega(T_\Mb[h]) \,. 
\end{align}
In this equation, $G_\Mb(h) = \int_M G_{ab}(x) h^{ab}(x)\, d{\rm vol}_\Mb(x)$ is the Einstein tensor
corresponding to the spacetime $\Mb$ (as usual, considered as an object of $\Loc$) smeared with a 
$C_0^\infty$ test-tensor field $h$. $T_\Mb[h]$ is the stress-energy tensor of a locally covariant
theory obeying the time-slice property, so it arises from the relative Cauchy evolution of the locally
covariant theory as indicated in Sect.~\ref{sect:rce}. (Actually, the relative Cauchy evolution only
specifies $[T_\Mb[h],A]$, i.e.\ the commutator of $T_\Mb[h]$ with elements $A$ in $\Af(\Mb)$, so fixing $T_\Mb[h]$ at best up to scalar multiples of the unit operator depending linearly on $h$.\footnote{In Hilbert-space representations, $T_\Mb[h]$ is an unbounded operator, one also has to consider the domain of algebra elements $A$ for which the commutator can be formed, or in which precise
mathematical sense the commutator is to be understood.} Similarly, in models
where $\omega(T_\Mb[h])$  is defined by a process of renormalization, there is a
residual finite renormalization ambiguity. As there does not seem
to be a general, locally covariant way to fix these ambiguities, \eqref{eqn:SclEinst} needs further input. We refer to \cite{Wald_qft,DapFrePin-Cosmo,VerchRegensburg} for
further discussion.)

For linear quantum field models and states $\omega$ fulfilling the $\mu$SC, it holds that
$\omega(T_\Mb[h])$ is actually given by a smooth, symmetric tensor field $\omega(T_{\Mb(ab)}(x))$ $(x \in M)$
on $\Mb$ so that $\omega(T_\Mb[h]) = \int_M \omega(T_{\Mb(ab)}(x))h^{ab}(x)\,d{\rm vol}_\Mb(x)$.
Then for any smooth future-directed timelike curve $\gamma$ in $\Mb$ parametrized by proper time $\tau$,
$$ \varrho_{\omega,\gamma}(\tau) = 
\omega(T_{\Mb(ab)}
(\gamma(\tau)))\dot{\gamma}{}^a(\tau) \dot{\gamma}{}^b(\tau) $$
is the expectation value of the energy density along $\gamma$ at $\gamma(\tau)$. 
It is known that
$$ \inf_{\omega}\, \varrho_{\omega,\gamma}(\tau) = -\infty $$
as $\omega$ ranges over the set of states fulfilling the $\mu$SC
with $\gamma$ and $\tau$ fixed. That means, at a given spacetime point
$x$, the expectation value of the energy density, for any observer, is unbounded below as a functional
of the (regular) states $\omega$ \cite{Few:Bros}. Consequently, the weak energy condition usually
assumed in the macroscopic description of matter in general relativity fails to hold in general for the expectation value of the stress-energy tensor of quantized fields. (There is an argument showing
such a behaviour also for stress-energy tensor expectation values of general Wightman-type quantum
fields on Minkowski spacetime \cite{EGJ}.) 

This violation of the weak energy condition for quantized fields is not unconstrained, however. At least for
linear quantum fields, there are quantum energy inequalities which provide restrictions on the magnitude and 
duration of the violation of the weak energy condition. Here, one says that a set of states $\boldsymbol{S}_{qei}(\Mb)$
on $\Af(\Mb)$ fulfils a \emph{quantum energy inequality (QEI)} if for any smooth, future-directed
timelike curve $\gamma$, defined on some open proper time interval $I$, there is for any 
$f \in C_0^\infty(I,\mathbb{R})$ some constant $c_{\gamma}(f) > -\infty$ such that
\begin{align} \label{QWEI}
\inf_{\omega \in \boldsymbol{S}_{qei}(\Mb)}\, \int_I \varrho_{\omega,\gamma}(\tau)f^2(\tau) \, d\tau 
 \ge c_\gamma(f) \,.
\end{align}
In other words, when averaging the energy density expectation values 
with a smooth quadratic weight function along any 
timelike curve, one obtains a quantity which is bounded below as long as the states
range over the set of states $\boldsymbol{S}_{qei}(\Mb)$.

Quantum energy inequalities were first discussed by Ford~\cite{Ford78},
initially motivated on thermodynamic grounds, but then later derived
in free models on Minkowski space~\cite{Ford:1991,FordRoman:1997,Flanagan:1997,FewsterEveson:1998}
and some curved spacetimes, e.g.,~\cite{PfenningFord_static:1998,FewsterTeo:1999}.
They been established rigorously for the quantized (minimally coupled) Klein-Gordon \cite{Fews00}, Dirac \cite{FewsterVerch-DiracQWEI,DawsFews06} and free
electromagnetic fields \cite{Few&Pfen03}, in all cases for all $\Mb$ of $\Loc$, and for
$\boldsymbol{S}_{qei}(\Mb)$ coinciding with the set of states fulfilling the $\mu$SC on $\Af(\Mb)$.
The status of QEIs for interacting quantum field theories remains to be clarified,
but it is not expected that they will hold in general without further modification \cite{OlumGraham}. Nonetheless, results are known for some
interacting models in two spacetime dimensions \cite{Fe&Ho05,BosCadFew13}
and some model-independent results are known in Minkowski space
\cite{BostelmannFewster09}.
One of the main applications of QEIs is to put restrictions on the occurrence of spacetimes with unusual causal behaviour as solutions to the semiclassical Einstein equations, e.g.\ spacetimes with closed timelike curves. We refer to the reference \cite{Fews05} and literature cited there 
for considerable further discussion.

It has been shown (for the quantized, minimally coupled Klein-Gordon field) that
the lower bounds $c_{\Mb,\gamma}(f)$ in \eqref{QWEI} can be chosen such that they comply with local covariance
\cite{FewsterSmith,Fewster2007}, i.e., for any morphism $\psi : \Mb \to \Nb$, they obey 
\begin{align} \label{cov-qei-bound}
    c_{\Nb,\psi \circ \gamma}(f) = c_{\Mb,\gamma}(f) \,.
\end{align}

Conversely, one can use QEIs as the basis of a selection 
criterion for a locally covariant state space. To be specific, suppose that for all 
objects $\Mb$ of $\Loc$ and timelike curves $\gamma : I \to \Mb$ (where $I$ is an open interval), a map
$f \mapsto c_{\Mb,\gamma}(f) \in \mathbb{R}$ ($f \in C_0^\infty(I,\mathbb{R})$) has been selected such that the
covariance condition \eqref{cov-qei-bound} is fulfilled. Let $\Af$ be a locally covariant theory and define $\boldsymbol{S}_{qei,c}(\Mb)$ to consist of all the states $\omega$ on $\Af(\Mb)$ for which the expectation value of the stress-energy tensor is defined 
and obeys
\begin{align} \label{qei-cov}
 \int_I \varrho_{\omega,\gamma}(\tau)f^2(\tau) \, d\tau 
 \ge c_{\Mb,\gamma}(f) \,, \quad f \in C_0^\infty(I,\mathbb{R})\,.
\end{align}
Evidently  $\boldsymbol{S}_{qei,c}(\Mb)$ is stable under
formation of convex combinations; if it is also stable under
 operations induced by elements in
$\Af(\Mb)$ (and this is generally the case for quantized linear fields),
then one can define a state space $\mathscr{S}$ for $\Af$ by setting $\mathscr{S}(\Mb) = \boldsymbol{S}_{qei,c}(\Mb)$.
As mentioned before, this definition is consistent (up to some details not spelled out here in full) with the microlocal spectrum condition as a selection criterion for a state space for locally covariant linear quantum fields, upon appropriate
choice of the $c_{\Mb,\gamma}(f)$. In fact, the two criteria result in the same state space, as we will indicate next, with the help of yet another selection criterion.

\subsection{Passivity}

Another very natural selection criterion is that the physical states of a locally covariant quantum field 
theory should be locally in the folia of ground states, or thermal equilibrium states, in spacetimes
which admit sufficient time-symmetry that such states exist. This is the case for ultrastatic globally hyperbolic spacetimes, i.e., those spacetimes in 
standard form $\Mb=(\mathbb{R} \times \Sigma, dt \otimes dt -  h, \ogth,
\tgth)$ 
where we write spacetime points as $(t,{\bf x})$ with $t \in \mathbb{R}$ and ${\bf x} \in \Sigma$,  
the metric $h$ is a ($t$-independent) complete Riemannian metric on $\Sigma$, and $\tgth$ is chosen so that $dt$ is future-directed. 
Then there is a global time-symmetry on
that spacetime, i.e.\ a Killing-flow $\vartheta_t : (t_0,{\bf x}) \mapsto (t_0 + t,{\bf x})$, $t \in \mathbb{R}$.
Given a locally covariant theory $\Af$, this leads to an induced 1-parametric group $\{\alpha_t\}_{t \in \mathbb{R}}$
of unital $*$-automorphisms of $\Af(\Mb)$ for any ultrastatic $\Mb$. An invariant 
 state $\omega$ on $\Af(\Mb)$ ($\omega \circ \alpha_t = \omega$) is called
a \emph{ground state} for $\{\alpha_t\}_{t \in \mathbb{R}}$ if there is a dense unital $*$-subalgebra $\Af_0(\Mb)$ of $\Af(\Mb)$ such that
$$   \frac{1}{i} \left. \frac{d}{dt} \right|_{t = 0}\, \omega(A \alpha_t(B)) \ge 0 \,, \quad A,B \in \Af_0(\Mb)\,.$$
An invariant state $\omega$ on $\Af(\Mb)$ is called \emph{passive} for $\{\alpha_t\}_{t \in \mathbb{R}}$ if there is a
dense unital $*$-subalgebra $\Af_0(\Mb)$ of $\Af(\Mb)$ with the property that
$$ \frac{1}{i} \left. \frac{d}{dt} \right|_{t=0} \, \omega(U^* \alpha_t(U))  \ge 0 $$
holds for all unitary elements of $\Af_0(\Mb)$ which are continuously connected to the unit element of
$\Af_0(\Mb)$ (cf.\ \cite{FewVer-Passivity,PuszWoronowicz} for further details).
Passive states are generalizations of KMS-states (which can be regarded as thermal equilibrium states
in the setting of ultrastatic spacetimes); see \cite{PuszWoronowicz}
for further discussion on this point.

Ground states and KMS-states of the quantized linear scalar Klein-Gordon field
on ultrastatic spacetimes are Hadamard states, i.e.\ they fulfil the microlocal spectrum condition,
as was shown in \cite{SahlmannVerch-KMS}. The same holds for convex mixtures of KMS-states at different
temperatures which are the generic examples of passive states.\footnote{In fact, the result on the Hadamard property of
ground states and KMS-states on ultrastatic spacetimes holds for more general types of quantized linear
fields, and more generally also on static (not necessarily ultrastatic) spacetimes.}  As a consequence,
ground states and KMS-states also fulfil quantum energy inequalities.

For the quantized linear Klein-Gordon field on an ultrastatic spacetime, one can show that the converse holds as well.
This was established in \cite{FewVer-Passivity} under certain additional
technical assumptions on which we suppress here, contenting ourselves with a simplified statement which
will now be outlined. The assumption is that the algebra
$\Af(\Mb)$ assigned to an ultrastatic spacetime $\Mb$ with underlying
manifold $\RR\times\Sigma$ admits
a set of states $\boldsymbol{S}_{qei}(\Mb)$, closed under convex combinations and operations, for which the stress-energy expectation values are well-defined, and
obeying a QEI of the form \eqref{QWEI}. The QEI then holds in particular for time-flow trajectories of the ultrastatic spacetime i.e.\ for all
$\gamma_{\bf x}(\tau) = (\tau,{\bf x})$, $t \in \mathbb{R}$, ${\bf x} \in \Sigma$. It is then assumed that in this case,
the quantum energy inequality holds in the form
$$ \inf_{\omega \in \boldsymbol{S}_{qei}(\Mb)}\, \int \varrho_{\omega}(\tau,{\bf x}) f^2(\tau)\,d\tau \ge c_{\Mb}(f,{\bf x})\,,
 \quad f \in C_0^\infty(\mathbb{R},\mathbb{R})\,,$$
using the abbreviations $\varrho_\omega(\tau,{\bf x})$ for $\varrho_{\omega,\gamma_{\bf x}}(\tau)$, and 
$c_{\Mb}(f,{\bf x})$ for $c_{\Mb,\gamma_x}(f)$. Making the assumption that $c_{\Mb}(f,{\bf x})$
is locally integrable and that $\Sigma$ is compact, it has been shown in \cite{FewVer-Passivity} that (i) there is
a passive state on $\Af(\Mb)$, (ii) assuming a form of energy-compactness, there is a passive state which lies in the 
folium of some state in $\boldsymbol{S}_{qei}(\Mb)$, (iii) assuming clustering properties in time, there is a ground
state in $\boldsymbol{S}_{qei}(\Mb)$.\footnote{While the results in \cite{FewVer-Passivity} have only been established for 
compact $\Sigma$, the results could be extended to noncompact $\Sigma$ upon making suitable integrability
assumptions on $c_{\Mb}(f,{\bf x})$ with respect to ${\bf x} \in \Sigma$.}  This shows that --- apart from some further technical details --- one can generally expect that the imposition of  quantum energy inequalities entails the existence of a ground (or passive) state in the folium of those obeying the QEI.

\subsection{And all that --- relations between the conditions on states}

Once a set of physical states has been specified on ultrastatic spacetimes, then, if a locally covariant
quantum field theory $\Af$ satisfies the time-slice property, the specification can be carried over to each spacetime $\Mb$ of $\Loc$.
For there certainly is an ultrastatic spacetime $\Nb$ with
Cauchy surfaces oriented-diffeomorphic to those of $\Mb$ and
therefore a chain of Cauchy morphisms linking $\Mb$ to $\Nb$, by Proposition~\ref{prop:Cauchy_chain}. The set of physical states
on $\Nb$ can then be pulled back along this chain to give
a set of states on $\Mb$. This raises a question (not addressed
in the literature) of whether
the state space on $\Mb$ obtained in this way depends on the details
of the construction; 
evidently a necessary condition is that the chosen physical states
on ultrastatic spacetimes are dynamically stable in the sense
that $\rce_\Nb[h]^*\boldsymbol{S}(\Nb) = \boldsymbol{S}(\Nb)$,
for arbitrary metric perturbations with time-compact support.  

Alternatively, one may have a specification of the state spaces
in all spacetimes, but without knowing whether any such states
exist. Here, again, the deformation argument can be used, if
(a) existence can be established in ultrastatic spacetimes, and 
(b) one has $\Af(\psi)^*\Sf(\Nb)\subset\Sf(\Mb)$ and $(\Af(\psi)^{-1})^*\Sf(\Mb)
\subset\Sf(\Nb)$ for all Cauchy morphisms $\psi:\Mb\to\Nb$. Indeed this was how Fulling, Narcowich and Wald originally proved existence of Hadamard states for the quantized
linear Klein-Gordon field on globally hyperbolic spacetimes: 
By proving that ground states on ultrastatic spacetimes have the Hadamard
property, and then making use of
the fact that the Hadamard property propagates throughout any globally hyperbolic spacetime once it is known to hold in
the neighbourhood of a Cauchy-surface \cite{FullingNarcowichWald}.      
Using the equivalence of Hadamard property and microlocal spectrum condition,
this propagation  of the Hadamard property is equivalent to the propagation of the wavefront set along bicharacteristics via
the bicharacteristic flow \cite{DuistermaatHoermander1}. 
Thus the requirements that ground states and KMS-states for ultrastatic spacetimes should
be counted among the physical states, and that all physical states should be locally quasiequivalent, 
are consistent with the demand that all physical states should be locally quasiequivalent to the states fulfilling the microlocal spectrum condition.
But relying on the results
of \cite{FewVer-Passivity}, one can even show more: The microlocal spectrum condition implies QEIs, even with
a locally covariant lower bound, and this guarantees for locally covariant quantum field theories (up to some additional technical assumptions) that on ultrastatic
spacetimes there will be ground states which are locally quasiequivalent to the states which fulfil locally covariant QEIs. 
In other words, the following three selection criteria:
\begin{itemize}
\item microlocal spectrum condition
\item locally covariant quantum energy inequalities
\item ground- or KMS-states on ultrastatic spacetimes
\end{itemize}
for the local folia of physical states are \emph{equivalent} for the locally covariant theory of the quantized linear (minimally coupled) Klein-Gordon field.
This is interesting since  these selection criteria have different motivations and implications. 
In fact, these results can be generalized at least to a larger
class of locally covariant linear quantized fields, and potentially also to certain perturbatively constructed quantum fields in a suitable version ---
a key result in this context may be that some perturbatively constructed interacting quantum fields have been shown to satisfy the time-slice property
\cite{ChiFre:2009}. However, there are also examples of non-minimally coupled quantized linear scalar fields which do not fulfil quantum energy
inequalities \cite{FewsterOsterbrink}, and arguments for interacting models
\cite{OlumGraham}, indicating that quantum energy inequalities as we have stated them are a less general property than, e.g.\ the $\mu$SC. Therefore, as mentioned
earlier, the status of quantum energy inequalities, especially in interacting quantum field theories, remains yet to be fully understood. 

\subsection{Other special states}

For quantum fields in curved spacetime, particularly for the quantized linear Klein-Gordon field, several other types of states have been proposed
as physical states, or states with special interpretation, from the very beginning of the development of
the theory.  
However, they are in some cases restricted to special spacetime geometries, or are at variance with
local covariance, or fail to be locally quasiequivalent to states fulfilling the microlocal spectrum condition in general. We shall list some
of them.
\\[4pt]
{\bf Adiabatic vacuum states}. This class of states was originally introduced by Parker in his seminal approach to particle creation in quantum
fields on expanding cosmological spacetimes \cite{Parker:1969,Parker:1971}. Adiabatic vacua for the quantized linear Klein-Gordon field have
been shown to define a single local quasiequivalence class of states on Friedmann-Lema\^{i}tre-Robertson-Walker spacetimes \cite{LuedersRoberts}
and to be locally quasiequivalent to states fulfilling the $\mu$SC \cite{JunkerSchrohe}; the latter reference extends the definition from
cosmological spacetimes to general globally hyperbolic spacetimes.  
\\[4pt]
{\bf Instantaneous vacuum states}. This class of states is essentially defined by picking a Cauchy-surface in a globally hyperbolic spacetime
and defining a two-point function for the quantized field in terms of the Cauchy-data as that two-point function which would correspond to
the ultrastatic vacuum defined by the Riemannian geometry on the Cauchy-surface obtained from the ambient spacetime metric \cite{AshtekarMagnon:1975}.
However, these states in general fail to be locally quasiequivalent to states satisfying the $\mu$SC unless
the Cauchy-surface is actually part of an ultrastatic (or at least stationary) foliation \cite{TorreMad:1999,Junker:1996}.
\\[4pt]
{\bf States of low energy}. This is a class of homogeneous, isotropic states on Friedmann-Lema\^{i}tre-Robertson-Walker spacetimes which
minimize the averaged energy density \eqref{QWEI} for given averaging function $f$. These states fulfil the microlocal spectrum condition 
\cite{Olbermann} and have several interesting properties \cite{DegVer2010}.
\\[4pt]
{\bf Local thermal equilibrium states}. This is a class of states to which one can, approximately, ascribe a temperature at each point in spacetime,
where the temperature together with the temperature rest frame is a function of the spacetime point. This class of states was introduced in
\cite{BuOjiRoos} and further investigated in \cite{Bu-hhb}; some applications to quantum fields in curved spacetime appear in 
\cite{VerchRegensburg,SchlemmerVerch,Solveen:2010,Solveen:2012}.
\\[4pt]
{\bf BMS-invariant states at conformal lightlike infinity}. For a class of asymptotically flat spacetimes, conformal lightlike infinity is a boundary manifold
which is invariant under the Bondi-Metzner-Sachs (BMS) group. Certain types of quantum fields induce quantum fields of ``conformal characteristic data'' on 
conformal lightlike infinity; then specifying a BMS-invariant, positive energy state on the ``conformal boundary'' quantum field determines a 
state of the quantum field on the original spacetime. This state fulfils the
$\mu$SC under very general assumptions \cite{DMP:05,Moretti:2008}
and has been instrumental in proving that there are solutions to the semiclassical 
Einstein equations for cosmological spacetimes with the non-conformally coupled massive quantized Klein-Gordon field \cite{Pinamonti:2011}.  
\\[4pt]
{\bf SJ-states and FP-states}. Recently, an interesting proposal for distinguished states of the quantized linear Klein-Gordon field has been made in
\cite{AAS}. There, the two-point function of such a state on $\Af(\Mb)$ is determined from $E_{\Mb}$, the propagator of the Klein-Gordon operator on
the globally hyperbolic hyperbolic spacetime $\Mb$, 
regarded as an operator on the $L^2$ space of scalar functions on $\Mb$ induced by the volume form of $\Mb$. The two-point function arises from a polar decomposition of $iE_{\Mb}$ by taking the
positive spectral part $(1/2)(|iE_{\Mb}| + iE_{\Mb})$ of $iE_{\Mb}$ as its $L^2$ kernel. The resulting (quasifree) states were named
\emph{SJ-states} in \cite{AAS}. This construction can be shown to yield a well-defined 
pure state if there is a morphism
$\psi : \Mb \to \Nb$ in $\Loc$ such that $\psi(\Mb)$ is 
relatively compact in $\Nb$  \cite{FewsterVerch-SJ}.
A heuristic argument in \cite{AAS} (corroborated in \cite{FewsterVerch-SJ}) shows that the SJ-state of Minkowski spacetime agrees with the Minkowski vacuum state of
the quantized linear Klein-Gordon field. This provided motivation in \cite{AAS} to regard the SJ-states as distinguished ``vacuum states'' for the 
quantized linear Klein-Gordon field in any spacetime. However,
explicit calculation for the case of ``ultrastatic slab'' spacetimes
shows that SJ-states in general fail to fulfil the $\mu$SC,
and even fail to be locally quasiequivalent to states fulfilling the $\mu$SC \cite{FewsterVerch-SJ}. Furthermore, derivatives of Wick-ordered quantum fields in general fail to
have finite fluctuations in SJ-states \cite{FewsterVerch-NecHad}. A modified construction of SJ-states, using a smoothing procedure on $E_{\Mb}$ and yielding states fulfilling the $\mu$SC  (at least on ultrastatic slabs and similar slabs
of cosmological spacetimes) 
has been proposed in \cite{BruFre:2014} for the
quantized linear Klein-Gordon field; this construction requires smoothing functions which parametrize the states. For the case of the quantized Dirac field, there is a construction method for states which is conceptually related to SJ states,
they are called Fermionic projector (FP) states (and incidentally predate SJ states); see \cite{Finster-DiracSea,FinsterReintjes:2013,FewsterLang-FP}. 
It has been shown in \cite{FewsterLang-FP} that FP states also fail to fulfil the $\mu$SC in general, and that a smoothing procedure again leads to a modified
FP state construction rendering states fulfilling the $\mu$SC.

\section{Locally covariant fields}\label{sect:fields} 

Conventional approaches to QFT, and the Wightman axiomatic framework, 
focus on quantum fields as the primary object of study. 
In the algebraic approach, the emphasis is rather on the
local algebras of observables; quantum fields are regarded as ways of parameterising
those local algebras. 
In the locally covariant framework, quantum fields take on a new aspect --  
they not only parameterise the local algebras in one given spacetime, 
but do so in all spacetimes in a compatible way. This idea was present
in the literature on QFT in curved spacetimes for a long time in
the context of the stress-energy tensor (see, e.g.,~\cite{Kay79,Wald_qft}); 
its use in the locally covariant context began with the treatment of the 
spin--statistics connection in curved spacetime~\cite{Verch01} and
the treatment of perturbation theory~\cite{Ho&Wa01}. It was 
one of the motivating ideas behind~\cite{BrFrVe03}, in which it took
on a more functorial form.

\subsection{General considerations in $\Loc$}\label{sect:fields_gen}

The basic idea can be illustrated by the Klein--Gordon theory $\Af:\Loc\to\Alg$.
As described in Section~\ref{sec:KG}, each algebra  $\Af(\Mb)$ is
generated by elements $\Phi_\Mb(f)$ carrying the interpretation of smeared fields;
under a morphism $\psi:\Mb\to\Nb$ the smeared fields on $\Mb$ and $\Nb$
are related by equation~\eqref{eq:covariantscalarfield}, which
can be rewritten as an equality of functions from $\CoinX{\Mb}$
to $\Af(\Mb)$
\begin{equation}\label{eq:covariantscalarfield2}
\Af(\psi)\circ\Phi_\Mb = \Phi_\Nb\circ\psi_*.
\end{equation}
Indeed, as the discussion of Section~\ref{sec:KG} makes clear, much of the
general theory has been structured on the basis of this observation,
which can be given a more categorical form as follows. 
First, the assignment of test function spaces to spacetimes may be
formalised as a functor
\begin{equation}\label{eq:Df_def}
\Df:\Loc\to\Set,\qquad \Df(\Mb)=\CoinX{\Mb}, \qquad \Df(\Mb\stackrel{\psi}{\to}\Nb) = \psi_*,
\end{equation}
where $\Set$ is the category of sets and (not necessarily injective) functions.  
Second, when we regard $\Af(\psi)$ as a function, we are appealing to the
existence of a forgetful functor $\Uf:\Alg\to\Set$ that maps each $\Alg$
object to its underlying set and each $\Alg$-morphism to its underlying function. 
Then `$\Af(\psi)$ regarded as a function' can be represented formally
by $\Uf(\Af(\psi)) = (\Uf\circ\Af)(\psi)$. Equation~\eqref{eq:covariantscalarfield2}
now becomes an equality of $\Set$-morphisms:
\begin{equation}\label{eq:covariantscalarfield3}
(\Uf\circ\Af)(\psi)\circ\Phi_\Mb = \Phi_\Nb\circ \Df(\psi)  
\end{equation}
which is required to hold for every $\Loc$-morphism $\psi:\Mb\to\Nb$,
and asserts precisely that the functions $\Phi_\Mb$ form the components
of a \emph{natural transformation} $\Phi:\Df\nto\Uf\circ\Af$. The naturality condition
can be represented diagrammatically as the requirement that the diagram
\begin{equation}
\begin{tikzpicture}[baseline=0 em, description/.style={fill=white,inner sep=2pt}]
\matrix (m) [ampersand replacement=\&,matrix of math nodes, row sep=3em,
column sep=2.5em, text height=1.5ex, text depth=0.25ex]
{\Mb \&  \Df(\Mb) \&  (\Uf\circ\Af)(\Mb) \\
\Nb \& \Df(\Nb) \&  (\Uf\circ\Af)(\Nb)\\ };
\path[->,font=\scriptsize]
(m-1-1) edge node[auto] {$ \psi $} (m-2-1)
(m-1-2) edge node[auto] {$ \Phi_\Mb $} (m-1-3)
        edge node[auto] {$ \Df(\psi) $} (m-2-2)
(m-2-2) edge node[auto] {$ \Phi_\Nb $} (m-2-3)
(m-1-3) edge node[auto] {$ (\Uf\circ\Af)(\psi) $} (m-2-3);
\end{tikzpicture}
\end{equation}
commutes for every morphism $\psi:\Mb\to\Nb$. 

\begin{definition} 
A locally covariant scalar field
of theory $\Af$ is a natural transformation $\Phi:\Df\nto\Uf\circ\Af$. 
The collection of all locally covariant scalar fields is denoted $\Fld(\Df,\Af)$. 
\end{definition}
This definition encompasses both linear and nonlinear fields -- for example, in the Weyl formulation of the
Klein--Gordon theory, the map from test functions to generators $f\mapsto W_\Mb(f_\sim)$ defines a nonlinear locally covariant field.

Unlike many structures in locally covariant QFT, individual locally
covariant fields are not stable under the evolution entailed by 
the timeslice property (in particular, relative Cauchy evolution). 
Fixing $\Phi_\Mb(f)\in\Af(\Mb)$, it is of course true that
$\Phi_\Mb(f)\in\Af^\kin(\Mb;S)$ for any subset $S\in\OO(\Mb)$
that contains a Cauchy surface of $\Mb$. In general, however,
we cannot write $\Phi_\Mb(f)$ in the form $\Phi_\Mb(h)$ for $h$ supported in $S$; while this can be done for linear fields such as the Klein--Gordon
model, the evolution induced by the timeslice assumption becomes much more involved as soon as Wick powers are included~\cite{ChiFre:2009}.  

The description of fields at the functorial level, rather
than that of individual spacetimes, opens new ways of manipulating them as mathematical objects. As shown in~\cite{Fewster2007},
$\Fld(\Df,\Af)$ can be given the structure of a unital $*$-algebra: given 
$\Phi,\Psi\in\Fld(\Df,\Af)$, and $\lambda\in\CC$, we may define new
fields $\Phi+\lambda\Psi$, $\Phi\Psi$, $\Phi^*$ by
\begin{align}
(\Phi+\lambda\Psi)_\Mb(f) &= \Phi_\Mb(f) + \lambda\Psi_\Mb(f)  \label{eq:fld_lin}
\\ 
(\Phi\Psi)_\Mb(f) &= 
\Phi_\Mb(f)\Psi_\Mb(f), \label{eq:fld_prod} \\
(\Phi^*)_\Mb(f) &= \Phi_\Mb(f)^* \label{eq:fld_star}
\end{align}
and the unit field may be defined by $\II_\Mb(f) = \II_{\Ff(\Mb)}$, for all $f\in\CoinX{\Mb}$. Furthermore, in the $C^*$-algebraic setting, one may
even find a $C^*$-norm on a $*$-subalgebra of $\Fld(\Df,\Af)$. 
The abstract algebra of fields has a number of interesting features: 
for example, it carries an action of the global gauge group (see Sec.~\ref{sect:gauge}). 
 
For many quantum fields of interest the maps $\Phi_\Mb$
are linear. Such {\em linear fields} may be singled out as follows:
regarding $\Df$ now as a functor from $\Loc$ to $\Vect$, the category of complex vector spaces and (not necessarily injective) linear maps, and writing $\Vf$
for the forgetful functor $\Vf:\Alg\to\Vect$, the linear fields of 
the theory $\Af$ are natural transformations $\Phi:\Df\nto\Vf\circ\Af$,
and form the collection $\Fldlin(\Df,\Af)$. While $\Fldlin(\Df,\Af)$
can be given the structure of a vector space, by~\eqref{eq:fld_lin},
it does not in general admit a product structure, because~\eqref{eq:fld_prod}
creates a field that is nonlinear in its argument. Similarly the $*$-operation 
of~\eqref{eq:fld_star} creates an antilinear field and so cannot be
defined as a map of $\Fldlin(\Df,\Af)$ to itself. However, we can define
an antilinear involution $\star$ on $\Fldlin(\Df,\Af)$ by
\begin{equation}\label{eq:fldlin_star}
\Phi^\star_\Mb(f) = \Phi_\Mb(\overline{f})^* \qquad (f\in\Df(\Mb))
\end{equation}
and it is natural to do so in this context. For example,
$\Phi$ is \emph{hermitian} if $\Phi^\star=\Phi$.\footnote{The reader might wonder
why this is not adopted for $\Fld(\Df,\Af)$ in place of~\eqref{eq:fld_star}.
The reason is that $\Phi^*\Phi$ is a positive element of $\Fld(\Df,\Af)$
in the sense that $(\Phi^*\Phi)_\Mb(f) = \Phi_\Mb(f)^*\Phi_\Mb(f)$ is
a positive element in $\Af(\Mb)$ for every $f\in\Df(\Mb)$, while $\Phi^\star\Phi$ need not be positive in this way. Order 
structure and functional calculus for abstract fields is discussed
in \cite{Fewster2007}.}  

The definition of a locally covariant scalar field can be generalised in various ways, by modifying the functor $\Df$ (as well as choosing whether to work in $\Set$ or $\Vect$). For example, $\Df$ could be the functor assigning compactly supported smooth sections of tensor fields of some specified (but arbitrary) type; the
corresponding linear fields can be regarded as tensor fields of the dual type.
Again, in some applications one might allow
certain subsets of compactly supported distributions (see, e.g.,~\cite{Fewster2007}). 
In all such cases, we will use the notation $\Fld(\Df,\Af)$ for the natural
transformations from $\Df$ (in $\Set$) to $\Uf\circ\Af$ and 
$\Fldlin(\Df,\Af)$ for those from $\Df$ (in $\Vect$) to $\Vf\circ\Af$. 
Spinorial fields require additional structure beyond those of $\Loc$
and will be discussed briefly in Sect.~\ref{sect:spin}. 

We can also define multilocal fields, replacing $\Df$ by
$\Df^{\times k}(\Mb) = \CoinX{\Mb}^{\times k}$, $\Df^{\times k}(\psi) = 
\psi_*^{\times k}$, or by $\Df^{\otimes k}(\Mb) = \CoinX{\Mb}^{\otimes k}$, $\Df^{\otimes k}(\psi) =  \psi_*^{\otimes k}$ in the linear case. 
Local fields can be combined to form multilocal fields in obvious ways. 
For example, given $\Phi,\Psi\in\Fldlin(\Df,\Af)$, we may define bilocal
fields $\Phi\stackrel{\rightarrow}{\otimes}\Psi$  and 
$\Phi\stackrel{\leftarrow}{\otimes}\Psi$,
i.e., natural transformations from $\Df^{\otimes 2}$ to $\Vf\circ\Af$,
by
\begin{align}\label{eq:bilocal}
(\Phi\stackrel{\rightarrow}{\otimes}\Psi)_\Mb(f\otimes h) &= \Phi_\Mb(f)\Psi_\Mb(h) \nonumber \\
(\Phi\stackrel{\leftarrow}{\otimes}\Psi)_\Mb(f\otimes h) &= \Psi_\Mb(h)\Phi_\Mb(f)
\end{align}
for $f,h\in\Df(\Mb)$, $\Mb\in\Loc$. 
These structures will be useful in  Sec.~\ref{sec:KGreform},
where it will be shown that the abstract viewpoint on fields allows
the Klein--Gordon theory to be specified directly at the functorial
level in terms of its generating field.

\subsection{The inclusion of spin}
\label{sect:spin}
 
The inclusion of fields with spin requires a modification of the category
of spacetimes to incorporate spin structures. For
definiteness, let us work in $n=4$ spacetime dimensions,
in which there are some simplification~\cite{Isham_spinor:1978}. 
In particular, every globally hyperbolic manifold $\Mb\in\Loc$ admits
a unique spin bundle  (up to equivalence), namely 
the trivial bundle $S\Mb:=\Mb\times\SL(2,\CC)$, regarded
as a right principal bundle. A spin
structure in this context is a smooth double covering   
$\sigma$ from $S\Mb$ to the bundle $F\Mb$ of oriented and
time-oriented orthonormal frames on $\Mb$ (also a right-principal bundle with structure group given by the proper orthochronous Lorentz group
$\Lc^\uparrow_+$), such that $\sigma\circ R_A=R_{\Lambda(A)}\circ\sigma$,
where $\Lambda:\SL(2,\CC)\to\Lc^\uparrow_+$ is the standard
double cover of groups and we use $R$ for each right action.

There is always at least one spin structure in our current
setting, and the distinct possibilities are classified up to equivalence
by the cohomology group $H^1(\Mb;\ZZ_2)$. 
We replace $\Loc$ by a new category $\SpLoc$, whose
objects are pairs $(\Mb,\sigma)$
where $\sigma$ is a spin structure for $\Mb$; a morphism
between objects $(\Mb,\sigma)$ and $(\Mb',\sigma')$
of $\SpLoc$ is a bundle morphism $\Psi:S\Mb\to S\Mb'$ such that
\begin{itemize}
\item $\Psi(p,A) = (\psi(p), \Xi(p)A)$ for some $\Loc$-morphism 
$\psi:\Mb\to\Mb'$ and smooth function $\Xi:\Mb\to\SL(2,\CC)$;
\item $\sigma'\circ\Psi = \psi_*\circ\sigma$, where
$\psi_*:F\Mb\to F\Mb'$ is induced by the tangent map of $\psi$.  
\end{itemize}
It is convenient to write $\Psi=(\psi,\Xi)$ under these circumstances.

A locally covariant quantum field theory
is now a functor $\Af:\SpLoc\to\Alg$ (or $\CAlg$, or some other
category as in~\cite{Verch01}). Note
that such a functor encodes both geometric embeddings
and spin rotations. The timeslice property can be defined
as before, regarding $\Psi=(\psi,\Xi)$ as Cauchy in 
$\SpLoc$ whenever $\psi$ is Cauchy in $\Loc$. 

It is now possible to introduce locally covariant fields
of different spin, starting with the construction of 
appropriate test function spaces. Let $\rho$ be any (real or complex)
representation of $\SL(2,\CC)$ on vector space $V_\rho$,
and write $\KK=\RR$ (resp., $\CC$) in the real (resp., complex) case.
Given any object $(\Mb,\sigma)$ of $\SpLoc$, 
let $\Df_\rho(\Mb,\sigma)=\CoinX{\Mb;V_\rho}$ be the 
space of compactly supported functions on $\Mb$ with values in $V_\rho$;\footnote{Somewhat
more technically, $\Df_\rho(\Mb,\sigma)$ may be regarded
as the space of compactly supported sections of the bundle $\Mb\ltimes_\rho V_\rho$ associated to $S\Mb$ and $\rho$.} 
given any $\SpLoc$ morphism $\Psi:(\Mb,\sigma)\to (\Mb',\sigma')$, 
define $\Df_\rho(\Psi):\Df_\rho(\Mb,\sigma)\to
\Df_\rho(\Mb',\sigma')$ by
\[
\Df_\rho(\Psi)f =\psi_*\left(\rho(\Xi)f\right),
\]
where $\Psi=(\psi,\Xi)$ as above, and $\CoinX{\Mb;V_\rho}\owns\rho(\Xi)f:p\mapsto \rho(\Xi(p))f(p)$. It is easily checked that
$\Df_\rho$ is a functor from $\SpLoc$ to the category of vector spaces over $\KK$.  The $\KK$-linear locally covariant fields $\Fldlin(\Af,\Df_\rho)$ 
are now naturally regarded as fields of `type $\rho$'.\footnote{Terminology
here is parallel to \cite{Verch01} but one could equally make
a case for labelling the type by the dual (also known as contragredient) representation $\rho^*$,
which was our convention in Sect.~\ref{sect:fields_gen}.}

Particular interest attaches to the irreducible complex
representations $D^{(k,l)}$ ($k,l\in\NN_0$) of $\SL(2,\CC)$ on 
the vector space $V^{(k,l)}=(\textcircled{s}^k\CC^2)\otimes 
(\textcircled{s}^l\CC^2)$, where $\textcircled{s}$ denotes
a symmetrised tensor product, and 
\[
D^{(k,l)}(A) = 
A{}^{\textcircled{s} k}\otimes \overline{A}{}^{\textcircled{s} l}
\qquad (A\in \SL(2,\CC)),
\]
with the bar denoting complex conjugation. These
representations exhaust the finite-dimensional complex
irreducible representations of $\SL(2,\CC)$ up to equivalence, and 
are familiar from the Minkowski space theory~\cite{StreaterWightman}.
Irreducible real-linear representations of interest
are formed from $D^{(k,l)}\oplus D^{(l,k)}$ ($k\neq l$) or $D^{(l,l)}$
restricted to suitable subspaces. In an obvious way, we will 
call locally covariant $\KK$-linear fields associated with these
representations \emph{fields of type $(k,l)$}.
An important distinction is between the cases in which $\frac{1}{2}(k+l)$
is integer or half-integer, which (assuming the normal spin-statistics
relation) corresponds to bosonic or fermionic fields. We now
turn to a more detailed discussion of this point.

\subsection{The spin--statistics connection} \label{sect:spinstats}

It is an observed fact that particles of integer spin
display bosonic statistics and those of half-integer spin
obey fermionic statistics. One of the major early successes
of axiomatic QFT (see, e.g., the classic presentation in~\cite{StreaterWightman}) was to prove that these observed
facts are consequences of the basic axioms, and therefore
true in a model-independent fashion. Of course, these
proofs are formulated for Minkowski space QFT, and make full use of the 
Poincar\'e symmetry group, while the experimental
observations take place in a curved spacetime, so
it is important to understand how the spin-statistics theorem 
can be extended to general backgrounds. 

While a number of authors had demonstrated the inconsistency
of various free models with incorrect statistics in 
curved spacetimes~\cite{Wald_Smatrix:1979,ParkerWang:1989},\footnote{An interesting variant shows what happens if
negative-normed states are allowed~\cite{HiguchiParkerWang:1990}.} no model-independent
result was available until~\cite{Verch01}, 
which introduced a number of ideas that now form the
basis of locally covariant QFT. The 
framework employed in~\cite{Verch01} differs
in two important respects from that of~\cite{BrFrVe03}:
the category of spacetimes is $\SpLoc$ rather than $\Loc$, 
and the target category is neither
$\Alg$ nor $\CAlg$, but rather a category 
whose objects are nets of von Neumann algebras
indexed over relatively compact spacetime subsets. 
For these reasons (particularly the second) it would
be a departure from our development to describe
the results of~\cite{Verch01} in detail. 

In broad terms, however, the spin--statistics connection proved in~\cite{Verch01} is as follows. We consider a theory on $\SpLoc$ in which,
on each spacetime, the net of local von Neumann algebras is generated by a field of type $(k,l)$. It is also assumed that the
instantiation of the theory in Minkowski
space\footnote{Here, the trivial
spin structure $\sigma_0(A) = R_{\Lambda(A)}e$ is intended, where
$e=(\partial/\partial x^\mu)_{\mu=0,\ldots,3}$ is the orthonormal
frame on Minkowski space associated with standard inertial coordinates $x^\mu$.} is a Wightman theory with the corresponding component of $\Phi$ as a Wightman field. One supposes that there is a
spacetime $(\Mb,\sigma)$ and a pair of causally
disjoint and relatively compact regions $\langle O_1,O_2\rangle\in\OO^{(2)}(\Mb)$
so that $\Phi$ exhibits anomalous statistics
\[
\Phi_{(\Mb,\sigma)}(f_1) \Phi_{(\Mb,\sigma)}(f_2)^*
+ (-1)^{k+l} \Phi_{(\Mb,\sigma)}(f_2)^*  \Phi_{(\Mb,\sigma)}(f_1)
= 0
\]
for all $f_i\in\CoinX{O_i;V^{(k,l)}}$ (or if the same holds with the adjoint removed). By a prototype of the rigidity arguments
discussed in Sect.~\ref{sect:rigidity},
it is proved that there must be a violation of the spin--statistics
connection in Minkowski space, which can only happen if all
smearings of $\Phi$ are trivial in Minkowski space~\cite[Thm 4--10]{StreaterWightman}. Thus the 
local algebras in Minkowski space, generated by $\Phi$, consist only of multiples of the unit and it follows that the same is true of local
algebras in all spacetimes. As these algebras are generated
by $\Phi$, one may conclude that \emph{in every spacetime}, 
all smearings of $\Phi$ are multiples of the unit operator
(vanishing in Minkowski space). 

To close this section, we note a number of potential
extensions. First, it is a slightly unsatisfactory feature  that 
the use of spin structures invokes unobservable geometric
structures from the start; similarly, the idea of spin is,  
to an extent, inserted by hand at
the start of the construction. It would be desirable
to understand more clearly why spin (which is tightly linked to rotations
in Minkowski space) continues to be an appropriate
notion in general curved spacetimes, and how it can
be incorporated in a more operational way. Second,
one would also like a spin--statistics connection
that is not based on algebras generated by a single 
field.  An account addressing these points is sketched in
\cite{Few_Regensburg:2015} and will appear in full shortly. 
The key idea is to base the framework on a category of spacetimes with global coframes
(i.e.,   a `rods and clocks' account of 
spacetime measurements). The spin--statistics theorem
that emerges from this analysis is again proved by
rigidity methods on the assumption that the
theory obeys standard statistics in Minkowski space.

\section{Subtheory embeddings and the global gauge group}\label{sect:gauge}

In category theory, it is often the morphisms between functors, 
i.e., natural transformations, that are the main point of interest.
Natural transformations appear in locally covariant QFT with important
physical interpretations: they are used in the description of
locally covariant fields and in order to compare theories.

The idea that equivalences of functors denote physically equivalent
theories was already present in~\cite{BrFrVe03}; the use of general
natural transformations to indicate subtheory embeddings was 
introduced in~\cite{FewVer:dynloc_theory}, while a systematic
study of endomorphisms and automorphisms of locally covariant
theories is given in~\cite{Fewster:gauge}, on which our presentation is based. 
In this section we consider
theories obeying Assumptions \ref{ax:loc_cov}--\ref{ax:timeslice} throughout.

\subsection{Subtheory embeddings: definition}
\begin{definition}
Let $\Af,\Bf:\Loc\to\Alg$ be locally covariant theories. 
Any natural transformation $\eta:\Af\nto\Bf$ is said to \emph{embed $\Af$ as a subtheory of $\Bf$}.
\end{definition}
The requirement that $\eta$ be natural means that there
is a collection of morphisms $\eta_\Mb:\Af(\Mb)\to\Bf(\Mb)$ ($\Mb\in\Loc$) such that the following diagram commutes
for every morphism $\psi:\Mb\to\Nb$ of $\Loc$:
\begin{equation}
\begin{tikzpicture}[baseline=0 em, description/.style={fill=white,inner sep=2pt}]
\matrix (m) [ampersand replacement=\&,matrix of math nodes, row sep=3em,
column sep=2.5em, text height=1.5ex, text depth=0.25ex]
{\Mb \&  \Af(\Mb) \&  \Bf(\Mb) \\
\Nb \& \Af(\Nb) \&  \Bf(\Nb)\\ };
\path[->,font=\scriptsize]
(m-1-1) edge node[auto] {$ \psi $} (m-2-1)
(m-1-2) edge node[auto] {$ \eta_\Mb $} (m-1-3)
        edge node[auto] {$ \Af(\psi) $} (m-2-2)
(m-2-2) edge node[auto] {$ \eta_\Nb $} (m-2-3)
(m-1-3) edge node[auto] {$ \Bf(\psi) $} (m-2-3);
\end{tikzpicture}
\end{equation}
That is, the transition between theories commutes
with the transitions between spacetimes. 
An example of a subtheory embedding is given by the 
even Klein--Gordon theory $\Af^{\text{ev}}$  defined 
in Sect.~\ref{sect:assumptions}. For each $\Mb$, let $\eta_\Mb:\Af^{\text{ev}}(\Mb)\to \Af(\Mb)$
be the inclusion of the subalgebra. As $\Af^{\text{ev}}(\Mb)$ is generated
by the unit and bilinear expressions $\Phi_\Mb(f)\Phi_\Mb(h)$ ($f,h\in\Df(\Mb)$) and $\Af^{\text{ev}}(\psi)$ is
a restriction of $\Af(\psi)$,
it is easily seen that $\eta_\Nb\circ\Af^{\text{ev}}(\psi)= \Af(\psi)\circ\eta_\Mb$ 
for all $\psi:\Mb\to\Nb$, and hence that $\eta:\Af^{\text{ev}}\nto\Af$. 
 
The physical interpretation of natural transformations as subtheory embeddings
 is supported by the following observations.
\begin{proposition}\label{prop:subtheory}
If $\eta:\Af\nto\Bf$, then (a)
\begin{equation}
\eta_\Mb\Af^\kin(\Mb;O)\subset \Bf^\kin(\Mb;O)
\end{equation}
for all nonempty $O\in\OO(\Mb)$ and $\Mb\in\Loc$,
with equality if $\eta$ is a natural isomorphism (also
called an \emph{equivalence});
(b) if $\psi\in\End(\Mb)$ then
\begin{equation}
\eta_\Mb\circ\Af(\psi) = \Bf(\psi)\circ\eta_\Mb;
\end{equation}
(c) for all $h\in H(\Mb)$, 
\begin{equation}
\eta_\Mb \circ \rce_\Mb^{(\Af)}[h]
=\rce_\Mb^{(\Bf)}[h] \circ  \eta_\Mb .
\end{equation}
\end{proposition}
\begin{proof}\smartqed
Part~(b) is simply a special case of the definition.
Similarly, applying the definition to the embedding $\iota_{\Mb;O}$
we have $\Bf(\iota_{\Mb;O})\circ\eta_{\Mb|_O}= \eta_\Mb\circ
\Af(\iota_{\Mb;O})$ from which part~(a) follows on taking images.
Part~(c) is proved in \cite[Prop.~3.8]{FewVer:dynloc_theory} and is again simply a matter of 
employing the basic definitions a number of times. 
\end{proof}
The above result shows that subtheory
embeddings act locally (part (a)) and intertwine both geometric
symmetries (part (b)) and the dynamics of the theory (c).
In particular, if the relative Cauchy evolution can be differentiated
to yield a stress-energy tensor, acting as a derivation, then (c)
implies
\begin{equation}
[T_\Mb^{(\Bf)}[\fb],\eta_\Mb A] = \eta_\Mb [T_\Mb^{(\Af)}[\fb],A]
\end{equation}
which shows clearly that $\eta_\Mb$ identifies degrees of
freedom of $\Af$ with some of those of $\Bf$ in a physically meaningful way.
 
An equivalence of a theory $\Af$ with itself -- an automorphism of the theory -- 
has a special significance. The automorphisms of any functor $\Af$ form a group $\Aut(\Af)$
under composition, and it is a pleasing aspect of locally covariant quantum
field theories that their automorphism groups can be interpreted as
global gauge groups~\cite{Fewster:gauge}: as Prop.~\ref{prop:subtheory}(a),(b)
shows, any $\zeta\in\Aut(\Af)$ has components $\zeta_\Mb$ that map each
local algebra  $\Af^\kin(\Mb;O)$ isomorphically to itself  and commute with
the action of spacetime symmetries. These are natural generalisations of
conditions set down by Doplicher, Haag and Roberts~\cite{DHRi} for global gauge symmetries 
in Minkowski AQFT.\footnote{DHR work
in the Hilbert space representation of the Poincar\'e invariant vacuum state,
and require that global gauge transformations should leave the vacuum vector
invariant. This also has an analogue in the present setting~\cite{Fewster:gauge}.} 
In the DHR analysis, nets of local algebras with nontrivial global gauge group are
called \emph{field algebras}. By contrast, a local algebra of observables is the 
subalgebra of the corresponding field algebra consisting of fixed elements under 
the action of the gauge group. One may make a similar construction in the
locally covariant context~\cite{Fewster:gauge}: defining $\Af_{\text{obs}}(\Mb)$
to be the subalgebra of $\Af(\Mb)$ fixed under the action of all $\zeta_\Mb$ ($\zeta\in\Aut(\Af)$),
each $\psi:\Mb\to\Nb$ induces a $\Af_{\text{obs}}(\psi):\Af_{\text{obs}}(\Mb)\to
\Af_{\text{obs}}(\Nb)$ by restriction of $\Af(\psi)$, and overall yields
a new locally covariant theory $\Af_{\text{obs}}$ that can be taken as the theory of observables'
relative to the `field functor' $\Af$. This interpretation is not entirely satisfactory  
(see \cite[\S 3.3]{Fewster:gauge} for some cautionary remarks) but works well
in a number of examples.

\subsection{Subtheory embeddings: classification}
 
The introduction of natural transformations raises the question of whether
they are operationally meaningful, given the need to discuss relationships
between theories on all possible spacetimes. This question is answered
by a rigidity argument similar to those used in Section~\ref{sect:rigidity}. 
\begin{theorem} \label{thm:subtheory_rigidity}
Suppose $\eta,\zeta:\Af\nto\Bf$, for theories $\Af,\Bf$, with $\Af$ assumed additive with respect to truncated multidiamonds.
If, for some $\Mb\in\Loc$ and nonempty $O\in\OO(\Mb)$,  $\eta_\Mb$ and $\zeta_\Mb$ agree on the local kinematic algebra $\Af^\kin(\Mb;O)$, then $\eta=\zeta$. 
\end{theorem}
\begin{proof} \smartqed 
This is a straightforward generalization
of \cite[Thm~2.6]{Fewster:gauge}. We remark that 
the timeslice property of $\Bf$ is not used.
\end{proof}
This result shows that the local behaviour of a subtheory embedding in one
spacetime is enough to fix it uniquely. In any individual spacetime, moreover, 
Prop.~\ref{prop:subtheory}(c) gives strong constraints and facilitates the 
classification of subtheory embeddings. 

Two examples have been worked out in full detail. For 
the example of finitely many independent minimally coupled Klein--Gordon fields, 
with $\nu_m$ denoting the number of fields of mass $m$,  the gauge group is a direct product of factors $G_m$ over the mass spectrum, with $G_m=\text{O}(\nu_m)$ for $m>0$ and  
$G_0 =\text{O}(\nu_0) \ltimes \RR^{\nu_0*}$,  
where  $\RR^{k*}$
denotes the additive group of $k$-dimensional real row vectors, and the semidirect
product is given by $\left(R,\ell\right)\cdot \left(R',\ell'\right) = \left(RR',\ell R_0' +\ell'\right)$  \cite{Fewster:gauge}.
For example, the theory $\Af^{(\nu)}$ consisting of $\nu$ Klein--Gordon fields $\Phi^{(j)}$ ($1\le j\le \nu$) of common mass $m>0$, has automorphisms $\zeta_R$ labelled by $R\in\text{O}(\nu)$, acting so that 
\begin{equation}\label{eq:On_action}
(\zeta_R)_\Mb\Phi_\Mb^{(j)}(f) = R_{i}^{\phantom{i}j}\Phi_\Mb^{(i)}(f)
\end{equation}
(summing on $i$), while in the massless case, there are automorphisms 
$\zeta_{(R,\ell)}$ labelled by $(R,\ell)\in \text{O}(\nu_0) \ltimes \RR^{\nu_0*}$, so that
\[
(\zeta_{(R,\ell)})_\Mb\Phi_\Mb^{(j)}(f) = R_{i}^{\phantom{i}j}\Phi_\Mb^{(i)}(f) +
\left(\int_\Mb  f\,\dvol_\Mb\right)\ell^j \II_{\Af^{(\nu)}(\Mb)}.
\]
It is not hard to verify that these formulae define automorphisms of $\Af^{(\nu)}(\Mb)$
that are components of natural transformations. What was shown in~\cite{Fewster:gauge}
was a rather more: every \emph{endo}morphism $\eta$ of $\Af^{(\nu)}$, at least 
under the additional assumption of regularity that every $\eta_\Mb^*$ maps states with
distributional $k$-point functions to states with distributional $k$-point functions, 
is one of the \emph{auto}morphisms described above. 

The second case studied was the Klein--Gordon theory with external sources~\cite{FewSchenkel:2014},
which is formulated on a category of spacetimes with sources. 
Here, the gauge group can be determined at the purely algebraic level,
without additional regularity conditions. As might be expected, the effect of the external sources
is to break the $\text{O}(\nu_m)$ symmetries for $m\ge 0$, leaving only a $\RR^{\nu_0*}$ symmetry
for $m=0$.  

In both examples just mentioned, every endomorphism of the theory turns out to be an automorphism;
there is no way of properly embedding the theory as a subtheory of itself. It is not hard 
to give examples of locally covariant theories where this is not the case: for example, 
the theory of countably many independent scalar fields $\Af^{(\aleph_0)}$ of common mass and
coupling constant has an endomorphism $\eta$ acting on the generating fields by
$\eta_\Mb \Phi^{(j)}_\Mb(f) = \Phi^{(j+1)}_\Mb(f)$ for all $f\in\CoinX{\Mb}$, $\Mb\in\Loc$.
However, under a condition of \emph{energy compactness} (weaker than either of the nuclearity~\cite{BucWic:1986} or Haag--Swieca~\cite{HaaSwi:1965} criteria) 
it may be shown that proper endomorphisms are excluded and that all endomorphisms
are automorphisms~\cite[Thm 4.6]{Fewster:gauge}. The additional  assumptions
required are that the instantiation of the locally covariant theory in Minkowski space
should comply with standard assumptions of AQFT, and also that there are
no `accidental symmetries' of the Minkowski space theory.  
The result also shows that, if the gauge group is given 
a natural topology (which requires the introduction of a state space) then it is compact. 

The gauge group provides a useful invariant of locally covariant theories, because
the automorphism groups of isomorphic functors are isomorphic. This allows one
to read off, for example, that the theories $\Af^{(j)}$ described above are
inequivalent for distinct values of $j$,  by virtue of their inequivalent gauge groups.  
In the same vein, the computation of the gauge group in~\cite{FewSchenkel:2014} was used to show that an earlier quantization of the Klein--Gordon theory with sources was incorrect,
because its gauge group contained unexpected symmetries. Regarding subtheory
embeddings, if $\eta:\Af\nto\Bf$ and $\zeta:\Bf\nto\Af$, and (say) $\Af$ obeys
the hypotheses of \cite[Thm 4.6]{Fewster:gauge}, then $\zeta\circ\eta$ must be
an automorphism of $\Af$, so (as $\zeta$ is monic) $\eta$ and $\zeta$ are both isomorphisms~\cite[Cor. 4.7]{Fewster:gauge} (cf.\ the 
Cantor--Schröder--Bernstein theorem for sets).
 
A good example of \emph{inequivalent} theories is given by Klein--Gordon theories $\Af_1$ and $\Af_2$ with
distinct masses $m_1$ and $m_2$~\cite{BrFrVe03} (we give a slightly different argument). 
Let $\omega_2$ be the Poincar\'e-invariant vacuum state on the Minkowski space theory 
$\Af_{2}(\Mb_0)$, which means that $\Af(\psi)^*\omega_2=\omega_2$ for every Poincar\'e transformation $\psi:\Mb_0\to\Mb_0$.
If there is an equivalence $\zeta:\Af_{1}\nto\Af_{2}$ then 
Proposition~\ref{prop:subtheory}(b) implies that $\omega=\zeta_{\Mb_0}^*\omega_2$ satisfies
\[
\Af_1(\psi)^*\omega=
\Af_1(\psi)^*\zeta_{\Mb_0}^*\omega_2= \zeta_{\Mb_0}^*\Af_2(\psi)^*\omega_2=\zeta_{\Mb_0}^*\omega_2=\omega
\]
for all Poincar\'e transformations $\psi$, so $\omega$ is a Poincar\'e-invariant state on $\Af_1(\Mb_0)$. 
Indeed, the GNS representation of $\Af_1(\Mb_0)$ induced by $\omega$ can be taken as $(\HH_2,\pi_2\circ\zeta_{\Mb_0},\DD_2,\Omega_2)$
where $(\HH_2,\pi_2 ,\DD_2,\Omega_2)$ is the GNS representation induced by $\omega_2$. 
Crucially, the unitary implementation $U_2(\psi)$ of Poincar\'e transformations in $\HH_2$ 
also implements the Poincar\'e transformations on $\Af_1(\Mb_0)$: 
setting $\pi=\pi_2\circ\zeta_{\Mb_0}$,
\begin{align*}
U_2(\psi)\pi(A)U_2(\psi)^{-1} &= 
U_2(\psi)\pi_2(\zeta_{\Mb_0}A)U_2(\psi)^{-1}= \pi_2(\Af_2(\psi)\circ\zeta_{\Mb_0}A) \\
&=
\pi_2(\zeta_{\Mb_0}\circ\Af_1(\psi)A) = \pi(\Af_1(\psi)A).
\end{align*}
Therefore $\omega$ is not only Poincar\'e-invariant but also obeys the spectrum condition,
i.e., the momentum operators $P_a$ corresponding to the unitary
representation of the translations have joint spectrum
in the forward lightcone.
If $\omega$ has a distributional $2$-point function, one may show
$P_a P^a \pi(\Phi_1(f))\Omega_2=\pi(-\Phi_1(\Box f))\Omega_2=m_1^2\pi(\Phi_1(f))\Omega_2$
for all $f\in\CoinX{\Mb_0}$ (cf.~e.g., the proof of \cite[Prop. 5.6]{Fewster:gauge}). Using the Reeh--Schlieder property of $\Omega_2$, we see that
$P_a P^a$ has an eigenvalue $m_1^2$. But $P_a P^a$ is the mass-squared operator for the
vacuum representation of $\Af_2(\Mb_0)$ and so has discrete spectrum $\{0,m_2^2\}$, a
contradiction.  Accordingly, there is no equivalence $\zeta$ between $\Af_1$ and $\Af_2$ so that $\zeta_{\Mb_0}^*\omega_2$
has distributional $2$-point function. 

In a similar way, but considering e.g., de Sitter spacetime instead of Minkowski space, one can rule out 
the possibility of (sufficiently regular) equivalences between Klein--Gordon theories with differing curvature couplings.

\subsection{Action on fields}

As shown in Section~\ref{sect:fields}, the locally covariant (linear) fields of a theory of a given type
form an abstract algebras (resp., vector spaces) $\Fld(\Df,\Af)$ (resp., $\Fldlin(\Df,\Af)$). 
The gauge group acts on these algebras/spaces in a natural fashion: Given any $\eta\in G$, and $\Phi\in\Fld(\Df,\Af)$, define the transformed field $\eta\cdot\Phi\in\Fld(\Df,\Af)$ by
\begin{equation}\label{eq:Gaction}
(\eta\cdot\Phi)_\Mb(f) = \eta_\Mb\Phi_\Mb(f)\qquad (f\in\CoinX{\Mb},~\Mb\in\Loc),
\end{equation}
which clearly obeys the naturality condition 
\begin{align}
\Af(\psi)(\eta\cdot\Phi)_\Mb(f) &= \Af(\psi)\circ \eta_\Mb(\Phi_\Mb(f)) = 
\eta_\Nb \circ \Af(\psi) \Phi_\Mb(f) = \eta_\Nb \Phi_\Nb(\psi_*f) \nonumber \\ 
&= (\eta\cdot\Phi)_\Nb(\psi_*f) 
\end{align}
for all $\psi:\Mb\to\Nb$, $f\in\CoinX{\Mb}$. Moreover, it is easily seen that
$\Phi\mapsto\eta\cdot\Phi$ is a $*$-automorphism of $\Fld(\Df,\Af)$, 
so we have defined a 
group homomorphism $G\mapsto \Aut(\Fld(\Df,\Af))$. Restricting to linear fields, \eqref{eq:Gaction} defines a representation of $G$ on 
$\Fldlin(\Df,\Af)$, which obeys $\eta\cdot\Phi^\star= (\eta\cdot\Phi)^\star$,
%\footnote{We compute: $(\eta\cdot\Phi^\star)_\Mb(f) = \eta_\Mb\Phi^\star_\Mb(f) = \eta_\Mb \Phi_\Mb(\overline{f})^*=  (\eta_\Mb\Phi_\Mb(\overline{f}))^*=(\eta\cdot\Phi)^\star_\Mb(f)$.} 
and is thus
a real linear representation. These representations are continuous with respect to a natural topology on $\Aut(\Af)$. 

In either case, we may define a {\em multiplet of fields} as any subspace of
$\Fld(\Df,\Af)$ (or \break $\Fldlin(\Df, \Af)$) transforming under an indecomposable representation of $G$.
Every field can then be associated with an equivalence class of $G$-representations.
Let $\rho, \sigma$ be the equivalence classes corresponding to fields
$\Phi$, $\Psi$. Then $\Phi^*$ transforms in the
complex conjugate representation $\bar{\rho}$ to $\rho$, while 
any linear combination of $\Phi$ and
$\Psi$ transforms in a subrepresentation of a quotient of $\rho\oplus\sigma$.
Here, the quotient allows for algebraic relationships; for example, if 
$\Phi$ and $\Psi$ belong to a common multiplet, then their linear combinations belong to the same multiplet. Similarly, $\Phi\Psi$ and $\Psi\Phi$ transform in (possibly different)
subrepresentations of quotients of $\rho\otimes\sigma$.  

For example, consider a locally covariant theory $\Af^{(3)}$ consisting of three independent massive scalar fields of common mass $m>0$ (and, for simplicity, minimal coupling), which has
an $\text{O}(3)$ of automorphisms described in \eqref{eq:On_action}. The scalar fields $\Phi^{(j)}$ ($j=1,2,3$)
span a $3$-dimensional multiplet associated with the defining representation $\sigma$ of $\text{O}(3)$, while the nonlinear fields $\Psi^{(S)}$  defined by 
\[
\Psi_\Mb^{(S)}(f) =  S_{ij}\Phi_\Mb^{(i)}(f)\Phi_\Mb^{(j)}(f) ,
\]
where $S$ is a complex symmetric $3\times 3$ matrix, 
span a $6$-dimensional subspace of $\Fld(\Df,\Af^{(3)})$ 
(carrying a subrepresentation of $\sigma\otimes\sigma$) and decomposes into a $1$-dimensional multiplet spanned by $\Psi^{(I)}$, where $I$ is the identity matrix, and a $5$-dimensional multiplet spanned by the $\Psi^{(S)}$, where $S$ is symmetric and traceless.

\subsection{Universal formulation of the free scalar field}\label{sec:KGreform}

Our treatment of the scalar field so far has followed
the traditional route of constructing algebras in each individual spacetime
and then specifying suitable morphisms between them in order to obtain a functor. 
One might characterise this as a bottom-up approach. 
We now describe an alternative top-down description, in which 
one specifies the theory directly at the functorial level.

First, observe that the Klein--Gordon operators
$P_\Mb\phi:=(\Box_\Mb+m^2+\xi R_\Mb)\phi=0$ 
form the components of a natural transformation
$P:\Df\nto\Df$, where $\Df:\Loc\to\Vect$ is 
as in \eqref{eq:Df_def} (but viewed as a functor to $\Vect$). 
This follows directly from the 
fact that $\psi$ is an isometry, and consequently $P_\Nb\psi_*f =
\psi_*P_\Mb f$ for all $f\in\Df(\Mb)$.
We may also define a bilocal natural scalar 
$E:\Df^{(2)}\nto \underline{\CC}$, whose component
in each $\Mb$ is precisely the advanced-minus-retarded bidistribution $E_\Mb$ for $P_\Mb$.
Here, $\underline{\CC}:\Loc\to\Vect$ is the constant functor
giving $\CC$ on all objects and $\id_\CC$ on all morphisms.
We also define $\II^{(\Af)}:\underline{\CC}\nto \Vf\circ\Af$ by
$\II^{(\Af)}_\Mb(z) =z\II_{\Af(\Mb)}$ ($z\in\CC$), where
$\Vf:\Alg\to\Vect$ is the forgetful functor.

Given these definitions, the Klein--Gordon theory may be
given a universal form:
\begin{definition}\label{def:univ_KG}
A \emph{Klein--Gordon theory} with field equation $P$ is a pair $(\Af,\Phi)$,
where $\Af:\Loc\to\Alg$ is a functor and $\Phi\in\Fldlin(\Df,\Af)$
is a linear field such that
\begin{itemize}
\item $\Phi^\star= \Phi$, i.e., $\Phi$ is hermitian
\item $\Phi\circ P = 0$, the zero field
\item $\Phi\stackrel{\rightarrow}{\otimes}\Phi -
\Phi\stackrel{\leftarrow}{\otimes}\Phi = i \II^{(\Af)}\circ E$ (see \eqref{eq:bilocal}),
\end{itemize} 
and which is universal in the sense that, if $(\Bf,\Psi)$ is any other pair with these properties,
then there is a unique subtheory embedding $\eta:\Af\nto\Bf$ such that
$\Psi = \eta\cdot\Phi$.
\end{definition}
This definition specifies Klein--Gordon theories up to equivalence,\footnote{Suppose $(\Af,\Phi)$ and $(\Bf,\Psi)$ both 
satisfy Definition~\ref{def:univ_KG}. Then there are naturals
$\eta:\Af\nto\Bf$ and $\zeta:\Bf\nto\Af$ such that
$\Psi=\eta\cdot\Phi$ and $\Phi=\zeta\cdot\Psi$. Hence
also $\Phi=(\zeta\circ\eta)\cdot \Phi$. But by the universal
property yet again, the only natural $\xi:\Af\nto\Af$ such
that $\Phi=\xi\cdot\Phi$ is the identity, $\xi_\Mb=\id_{\Af(\Mb)}$ for
all $\Mb\in\Loc$. Hence $\zeta\circ\eta=\id_{\Af}$ and by similar
reasoning applied to $(\Bf,\Psi)$, we also have $\eta\circ\zeta=
\id_{\Bf}$. Hence $\eta$ is an equivalence.}
so it is reasonable to speak of \emph{the} Klein--Gordon theory.
The original construction of the theory is needed to show that
the theory exists, but beyond that, it ought to be possible to work
with Definition~\ref{def:univ_KG} alone. 
Other models of locally covariant QFT can be given similar universal formulations. 

We now prove that $(\Af,\Phi)$ has the universal property, where $\Af$ is our standard Klein--Gordon functor and $\Phi$ its standard associated locally covariant field.
Suppose $(\Bf,\Psi)$ satisfies the other axioms. 
For each $\Mb\in\Loc$, we define a unital $*$-homomorphism
$\eta_\Mb:\Af(\Mb)\to\Bf(\Mb)$ by $\eta_\Mb\Phi_\Mb(f) = \Psi_\Mb(f)$ ($f\in\CoinX{\Mb}$),
which is well-defined because
the $\Phi_\Mb(f)$ generate $\Af(\Mb)$ and because both
$\Phi_\Mb$ and $\Psi_\Mb$ obey the relations itemized in
Definition~\ref{def:univ_KG}. Furthermore, $\Af(\Mb)$ is simple, so $\eta_\Mb$
is either monic or the zero map, and the latter case
is excluded because units and zeros are distinct for objects of $\Alg$.
Thus $\eta_\Mb:\Af(\Mb)\to\Bf(\Mb)$ is well-defined as an $\Alg$-morphism. Suppose that $\psi:\Mb\to\Nb$, then
\begin{align}
\eta_\Nb\Af(\psi) \Phi_\Mb(f) &= \eta_\Nb\Phi_\Nb(\psi_* f) 
= \Psi_\Nb(\psi_* f)  = \Bf(\psi) \Psi_\Mb(f) \notag \\ &=
\Bf(\psi)\eta_\Mb\Phi_\Mb(f)
\end{align}
for all $f\in\CoinX{\Mb}$. As the $\Phi_\Mb(f)$ generate
$\Af(\Mb)$, it follows that the components $\eta_\Mb$
cohere to form a natural transformation $\eta:\Af\nto\Bf$.
Uniqueness is clear from the foregoing argument, because
$\eta$ was fixed completely by requiring $\eta\cdot\Phi=\Psi$.

\section{Dynamical locality and SPASs}\label{sect:SPASs}

Our last topic brings us back to the fundamental
purpose of locally covariant QFT, namely, the description of 
common `physical content' in all (reasonable) spacetimes. 
The functorial definition of a theory certainly gives a common
mathematical definition across different spacetimes, but
under what circumstances can this be said to represent the
same physics?\footnote{In theories based on a classical Lagrangian 
one usually proceeds simply to use the `same' Lagrangian (modulo some
subtleties~\cite{FewsterRegensburg}) but this option is not
open in a general AQFT context.}
This question was addressed in~\cite{FewVer:dynloc_theory}
and will be briefly summarised here. 

We have already described how a pathological locally covariant theory $\Bf$ 
may be constructed from a basic theory $\Af$ (which can be as well-behaved
as one likes) -- see equation~\eqref{eq:oneortwo_obj} and~\eqref{eq:oneortwo_mor}.
In spacetimes with compact Cauchy surfaces the theory corresponds to
two copies of $\Af$, while in those with noncompact Cauchy surfaces we have
a single copy. Suppose we accept that $\Af$ represents the same physics
in all spacetimes (SPASs) according to some notion of what that might mean. 
Then $\Bf$ surely cannot also represent SPASs according to the same notion
as $\Af$ -- as their physical content coincides in some, but not all, spacetimes. 
 
This may be put into a more mathematical form as follows.  
For each $\Mb\in\Loc$, define $\zeta_\Mb:\Af(\Mb)\to\Bf(\Mb)$
and $\eta_\Mb:\Bf(\Mb)\to\Af(\Mb)\otimes\Af(\Mb)$ by  
\begin{equation}
\zeta_\Mb A = 
\begin{cases}
A &  \text{if $\Mb$ has noncompact Cauchy surfaces} \\
A\otimes\II & \text{if $\Mb$ has compact Cauchy surfaces}
\end{cases}
\end{equation}
\begin{equation} 
\eta_\Mb A = 
\begin{cases}
A & \text{if $\Mb$ has compact Cauchy surfaces} \\
A\otimes\II & \text{if $\Mb$ has noncompact Cauchy surfaces}.
\end{cases}
\end{equation} 
It is straightforward to check that these are well-defined morphisms and that,
for every $\psi:\Mb\to\Nb$, the naturality conditions $\zeta_\Nb\circ\Af(\psi) = 
\Bf(\psi)\circ\zeta_\Mb$ and $\eta_\Nb\circ\Bf(\psi) = 
(\Af\otimes\Af)(\psi)\circ\eta_\Mb$ hold.\footnote{Recall that if $\Mb$ has compact
Cauchy surfaces then so does $\Nb$.} Thus we have subtheory embeddings
\begin{equation}
\Af\xlongrightarrow[\zeta]{\cdot}\Bf \xlongrightarrow[\eta]{\cdot}\Af\otimes\Af,
\end{equation}
which are \emph{partial isomorphisms} -- meaning that there is at least one spacetime
for which the component of $\zeta$ is an isomorphism, and likewise for $\eta$.
However,  there are also spacetimes for which the corresponding component is not an
isomorphism, so neither $\zeta$ nor $\eta$ is an equivalence of theories.\footnote{This assumes
that $\Af$ is not isomorphic to $\Af\otimes\Af$. A convenient way of ruling out
such isomorphisms is to check that $\Af$ and $\Af\otimes\Af$ have nonisomorphic gauge groups.}
The situation just described gives a formal expression to the idea that $\Bf$ coincides
with $\Af$ in some spacetimes and with $\Af\otimes\Af$ in others. 

To summarize the discussion so far, let $\Tgth$ be a class of locally covariant quantum field theories.
We have argued that a necessary condition for $\Tgth$ to represent theories conforming to some particular notion of SPASs is that \emph{every partial isomorphism
between theories in $\Tgth$ is an isomorphism}. We refer to this necessary condition as the \emph{SPASs property}. Our example has shown immediately that the full class of locally covariant theories
does not have the SPASs property. 

The local structure of the pathological theory $\Bf$ is instructive. Suppose
$\Mb\in\Loc$ and choose $O\in\OO(\Mb)$ so that $\Mb|_O$ has noncompact
Cauchy surfaces. Then $\Bf^\kin(\Mb;O)\cong
\Bf(\Mb|_O)=\Af(\Mb|_O)$ whether or not $\Mb$ has compact Cauchy surfaces --
the kinematic local algebras only `sense' one copy of $\Af$. 
However,  another way to
sense the local degrees of freedom, based on dynamics, was introduced in \cite{FewVer:dynloc_theory}.

Let $\Cf$ be a locally covariant QFT and $K$ be a compact subset of $\Mb\in\Loc$. Define 
\[
\Cf^\bullet(\Mb;K) := \{C\in \Cf(\Mb): \rce_\Mb^{(\Cf)}[h]C = C~\text{for all $h\in H(\Mb)$ with
$\supp h\subset K^\perp$}\},
\]
where $K^\perp = \Mb\setminus J_\Mb(K)$ is the causal complement of $K$. 
Elements of $\Cf^\bullet(\Mb;K)$ are precisely those unaffected by any metric perturbation
supported in the causal complement of $K$.  In Sec.~\ref{sect:rce} we observed that
$\rce^{(\Cf)}_\Mb[h]$ acts trivially on any $\Cf^\kin(\Mb;O)$ with $O$ causally disjoint
from the support of $h$ -- here, we turn this around to give a new definition
of the local content of the theory. For each $O\in\OO(\Mb)$, we define 
\emph{dynamical algebra} $\Cf^\dyn(\Mb;O)$ to be
subalgebra of $\Cf(\Mb)$ generated by the $\Af^\bullet(\Mb;K)$
as $K$ ranges over all compact subsets of $O$ that have a multidiamond neighbourhood
with base in $O$ -- see  \cite[\S 5]{FewVer:dynloc_theory} for details.  

In the case of our pathological theory $\Bf$, we see that $\rce^{\Bf}_\Mb[h]= \rce^{\Af}_\Mb[h]$ if $\Mb$ has noncompact Cauchy surfaces and
$\rce^{\Bf}_\Mb[h]= \rce^{\Af}_\Mb[h]^{\otimes 2}$ if they are compact.  Correspondingly, 
we see that $\Bf^\dyn(\Mb;O)=\Af^\dyn(\Mb;O)$ in the noncompact case
and $\Bf^\dyn(\Mb;O)=\Af^\dyn(\Mb;O)^{\otimes 2}$ in the compact case. 
Thus the dynamical definition of locality senses degrees of freedom that are missed by the kinematical definition. This suggests focussing on theories of the 
following type:
\begin{definition}
A locally covariant QFT $\Cf$ is \emph{dynamically local} if 
\[
\Cf^\kin(\Mb;O)=\Cf^\dyn(\Mb;O)
\]
for all nonempty $O\in\OO(\Mb)$ and $\Mb\in\Loc$.
\end{definition}
Clearly the pathological theory $\Bf$ is not dynamically local.
More significantly: 
\begin{proposition}{\cite[Thm 6.10]{FewVer:dynloc_theory}}\label{prop:SPASs}
The class of dynamically local and locally covariant QFTs has the SPASs property. 
\end{proposition}
Thus dynamical locality at least satisfies our necessary condition for providing
a notion of SPASs (and there is no other condition known that does so
and incorporates the standard free theories). 

As an immediate application, we note the following. In Minkowski space
AQFT there are models with a minimal localization scale, i.e., the
local algebras are nontrivial only for sufficiently large regions (see, e.g., 
\cite{LecLon:2014} for simple examples). Proposition~\ref{prop:SPASs} excludes
the possibility that such models can be defined as locally covariant
and dynamically local theories. For if there is a spacetime $\Mb\in\Loc$
and a nonempty $O\in\OO(\Mb)$ for which $\Af^\kin(\Mb;O)$ 
is trivial, then $\Af(\Mb|_O)$ is trivial. Now there is a \emph{trivial
theory} $\If:\Loc\to\Alg$ so that $\If(\Nb)=\CC$ (regarded as a unital $*$-algebra) for all $\Nb\in\Loc$, and which maps every morphism to 
the identity morphism. This theory is a subtheory of $\Af$
in an obvious way, and we have shown that the two theories
coincide on $\Mb|_O$. Accordingly, if $\Af$ is dynamically local 
it is isomorphic to the trivial theory and hence trivial for all subregions
of all spacetimes.

Every dynamically local theory $\Af$ has a number of other nice properties: 
for instance, extended
locality of $\Af$ is equivalent to $\Af^\bullet(\Mb;\emptyset)=\CC\II_{\Af(\Mb)}$
for all $\Mb\in\Loc$, i.e., the absence of nontrivial elements in $\Af(\Mb)$ that
are fixed under arbitrary relative Cauchy evolution \cite[Thm 6.5]{FewVer:dynloc_theory} 
and $\Af$ is necessarily additive with respect to
truncated multidiamonds~ \cite[Thm 6.3]{FewVer:dynloc_theory}.
As a final application we return to one of our \emph{leit motifs}: the nonexistence of
natural states. 
\begin{theorem}{\cite[Thm 6.13]{FewVer:dynloc_theory}}\label{thm:nogo}
Suppose $\Af$ is a dynamically local quantum field theory and has a natural state $(\omega_\Mb)_{\Mb\in\Mand}$. If there is a spacetime $\Mb$ with noncompact Cauchy surfaces such that $\omega_\Mb$ induces a 
faithful GNS representation with the (full) Reeh--Schlieder property 
[i.e., the GNS vector corresponding to $\omega_\Mb$ is cyclic for the induced representation of 
$\Af(\Mb|_O)$ for all relatively compact $O\in\OO_0(\Mb)$],  then the relative Cauchy evolution is trivial in $\Mb$, and $\Af^\kin(\Mb;O)=\Af(\Mb)$ for all nonempty
$O\in\OO(\Mb)$. If, additionally, $\Af$ obeys extended locality, then $\Af$ is equivalent to the trivial theory $\If$. 
\end{theorem}
\begin{proof} \smartqed
By the natural state hypothesis, we have
$\omega_\Mb\circ\rce_{\Mb}[h] = \omega_\Mb$ for each $\Mb$ and all $h\in H(\Mb)$,
simply because the relative Cauchy evolution is a composition of (inverses of) morphisms $\Af(\psi)$ for Cauchy $\psi$.  Thus the relative Cauchy evolution is unitarily
implemented in the 
GNS representation $\pi_\Mb$ induced by $\omega_\Mb$:
\[
\pi_\Mb(\rce_\Mb[\hb] A) = U_\Mb[\hb]\pi_\Mb(A) U_\Mb[\hb]^{-1},
\]
where $U_\Mb[\hb]$ is defined by
$U_\Mb[\hb]\pi_\Mb(A)\Omega_\Mb = \pi_\Mb(\rce_\Mb[\hb]A)\Omega_\Mb$
and leaves the  GNS vector $\Omega_\Mb$ invariant.
Now let $h\in H(\Mb)$ and choose a nonempty relatively compact connected $O\in \OO(\Mb)$ such that $O\subset (\supp h)^\perp$ (such an $O$ exists because the Cauchy surfaces are noncompact). As already mentioned, $\rce_\Mb[h]$
acts trivially on $\Af^\kin(\Mb;O)$ \cite[Prop.~3.7]{FewVer:dynloc_theory},
so
\[
U_\Mb[\hb]\pi_\Mb(A)\Omega_\Mb =\pi_\Mb(A)\Omega_\Mb 
\]
for all $A\in\Af^\kin(\Mb;O)$. Using the Reeh--Schlieder property of $\omega_\Mb$ we may deduce that $U_\Mb[\hb]$ agrees with the identity operator on a dense set and hence $U_\Mb[\hb]=\II_{\HH_\Mb}$ for all $\hb\in H(\Mb)$,
so the relative Cauchy evolution is trivial on $\Af(\Mb)$ because $\pi_\Mb$
is faithful. 
Consequently, $\Af^{\bullet}(\Mb;K)=\Af(\Mb)$ for all compact sets $K$ and hence by dynamical locality $\Af^\kin(\Mb;O) = \Af^\dyn(\Mb;O)=\Af(\Mb)$ for each nonempty $O\in\OO(\Mb)$. This proves the first part of the theorem. 

For the second part, observe that there is a subtheory embedding
$\eta:\If\nto\Af$ of the trivial theory $\If$ into $\Af$ given by 
$\eta_\Nb z = z\II_{\Af(\Nb)}$ for all $\Nb\in\Loc$, $z\in\CC$.
Now consider two causally disjoint nonempty $O_1,O_2\in\OO(\Mb)$ 
By the above argument together with extended locality,
\[
\Af(\Mb) = \Af^\kin(\Mb;O_1)\cap\Af^\kin(\Mb;O_2) = \CC\II_{\Af(\Mb)},
\]
so $\eta_\Mb$ is an isomorphism. As $\If$ is obviously dynamically local, 
and $\Af$ is by assumption, 
Proposition~\ref{prop:SPASs} entails that $\eta$ is a natural isomorphism.
\end{proof} 

As mentioned, the property of dynamical locality has been checked for a
number of standard theories. Theories that satisfy dynamical locality include
\begin{itemize}
\item the Klein--Gordon scalar field in dimensions $n\ge 2$,
if at least one of the mass or curvature coupling is 
nonzero~\cite{FewVer:dynloc2,Ferguson:2013}, and 
the corresponding extended algebra of Wick polynomials for nonzero mass 
and either minimal or conformal coupling~\cite{Ferguson:2013} 
(one expects dynamical locality for general values of $\xi$);  
\item the free massless current in dimensions $n\ge 2$ (restricting to connected spacetimes) or $n\ge 3$ (allowing disconnected spacetimes)~\cite{FewVer:dynloc2};
\item the minimally coupled Klein--Gordon field with external sources
for $m\ge 0$, $n\ge 2$ -- in this case relative Cauchy evolution can
be induced by perturbations of both the metric and the external source
and one modifies the definition of the dynamical net accordingly~\cite{FewSchenkel:2014};
\item the free Dirac field with mass $m\ge 0$~\cite{Ferguson_PhD};
\item the free Maxwell field in dimension $n=4$, in a 
`reduced formulation' \cite{FewLang:2014a}.
\end{itemize}
These theories also obey the other hypotheses of Theorem~\ref{thm:nogo} and
so do not admit natural states. A more direct proof for the theory with sources 
appears in~\cite{Few_Chicheley:2015}.

There are some cases known in which dynamical locality fails,
which appears to be always related to the presence of broken global gauge
symmetries or topologically stabilised charges:
the free Klein--Gordon field with $m=0$, $\xi=0$ in dimensions $n\ge 2$,
owing to the rigid gauge symmetry $\phi\mapsto\phi+\text{const}$~\cite{FewVer:dynloc2}; the free massless current in $2$-dimensions allowing disconnected spacetimes~\cite{FewVer:dynloc2};
and the free Maxwell field in dimension $n=4$, in a 
`universal formulation' \cite{DappLang:2012,FewLang:2014a}, owing to the presence of  topological electric and magnetic
charges in spacetimes with nontrivial second de Rham cohomology, which 
are eliminated in the reduced theory mentioned above. As already mentioned in 
Sect.~\ref{sect:rigidity} the existence of topological
charges is also associated with a failure of injectivity (see \cite{SandDappHack:2012,BecSchSza:2014} for more discussion in related models). As suggested in~\cite{FewLang:2014a}, it would be interesting to 
investigate theories that are dynamically local modulo topological charges, with the aim
of generalizing Proposition~\ref{prop:SPASs}.

\medskip 
\noindent {\em Acknowledgment} We are grateful to Francis Wingham for
comments on the text. 

%
%\small{
%%\input{FewsterVerch_refs}
%\bibliographystyle{spmpsci_mod}
%\bibliography{covariant}}
%\end{document}

\end{document}